\documentclass[10pt,journal,compsoc]{IEEEtran}
\usepackage{algorithm,amsbsy,amsmath,amssymb,epsfig,bbm,mathrsfs,multirow,amsthm,xcolor}
\usepackage{epstopdf}
\usepackage{graphicx,caption,subcaption}%for three figures side-by-side
\usepackage{setspace}
\usepackage{algorithmic}
\usepackage{subcaption,float}
\usepackage[hidelinks]{hyperref}
\newtheorem{proposition}{Proposition} %for propositions

\newtheorem{definition}{Definition}
\newtheorem{lemma}{Lemma}
\newtheorem{theorem}{\textbf{\textsc{Theorem}}}
\newtheorem{remark}{\emph{Remark}}
\ifCLASSOPTIONcompsoc
\usepackage[nocompress]{cite}
\else
\usepackage{cite}
\fi
\ifCLASSINFOpdf
\else
\fi

\usepackage{setspace}

\begin{document}
		%
	% paper title
	% Titles are generally capitalized except for words such as a, an, and, as,
	% at, but, by, for, in, nor, of, on, or, the, to and up, which are usually
	% not capitalized unless they are the first or last word of the title.
	% Linebreaks \\ can be used within to get better formatting as desired.
	% Do not put math or special symbols in the title.
	\title{Federated Learning Framework with Straggling Mitigation and Privacy-Awareness for AI-based Mobile Application Services}
	\author{Yuris Mulya Saputra, Diep N. Nguyen, Dinh Thai Hoang, Quoc-Viet Pham, \\Eryk Dutkiewicz, and Won-Joo Hwang
		%School of Electrical and Data Engineering, University of Technology Sydney, Australia\\
		%\author{\IEEEauthorblockN{Yuris Mulya Saputra$^1$, Diep N. Nguyen$^1$, Dinh Thai Hoang$^1$, Eryk Dutkiewicz$^1$,\\
		%		Markus Dominik Mueck$^2$, and Srikathyayani Srikanteswara$^2$\\}
		%	$^1$ School of Electrical and Data Engineering, University of Technology Sydney, Australia\\
		%	$^2$ Intel Corporation\\ \vspace{-5mm}}	
		\IEEEcompsocitemizethanks{\IEEEcompsocthanksitem Y.~M.~Saputra is with the
			School of Electrical and Data Engineering, University of Technology Sydney, Sydney, NSW 2007, Australia (e-mail: yurismulya.saputra@student.uts.edu.au) and Department of Electrical Engineering and Informatics, Vocational College, Universitas Gadjah Mada, Yogyakarta 55281, Indonesia.
			\IEEEcompsocthanksitem D.~N.~Nguyen, D.~T.~Hoang, and E.~Dutkiewicz are with the
			School of Electrical and Data Engineering, University of Technology Sydney, Sydney, NSW 2007, Australia (e-mail: \{hoang.dinh, diep.nguyen, eryk.dutkiewicz\}@uts.edu.au).
			\IEEEcompsocthanksitem Quoc-Viet Pham is with the Korean Southeast Center for the 4th Industrial Revolution Leader Education, Pusan National University, Busan 46241, Republic of Korea (e-mail: vietpq@pusan.ac.kr).
			\IEEEcompsocthanksitem Won-Joo Hwang is with the Department of Biomedical Convergence Engineering, Pusan National University, Yangsan 50612, Republic of Korea (e-mail: wjhwang@pusan.ac.kr).
		
	}
	\thanks{}}
	
% make the title area
%\maketitle
%\vspace{-0.5cm}
% As a general rule, do not put math, special symbols or citations
% in the abstract or keywords.
\IEEEtitleabstractindextext{%
\begin{abstract}
		
		In this work, we propose a novel framework to address straggling and privacy issues for federated learning (FL)-based mobile application services, taking into account limited computing/communications resources at mobile users (MUs)/mobile application provider (MAP), privacy cost, the rationality and incentive competition among MUs in contributing data to the MAP. Particularly, the MAP first determines a set of the best MUs for the FL process based on the MUs' provided information/features. To mitigate straggling problems with privacy-awareness, each selected MU can then encrypt part of local data and upload the encrypted data to the MAP for an encrypted training process, in addition to the local training process. For that, each selected MU can propose a contract to the MAP according to its expected trainable local data and privacy-protected encrypted data. To find the optimal contracts that can maximize utilities of the MAP and all the participating MUs while maintaining high learning quality of the whole system, we first develop a multi-principal one-agent contract-based problem leveraging FL-based multiple utility functions. These utility functions account for the MUs' privacy cost, the MAP's limited computing resources, and asymmetric information between the MAP and MUs. Then, we transform the problem into an equivalent low-complexity problem and develop a light-weight iterative algorithm to effectively find the optimal solutions. Experiments with a real-world dataset show that our framework can speed up training time up to 49\% and improve prediction accuracy up to 4.6 times while enhancing the network's social welfare, i.e., total utility of all participating entities, up to 114\% under the privacy cost consideration compared with those of baseline methods.
		
		%This work proposes a novel framework to address straggling and privacy issues for federated learning (FL)-based mobile application services, considering limited computing/communications resources at mobile users (MUs)/mobile application provider (MAP), privacy cost, the rationality and incentive competition among MUs in contributing data to the MAP. Particularly, the MAP first determines a set of the best MUs for the FL process based on MUs' provided information/features. Then, each selected MU can encrypt part of local data and upload the encrypted data to the MAP for an encrypted training process, in addition to the local training process. For that, the selected MU can propose a contract to the MAP according to its expected local and encrypted data. To find optimal contracts that can maximize utilities while maintaining high learning quality of the system, we develop a multi-principal one-agent contract-based problem considering the MUs' privacy cost, the MAP's limited computing resource, and asymmetric information between the MAP and MUs. Experiments with a real-world dataset show that our framework can speed up training time up to 49% and improve prediction accuracy up to 4.6 times while enhancing network's social welfare up to 114% under the privacy cost consideration compared with those of baseline methods.

\end{abstract}

% Note that keywords are not normally used for peerreview papers.
}

% make the title area
\maketitle

\IEEEdisplaynontitleabstractindextext

\IEEEpeerreviewmaketitle

%==========================================================================================
%==========================================================================================
\IEEEraisesectionheading{\section{Introduction}\label{sec:Int}}	
	
	\IEEEPARstart{T}he ever-growing big data market between mobile users (MUs) and mobile application providers (MAPs) for emerging artificial intelligence (AI)-based mobile application services (e.g., healthcare, crowdsensing, and mobile social network applications)~\cite{Sarker:2021,Allied:2021} has recently attracted paramount interest from both industry and academia. Through utilizing the collected local data from many mobile devices, e.g., via embedded sensors, the AI-based mobile application services can extract the meaningful information leveraging machine learning (ML) approaches, e.g., centralized learning at the cloud server and local learning at mobile devices~\cite{Zhang2:2018,Sun:2019}. %For example, in healthcare services, the health information of the people can be accessed easily by recognizing their diverse daily activities (referred to as activity recognition) leveraging centralized learning at the MAP~\cite{Ordez:2016,Uddin:2019} and local learning at a specific device~\cite{Jeon:2017,Rokni:2018}. 
	Nonetheless, there exist two inherent challenges when applying such conventional ML approaches for AI-based mobile application services. First, MUs may not want to share their raw data due to the privacy risk of storing/processing data at the centralized servers. Second, MUs with limited computing/communications resources as well as unreliable wireless channels can be stragglers, adversely impacting the learning/training quality/time.
	
	Federated learning (FL) has been considered as a highly-effective learning approach to deal with the aforementioned problems. Like distributed learning, FL facilitates collaborative learning among multiple entities, e.g., MUs, without requiring them to share their raw data~\cite{Yang:2019,Lim:2020}. In particular, each participating MU can first train its local data to produce a local trained model independently. Then, the MAP can collect the local models from all participating MUs to update the global model iteratively. Nonetheless, similar to the centralized learning above, MUs are usually limited in computing and communications resources. For that, the quality of the FL process may deteriorate when some of the MUs experience low computing resources to train their local dataset and/or unstable wireless communication links when sharing the local trained models to the MAP at each learning round (referred to as \emph{straggling} MUs). Consequently, FL may likely suffer from significant delay because the MAP needs to wait for local trained models from all the participating MUs at each learning round to update the global model. To address the straggling problem, data from straggling MUs can be uploaded to an edge/a cloud server for the FL process (on behalf of the straggling MUs) as proposed in~\cite{Wang8:2019,Yoshida:2020,Mills:2020,Khan2:2020}. However, this approach cannot be widely adopted as MUs usually do not want to send the raw data to the edge/cloud server, e.g., the MAP, for processing due to privacy concerns. 
	%sharing local data from a small number of MUs to an edge/a cloud server for the FL process (on behalf of the straggling MUs) can be applied~\cite{Wang8:2019,Yoshida:2020,Mills:2020,Khan2:2020}. Although this FL approach can reduce the straggling problem for the participating MUs while maintaining a high-accurate global model accuracy, this approach may reveal the private information of the data-sharing MUs.
	
	To mitigate the straggling problems while preserving data privacy for the participating MUs under their limited computing resources, parts of local datasets from the participating MUs can be first encrypted to preserve their dataset privacy and then sent to the MAP for the additional encrypted training process. For example, the authors in~\cite{Marcano:2019,Yue:2021} first use homomorphic encryption (i.e., an encryption algorithm which allows operations to be executed directly on encrypted data/ciphertext without decryption) to protect the privacy of their local data and then send the encrypted data to the MAP for the encrypted training process, e.g., leveraging convolutional neural networks (CNN) and long-short term memory (LSTM). Nonetheless, both works only consider deep learning approaches to train the whole data at the MAP. In~\cite{Phong:2018,Fang:2021}, instead of encrypting the local datasets to share them with the MAP, MUs train their whole datasets locally then only encrypt their local models before sending them to the MAP. Note that~\cite{Marcano:2019}-\cite{Fang:2021} and other similar works in the literature do not account for the cryptography overhead/cost as well as limited computing and communications resources of the MAP and MUs. Moreover, they lack an incentive mechanism to motivate the MUs to participate in the FL processes, especially under the competition/conflict-of-interest among the MUs considering their diverse computing resources and data. To this end, how to optimally compensate all the participating MUs while maximizing the MAP's benefit under the privacy protection for the MUs and limited computing resources from the MAP and MUs remains as an open issue.

	In this paper, we propose a novel framework to address straggling and privacy issues for an FL-based mobile application service leveraging data encryption and additional encrypted training process at the MAP given the limited computing/communications resources at MUs/MAP, privacy cost, the rationality and incentive competition among MUs in contributing data to the MAP. In particular, the MAP can first select the best active MUs in its considered area (e.g., a shopping mall or a residential area) based on their shared information, i.e., selection metric values. This aims to obtain reliable trained model updates for the global model accuracy of the FL process with minimum straggling issues~\cite{Kang:2019}. After the MU selection process, the MAP can provide incentives/payments to encourage these MUs to contribute their data to the training process. In this way, the MAP can improve its learning quality to produce highly-accurate prediction model, thereby maximizing its utility correspondingly. Meanwhile, the participating MUs can maximize their utilities through acquiring certain incentives from the MAP in exchange for training the local datasets and/or sharing the encrypted datasets. 
	
	However, since there exists conflict-of-interest among the selected MUs due to their limited computing resources to train their local datasets, each participating MU may compete with other participating MUs in offloading encrypted dataset to the MAP for the additional training process under the MAP's limited computing resources (e.g., due to heavy workloads from other computationally-intensive services). 
	To cope with this problem, Stackelberg game models can be used to study the interactions between the MUs and the MAP~\cite{Yang:2016,Feng:2019,Kim:2020,Ding:2020}. Nonetheless, these game models are only applicable when the MUs and the MAP have full information about each other (referred to as \emph{information symmetry}), e.g., the MUs know the MAP's current computing resources.
	%these models are only effective when the MUs fully know the current available computing resource information of the MAP (referred to as \emph{information symmetry}). 
	In practice, the MAP usually keeps its computing resources as private information~(referred to as \emph{information asymmetry}). 
	Additionally, the MUs may not be interested in fully following the offers from the MAP due to the conflict of economic preferences between them. Hence, the above Stackelberg game models are not applicable to our considered problem.

	To address the aforementioned problem, we develop a \emph{multi-principal one-agent (MPOA)} contract-based framework~\cite{Bernheim:1986}, in which the participating MUs (as \emph{principals}) can first offer contract items, i.e., the sizes of local and uploaded encrypted datasets as well as their requested payments, to the MAP. Then, the MAP (as \emph{the agent}) can optimize the offered contracts on behalf of the participating MUs under its limited computing resources. Specifically, we formulate a privacy-preserving FL contract problem that can maximize the utilities of the MAP and the selected MUs for the entire FL process under the MAP's information-asymmetry and common constraints. These constraints guarantee that the MAP will always obtain a non-negative and maximum utility during the FL process. Unlike the contract problem for the conventional FL~\cite{Lim:2020}, in our framework, there exist two utility functions required to be addressed and optimized by each participating MU and the MAP, i.e., one for the encrypted (sharing) training data and the other for the unencrypted (local) training data.
	%In this case, each MU requires to optimize its utility for both encrypted data sharing and the local training process. Meanwhile, the MAP needs to optimize its utility for the encrypted training process and incentive management for the local training process at the MUs. As such, the contract-based framework with optimal encrypted data training at the MAP and optimal local data training at MUs can help to mitigate the straggling problems with privacy-awareness, and thus bring benefits for the MUs and high learning quality for the MAP, e.g., less training time and high accuracy level.
	
	%This contract problem can be formulated and addressed by the MAP since the participating MUs may not want to share their dataset size information with others due to the competition among them. 
	
	To obtain the optimal contracts, the FL contract optimization problem is first transformed into an equivalent problem with a lower complexity (due to constraint reduction). Then, an iterative algorithm is developed to quickly find the equilibrium solution, i.e., optimal contracts, for the transformed problem. This equilibrium solution can achieve the performance gap of less than 1\% compared with that of the information-symmetry FL contract, i.e., when each MU is fully aware of other MUs' available dataset sizes and the MAP's actual computing resources. Upon receiving the optimal contracts from the MAP, both the MAP and selected MUs can implement the proposed FL algorithm. %with straggling mitigation and privacy-awareness to boost the global model accuracy (even under many straggling scenarios) until the learning process converges. 
	We evaluate our FL-based framework using a real-world human activity recognition (HAR) dataset with 15M activity recognition samples from MUs' devices in 2019~\cite{WISDM:2019}. While considering the privacy costs, the experimental results show that our privacy-preserving FL-based framework can speed up the training time up to 49\% and improve the accuracy level up to 4.6 times compared with those of baseline methods, i.e., conventional FL methods without using the proposed framework. Furthermore, the proposed framework can enhance the social welfare of the mobile network up to 114\% compared with that of the baseline methods. The main contributions of our work are as follows:
	
	\begin{itemize}
		
		\item Propose a novel FL-based framework for the privacy-aware mobile application service that considers additional encrypted training process, privacy protection costs, the limited computing/communications resources at MUs and the MAP, and incentive competition among MUs in contributing data to the MAP. This framework is resilient to the straggling problems while preserving MUs' privacy. %under limited computing resources of the MAP and MUs.
		
		\item Introduce an MU selection scheme using the MUs' available dataset and device information. As such, we can choose the best MUs that can improve the overall privacy-preserving FL process with minimal straggling problems.
		
		\item Formulate the MPOA-based FL contract problem to address the encrypted training process competition among the selected MUs under the MAP's information asymmetry. Then, we transform the original FL contract problem into a low-complexity equivalent problem and design the iterative algorithm that can quickly find the equilibrium solution for the selected MUs. 
		
		\item Theoretically analyze the convergence of proposed FL with straggling mitigation and privacy-awareness.
		
		\item Perform extensive experiments with a real-world HAR dataset. The results can provide useful information to help the MAP in designing the stable FL process with privacy-awareness for AI-based mobile application services.
		
	\end{itemize}
	The rest of this paper is organized as follows. Section~\ref{sec:MU_scheme} introduces the proposed FL-based framework. Section~\ref{sec:MU_select} presents the proposed MU selection method. Then, the proposed MPOA contract-based problem and solution are described in Section~\ref{sec:PF} and Section~\ref{sec:PS}, respectively. Section~\ref{sec:FLS} explains the detailed learning process of the proposed FL. The performance evaluation is provided in Section~\ref{sec:PE}, and then the conclusion is summarized in Section~\ref{sec:Conc}.
	
	\section{Proposed Federated Learning Framework with Straggling Mitigation and Privacy-Awareness}
	\label{sec:MU_scheme}
	
	In this section, we first provide an overview of FL approach with straggling mitigation and privacy-awareness.
	Then, we present the proposed FL-based framework for the AI-based mobile application service.
	
	\subsection{Federated Learning with Straggling Mitigation and Privacy-Awareness Overview}
	\label{subsec:codfl}
	
	\begin{figure*}[t]
	\begin{center}
		$\begin{array}{cc} 
		\epsfxsize=2.4 in \epsffile{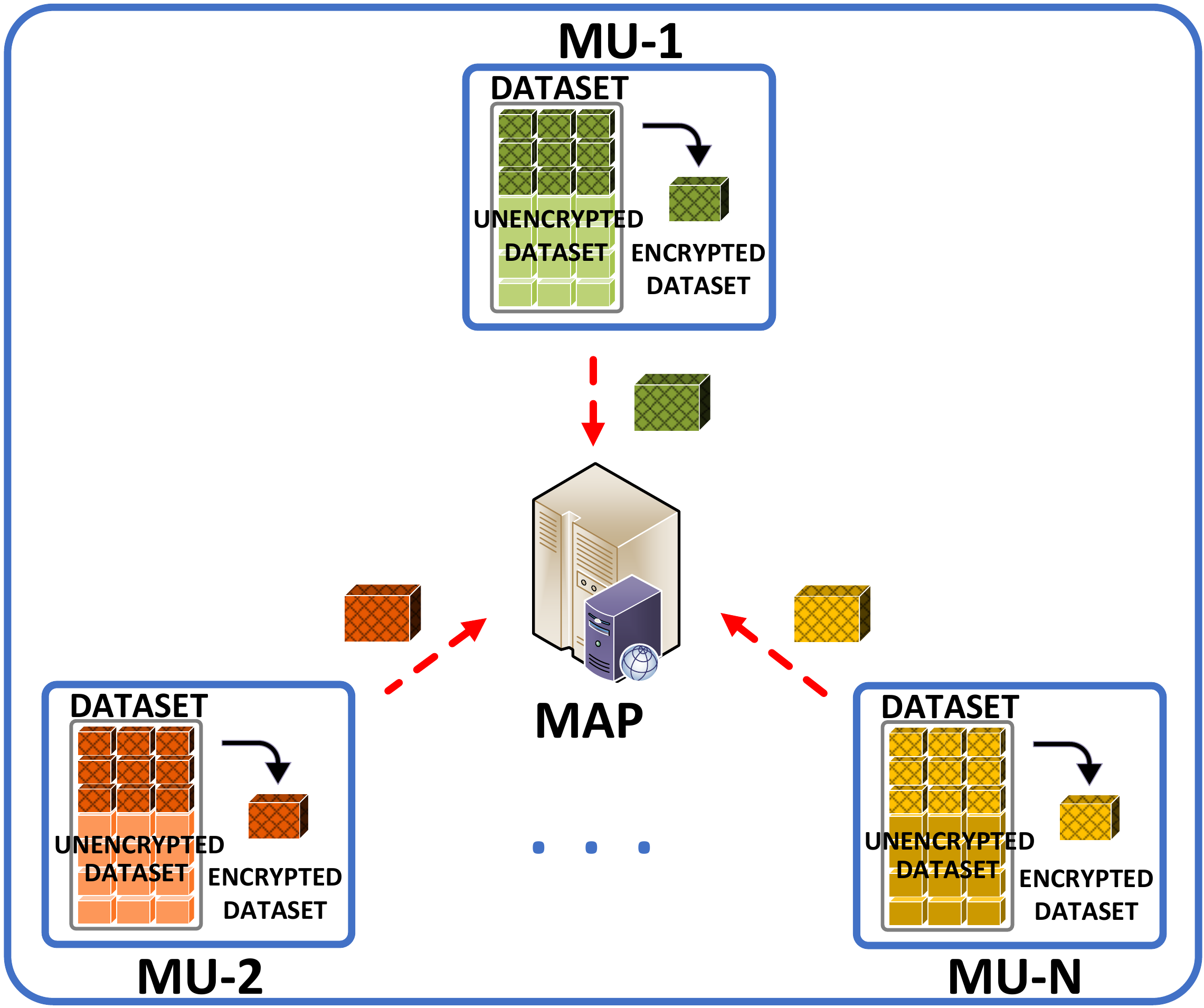} & 
		\hspace*{-.0cm}
		\epsfxsize=3.465 in \epsffile{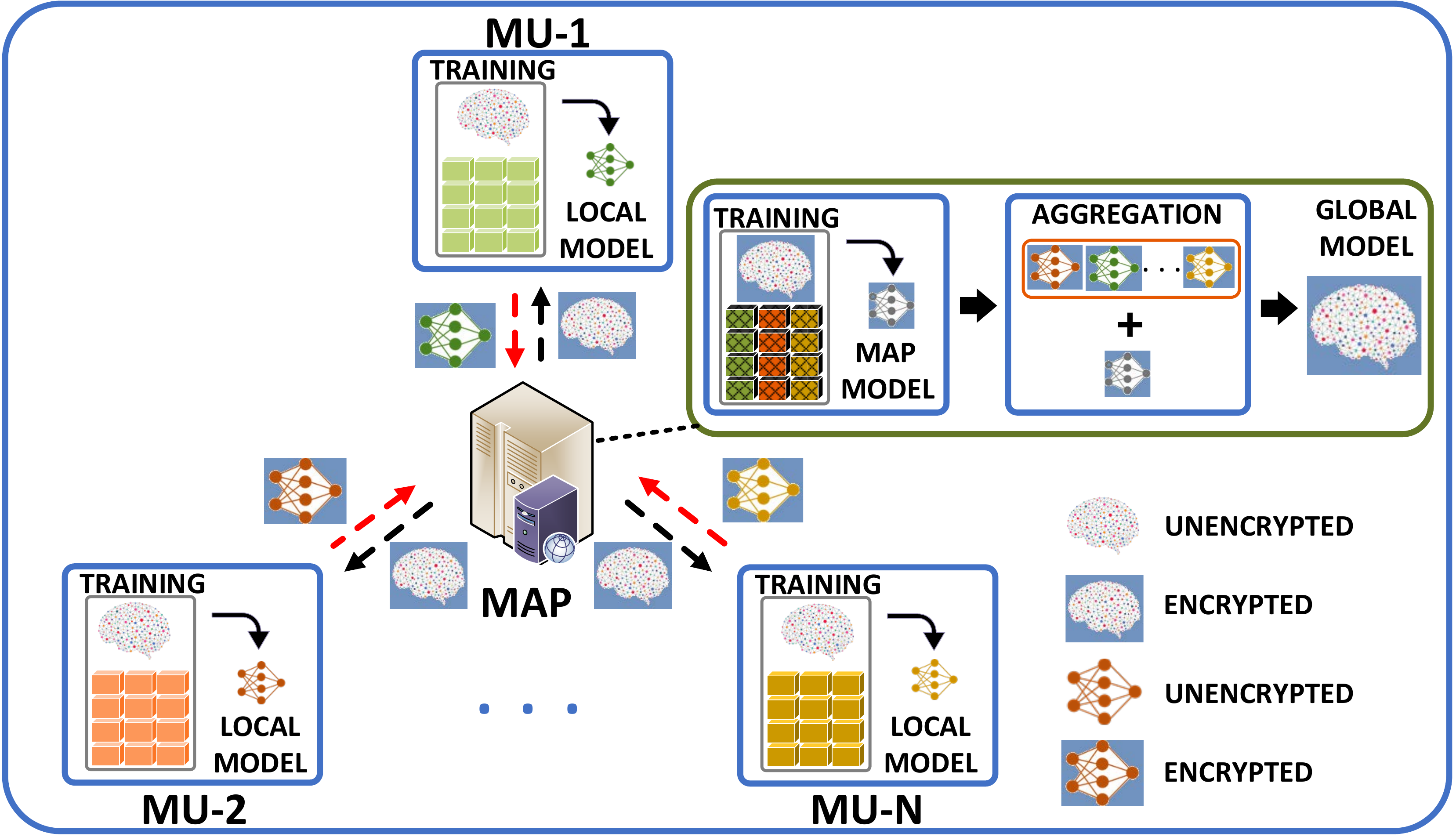} \\ [0.1cm]
		\text{\footnotesize (a) Pre-FL process} & \text{\footnotesize (b) FL process} \\ [0.2cm]
		\vspace*{-0cm}
		\end{array}$
		\vspace{-1.5\baselineskip}
		\caption{The FL with straggling mitigation and privacy-awareness with (a) one-time pre-FL and (b) iterative FL processes.}
		\vspace{-1.5em}
		\label{fig:SM2}
	\end{center}
	\end{figure*} 
	
	The process of FL with straggling mitigation and privacy-awareness is illustrated in Fig.~\ref{fig:SM2} that consists of the MAP (as a master node) and  participating MUs (as workers) $\mathcal{N} = \{1,\dots,n,\ldots,N\}$. Each MU-$n$ has its own sensing dataset $S_n = (\mathbf{F}_n, \mathbf{L}_n)$, where $\mathbf{F}_n$ and $\mathbf{L}_n$ are the training feature and training label matrices of the dataset at MU-$n$ with the size $D_n = |S_n|$, respectively. Here, $|S_n|$ is the number of samples at MU-$n$. Unlike the conventional FL, the proposed FL first incorporates a \emph{pre-FL process}, that is only executed once prior to implementing the FL process. The pre-FL aims to perform data encryption and offloading for an additional training process at the MAP, aiming at avoiding the straggling problems, addressing the limited computing resources at MUs, and maintaining the learning quality. Specifically, as shown in Fig.~\ref{fig:SM2}(a), each MU-$n$ first encrypts part of its sensing dataset to protect its privacy. Then, the encrypted datasets from all MUs can be collected by the MAP for the additional encrypted training process. In this case, the local/unencrypted dataset at each MU will be used for the local learning process as similarly implemented in the conventional FL approach. The mechanism on how to obtain optimal sizes of encrypted and unencrypted datasets is explained in Section~\ref{sec:PF}.
	
	Let ${D}^o_n$ and ${D}^l_n$ represent the sizes of encrypted and local (unencrypted) sensing datasets of MU-$n$, respectively, such that ${D}^o_n + {D}^l_n \leq {D}_n,\forall n \in \mathcal{N}$. This condition indicates that some MUs may not want to train all their datasets in practice. To encrypt a specific dataset, each MU-$n$ can use fully homomorphic encryption (FHE) that is widely used for data encryption in untrusted environments, e.g., third parties, without masking or dropping any features to protect the data privacy~\cite{Craig:2009}. To improve the performance of FHE, we adopt Brakerski/Fan-Vercauteren (BFV)-based FHE to support concurrent arbitrary operations on encrypted data~\cite{Fan:2012}. Specifically, the BFV-based FHE includes the following polynomial-time algorithms during the encryption and decryption processes.
	
	\begin{itemize}
		\item $\emph{\mbox{SKGen}}(n)$ to generate the secret key $g_n^{sk}$ for MU-$n$.
		
		\item $\emph{\mbox{PKGen}}(g_n^{sk})$ to generate the public key $g_n^{pk}$ as the function of $g_n^{sk}$ for MU-$n$.
		
		\item $\emph{\mbox{Enc}}(g_n^{pk},h)$ to encrypt data $h$ utilizing $g_n^{pk}$ with the output of encrypted data $h^*$.
		
		\item $\emph{\mbox{Dec}}(g_n^{sk},h^*)$ to decrypt data $h^*$ utilizing $g_n^{sk}$ with the output of original data $h$.
		
		\item $\emph{\mbox{Add}}(h^*_1,h^*_2), \emph{\mbox{Sub}}(h^*_1,h^*_2)$ and $\emph{\mbox{Mul}}(h^*_1,h^*_2)$ to respectively add, substract, and multiply encrypted data $h^*_1$ and $h^*_2$ with the output of $h^*_{add}, h^*_{sub}$, and $h^*_{mul}$.
		
		%		\item $Sub(h^*_1,h^*_2)$ to substract encrypted data $h^*_1$ and $h^*_2$ with the output of $h^*_{sub}$.
		%		
		%		\item $Mul(h^*_1,h^*_2)$ to multiply encrypted data $h^*_1$ and $h^*_2$ with the output of $h^*_{mul}$.
		
	\end{itemize}
	
	Based on the aforementioned algorithms, each MU-$n$ first can generate the secret key $g_n^{sk}$ using $\emph{\mbox{SKGen}}(n)$ and public key $g_n^{pk}$ using $\emph{\mbox{PKGen}}(g_n^{sk})$. The $g_n^{sk}$ of MU-$n$ remains private to other MUs and the MAP. Meanwhile, $g_n^{pk}$ of MU-$n$ is known by other MUs and the MAP. Then, each MU-$n$ can encrypt a part of the local dataset $S_n$, i.e., $S_n^\dagger$, utilizing its $g_n^{pk}$ to produce the encrypted dataset $S_n^o = (\mathbf{F}_n^o, \mathbf{L}_n^o)$, i.e., $\emph{\mbox{Enc}}(g_n^{pk},S_n^\dagger) = S_n^o$, where $\mathbf{F}_n^o$ and $\mathbf{L}_n^o$ are the encrypted training feature and encrypted training label matrices of the dataset at MU-$n$ with the size $D_n^o = |S_n^o|$, respectively. All the $S_n^o, \forall n \in \mathcal{N}$, then can be collected by the MAP for the accumulation process prior to the encrypted training execution. As such, the training feature and label matrices of encrypted dataset from all MUs in $\mathcal{N}$ can be respectively concatenated into 
	\begin{equation}
	\label{eqn3a1}
	\begin{aligned}
	\mathbf{F}^o &= 
	\begin{pmatrix}
	\mathbf{F}^o_1 \\
	\mathbf{F}^o_2 \\
	\vdots   \\
	\mathbf{F}^o_N
	\end{pmatrix} \text{, and }
	\mathbf{L}^o &= 
	\begin{pmatrix}
	\mathbf{L}^o_1 \\
	\mathbf{L}^o_2 \\
	\vdots   \\
	\mathbf{L}^o_N
	\end{pmatrix}.
	\end{aligned}
	\end{equation} 
	%	Alternatively, from (\ref{eqn:for2a}) and (\ref{eqn3a1}), we can obtain that $\mathbf{F}^o = \mathbf{G}\mathbf{H}\mathbf{F} \text{ and } \mathbf{L}^o = \mathbf{G}\mathbf{H}\mathbf{L}$, 
	%	%\begin{equation}
	%	%\label{eqn:for2b}
	%	%\begin{aligned}
	%	%\mathbf{F}^o = \mathbf{G}\mathbf{H}\mathbf{F} \text{, and } \mathbf{L}^o = \mathbf{G}\mathbf{H}\mathbf{L},
	%	%\end{aligned}
	%	%\end{equation}
	%	where $\mathbf{F}$ and $\mathbf{L}$ are the total training feature and total training label matrices of all workers in $\mathcal{N}$, respectively. Furthermore, $\mathbf{G}$ is the global generator matrix concatenated from all $\mathbf{G}_n, \forall n \in \mathcal{N}$, with dimension $\sum_{n \in \mathcal{N}} D_n^o \times \sum_{n \in \mathcal{N}} D_n$ and $\mathbf{H} = \text{diag}([\mathbf{h}_1,\ldots,\mathbf{h}_n,\ldots,\mathbf{h}_N])$ is the global weight matrix with dimension $\sum_{n \in \mathcal{N}} D_n \times \sum_{n \in \mathcal{N}} D_n$. %To this end, (\ref{eqn:for2b}) can be considered as the sensing data encoding for the total dataset of all workers in $\mathcal{N}$~\cite{Prakash:2021}.
	
	To minimize extra privacy leakage when sharing the encrypted datasets $S^o_n, \forall n \in \mathcal{N}$, with the MAP, e.g., metadata leakage of the encrypted datasets via their plaintext headers~\cite{Nikitin:2019}, a privacy protection level $\varepsilon_n$ ($\varepsilon_n \leq 1$)~\cite{Xu:2015} with respect to the encrypted dataset size of MU-$n$ can be deployed. Specifically, the higher the privacy protection level (i.e., higher $\epsilon_{n}$) is, the higher the privacy cost for the MU-$n$ is. As such, the privacy cost for the encrypted dataset generation and sharing of MU-$n$ can be written by~\cite{Prakash:2021}
	\begin{equation}
	\label{eqn:for2c}
	\begin{aligned}
	\xi_n = \frac{\beta}{2}\log_2\bigg(1 + \frac{\varepsilon_nD^o_{n}}{A_n^2}\bigg),
	\end{aligned}
	\end{equation}
	where $\beta$ is the unit cost of the privacy and $A_n$ is a function to quantify the influence of raw dataset distribution with respect to the number of features at MU-$n$ defined in~\cite{Prakash:2021}. Equation (\ref{eqn:for2c}) is later used for contract optimization in Section~\ref{sec:PF} (note that our framework can adopt any privacy cost metric for the contract optimization).
	
	Upon completing the pre-FL process, the FL process between the MAP and all MUs in $\mathcal{N}$ can be observed in Fig.~\ref{fig:SM2}(b). At each learning round, all MUs first train their local unencrypted datasets and then send the encrypted local trained models to the MAP (to preserve the privacy of the local trained models). Meanwhile, the MAP also trains all the collected encrypted datasets to produce another encrypted trained model. After completing the learning process within a training time threshold at each round (which is predefined by the MAP), the MAP can collect the encrypted local models from the MUs and aggregate them to obtain the aggregated encrypted local model. After that, the aggregated encrypted local model from MUs and encrypted model from the MAP can be combined together to update the encrypted global model, which is then used for the next FL process at the MUs and MAP. Note that the MUs are required to decrypt the encrypted global model before utilizing it to train the unencrypted local dataset~\cite{Phong:2018}. Moreover, for each learning round, the MAP only requires to aggregate any encrypted models that are received within a pre-defined training time threshold, due to additional training at the MAP. The above processes aim to guarantee the learning convergence while minimizing the straggling problems from the straggling MUs.
	
	\subsection{Proposed Federated Learning-Based Framework for the Mobile Application Service}
	
 	To effectively execute the above FL approach under the limited computing resources of the MAP as well as
 	participating MUs, privacy protection costs in implementing encryption mechanism, and the conflict-of-interest among the MUs, we then describe the proposed FL-based framework for the mobile application service. Suppose that there exists a set of active MUs in the network (as the candidate workers that utilize smart devices, e.g., smartphones and smartwatches) denoted by $\mathcal{M} = \{1,\dots,m,\ldots,M\}$. At a particular period, each active MU-$m$ through the mobile device can capture motion sensing information through its embedded sensor devices, e.g., accelerometer and gyroscope. This sensing information can be stored in a log file at the mobile device's internal storage. To participate in the FL process, the interested MUs can first install a human activity service application established by the MAP at their mobile devices' built-in Android or iOS platforms, e.g., Samsung Health (www.samsung.com/global/galaxy/apps/samsung-health) and Apple Health (www.apple.com/au/ios/health). Then, the MUs can securely share their selection metric values with respect to sizes of datasets, device configurations, and communication connections in the network, for the entire FL process. Note that such sharing information is only used to help the MAP choose the best MUs without disclosing private sensing data that are collected by the selected MUs when they perform the sensing process. In particular, the set of participating MUs in $\mathcal{N}$, where $\mathcal{N} \subset \mathcal{M}$, is choosen by the MAP based on the above accessable information. %(with some abusing of notation).
 	The selection of $N$ MUs ensures that they can provide useful sensing data for the FL process.
 	%Nevertheless, due to limited available computing resources, each MU may suffer from a long training time when its all sensing dataset is trained locally, leading to the straggling problem for the whole FL process. To address the problem, MUs can choose to train a certain subset of the available sensing dataset locally and offload the rest of the dataset to the MAP for the additional training process. However, uploading dataset directly to the MAP will break the main goal of FL approach. To protect the privacy of data that will be offloaded, we first need to pre-process them such that the data can be sent and trained at the MAP in a secure way. In this case, each MU can transform the remaining sensing dataset into the coded sensing dataset using a specific encoding scheme, e.g., the multiplication among a random generator matrix, weight matrix, and the sensing dataset~\cite{Prakash:2021}. After encoding the offloaded dataset, it can be sent to the MAP for the coded computing, aiming at maintaining the high-quality FL process with minimum privacy disclosure. 
 	%Then, each MU-$n$ can conduct data encryption and offloading processes as described in Section~\ref{subsec:codfl}.
 	
 	\begin{figure*}[!t]
 		\centering
 		\includegraphics[scale=0.4]{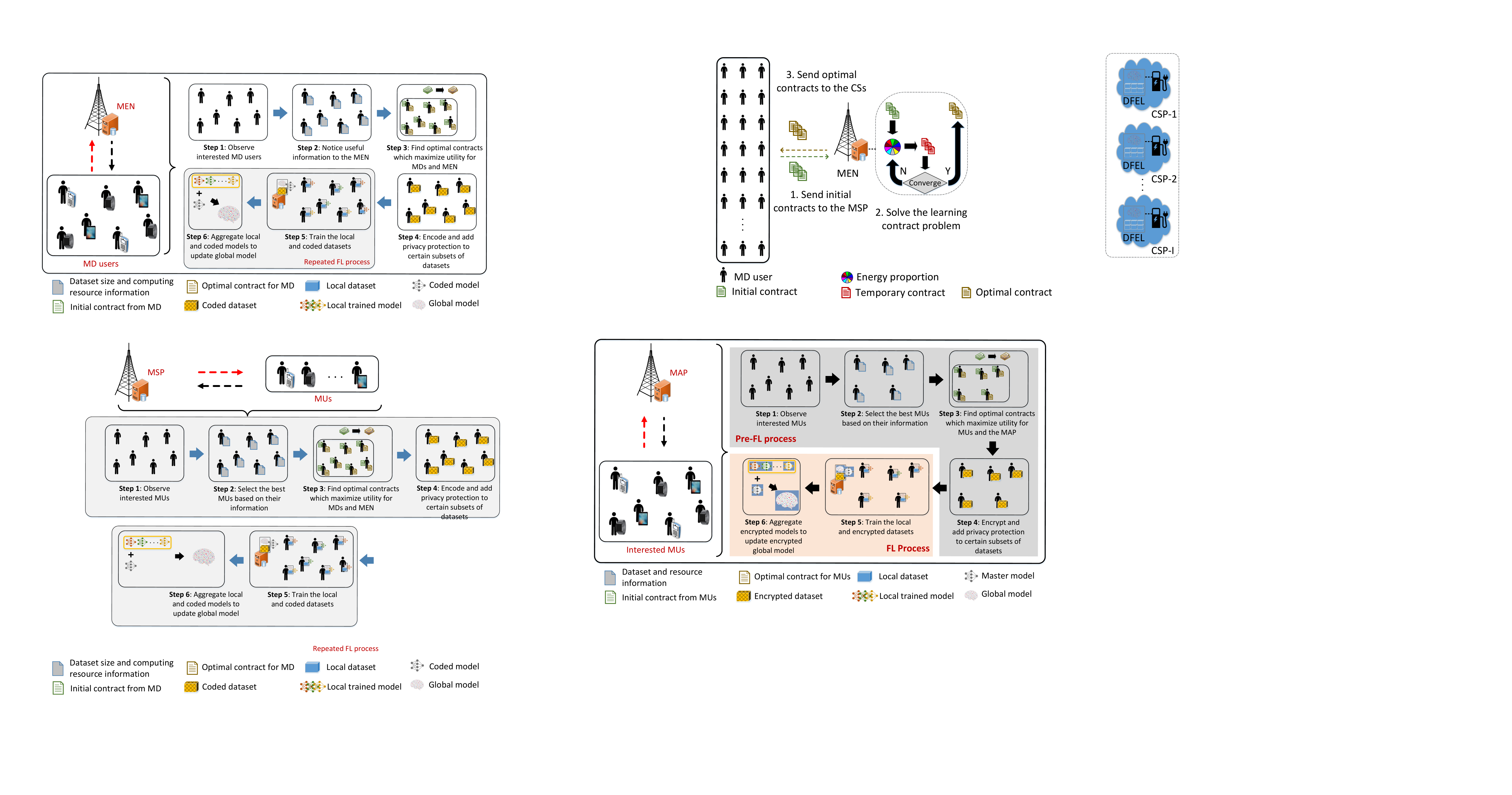}
 		\caption{The proposed FL-based framework for the mobile application service.}
 		\vspace{-0.5em}
 		\label{fig:SM}
 	\end{figure*}
 	
 	To reduce the straggling problem at MU-$n$, the size of local sensing dataset must be less than or equal to the maximum size of dataset that can be processed at the MU-$n$, i.e., ${\hat D}^l_n$, for the entire FL process, which is ${D}^l_n \leq {\hat D}^l_n$~\cite{Lim:2020}. For the collected encrypted dataset at the MAP, its size is constrained by maximum size of encrypted sensing samples $D^o_{max}$ that can be trained at the MAP (based on the MAP's computing resources). In practice, the MAP retains its willingness to train the collected encrypted datasets as an unknown information for the MUs due to its economic benefit. This willingness is defined as the type of the MAP~\cite{Bolton:2005} and influenced by the MAP's current computing resources. Particularly, a higher type implies the willingness to train more encrypted sensing datasets from the MUs due to its higher available computing resources. In other words, the willingness to train more encrypted datasets can compensate the low computing resources and straggling problems more, and thus bring more benefits for the MUs and better learning quality for the MAP. To this end, a finite set of the MAP's types can be specified as $\Pi = \{\pi_1,\ldots,\pi_i,\ldots,\pi_{I}\}$, and $\pi_1 < \pi_2 < \ldots < \pi_i < \ldots < \pi_{I-1} < \pi_I$,
 	%\begin{equation}
 	%\label{eqn:for1b}
 	%\begin{aligned}
 	%\pi_1 < \pi_2 < \ldots < \pi_i < \ldots < \pi_{I-1} < \pi_I,
 	%\end{aligned}
 	%\end{equation}
 	where $i \in \mathcal{I}$ with $\mathcal{I} = \{1,\ldots,i,\ldots,I\}$, represents the type index. 
 	
 	Despite that the MAP's type is private, the MUs can still observe the distribution of MAP's types, i.e, $\phi_i$, where $\sum_{i=1}^{I} \phi_i = 1, \forall i \in \mathcal{I}$~\cite{Xu:2015}, e.g., by monitoring public workloads of the MAP in the previous FL processes. For the MAP with type $\pi_i$, it has the maximum trainable encrypted dataset size ${\hat D}^o_i$, where ${\hat D}^o_i = \frac{\pi_i}{\pi_I}D^o_{max}, \forall i \in \mathcal{I}$. From ${\hat D}^o_i$, the MAP can allocate the encrypted dataset proportion for each MU-$n$ correspondingly. Specifically, for the MAP with type $\pi_i$, the encrypted dataset proportion vector of all MUs can be denoted by  $\boldsymbol{\eta}=[\boldsymbol{\eta}_1,\ldots,\boldsymbol{\eta}_i,\ldots,\boldsymbol{\eta}_I]$, where $\boldsymbol{\eta}_i=[\eta_i^1,\ldots,\eta_i^n,\ldots,\eta_i^{N}]$, and $0 \leq \eta_i^n \leq 1, \forall i \in \mathcal{I}, \forall n \in \mathcal{N}$. We denote the expected size of encrypted dataset and corresponding payment vectors of all participating MUs for all MAP's types as $\boldsymbol{D}^o=[\boldsymbol{D}^o_1,\ldots,\boldsymbol{D}^o_i,\ldots,\boldsymbol{D}^o_I]$ and $\boldsymbol{\rho}^o=[\boldsymbol{\rho}^o_1,\ldots,\boldsymbol{\rho}^o_i,\ldots,\boldsymbol{\rho}^o_I]$, respectively, where $\boldsymbol{D}^o_i = [D^o_{i,1},\ldots,D^o_{i,n},\ldots,D^o_{i,N}]$ and $\boldsymbol{\rho}^o_i = [\rho^o_{i,1},\ldots,\rho^o_{i,n},\ldots,\rho^o_{i,N}]$.  In this case, the actual size of encrypted dataset at MU-$n$ considering the MAP with type $\pi_i$ will be $\eta_i^n D^o_{i,n}$. Moreover, each MU-$n$ can offer a bigger size of encrypted sensing dataset to the MAP when the MAP's type increases because the MU-$n$ can obtain higher payment from the MAP. Furthermore, we denote $\boldsymbol{D}^l=[D^l_1,\ldots,D^l_n,\ldots,D^l_N]$ and $\boldsymbol{\rho}^l=[\rho^l_1,\ldots,\rho^l_n,\ldots,\rho^l_N]$ to be the size of local dataset and corresponding payment vectors of all MUs, respectively. Overall, as illustrated in Fig.~\ref{fig:SM}, the whole steps of proposed FL-based framework can be executed sequentially as follows:
 	%Specifically, the MAP first monitors participating MDs in the network. Then, the participating MDs inform their available datasets and current computing resources to the MAP via the application. Based on their available computing resources, the MDs send initial contracts containing the local and coded sensing datasets along with their corresponding payments for local training and coded sensing dataset transmission. After that, the MAP via its MEN will optimize the contracts and send them back to the MDs. Using the optimal contracts, all participating MDs train the optimal local sensing datasets to produce local trained models. Meanwhile, the MEN also trains the optimal coded sensing datasets to produce a coded trained model. Upon reaching the pre-defined deadline time of each learning round, these local models are then collected by the MEN and combined with the coded trained model to update the current global model (before sending it back to the MDs). The training process continues until the global prediction model converges or after a pre-defined number of learning rounds is reached.
 	\begin{itemize}
 		\item Step 1: The MAP broadcasts FL invitation to all MUs in the mobile network and observes the interested MUs therein. If an MU agrees to join the FL process, he/she can securely share his/her selection metric values to the MAP via the application.
 		
 		\item Step 2: The MAP chooses $N$ active MUs based on their securely-shared information. More details are described in Section~\ref{sec:MU_select}.
 		
 		\item Step 3: The selected MUs send initial contracts containing the sizes of required local training data and encrypted data along with their requested payments to the MAP for local and additional training processes, respectively. Then, the MAP will optimize the contracts and send them back to the MUs.
 		
 		\item Step 4: Using the optimal sizes of encrypted data, all the MUs encrypt and add the privacy protection to specific subsets of dataset prior to sending them to the MAP. 
 		
 		\item Step 5: All the MUs train the optimized local datasets to produce local trained models independently. Meanwhile, the MAP will use all the collected encrypted datasets to produce an encrypted trained model. 
 		
 		\item Step 6: After reaching the training time threshold at each learning round, their encrypted local models are then collected by the MAP and combined with the MAP's encrypted trained model to update the current encrypted global model (before sending it back to the MUs).
 	\end{itemize}
 	In this case, Steps 1 to 4 are only implemented once before the FL process starts. Differently, Steps 5 and 6 are repeated for each learning round until the global model converges or after a pre-defined learning rounds is reached. More details are presented in the following sections. 
	
	%	\subsection{System Model}
	%	\label{sec:SM}
	
	%MEN can sell the model to help other vulnerable users, e.g., older people tracking, disable people, crime monitoring, and driving behaviors.

	%\begin{figure*}[!t]
	%	\centering
	%	\includegraphics[scale=0.28]{Figs/FL_with_Coded_Computing}
	%	\caption{The general coded FL process.}
	%	%\vspace{-0.5em}
	%	\label{fig:SM2}
	%\end{figure*}

	\section{MU Selection with Available Dataset and Device Information}
	\label{sec:MU_select}
	
	To select $N$ participating MUs in the FL process, the MAP via the service application can first collect selection metric values from current active MUs in $\mathcal{M}$ based on their available sensing dataset quantity, general device information, and communication connections. Specifically, each MU-$m$ can first observe the quantity of available sensing dataset as the size of the dataset $D_m$. Since we consider the time-series sensing dataset in this work, the larger the size of dataset is (due to longer sensing period), the better the quality of dataset for the learning process is. For example, an MU that employs a sensing process for three hours usually obtains a higher number of training samples than those with less than one hour (considering that both MUs have the same sampling rate to do the sensing process continuously). As such, the MU with a larger dataset size can likely improve the learning process quality with a higher accuracy~\cite{Song:2014}.
	
	%In addition to the dataset quantity, t
	The active MUs also require to observe their general device configuration, e.g., computing resources, to execute the training process locally at the MUs efficiently. Here, the computing resources of each MU-$m$ can be denoted by $\kappa_m$.
	%\begin{equation}
	%\label{eqn:for0}
	%\begin{aligned}
	%\kappa_m = \kappa_{m,0} - \kappa_{m,used}.
	%\end{aligned}
	%\end{equation} 
	For the communication connections, each active MU-$m$ can observe the available transmission rate $R_m$ (based on its hardware configuration).
	
	From the observed $D_m$, $\kappa_m$, and $R_m$, the MUs via the application can first normalize the sets of $D_m$, $\kappa_m$, and $R_m$, $\forall m \in \mathcal{M}$, within $0-1$ range separately and obtain the normalized dataset size ${\tilde D}_m$, normalized computing resources ${\tilde \kappa}_m$, and normalized transmission rate ${\tilde R}_m$, $\forall m \in \mathcal{M}$. Then, the MUs via the application can calculate their selection metric values, i.e., $\chi_m = \iota_1{\tilde D}_m + \iota_2{\tilde \kappa}_m + \iota_3{\tilde R}_m$.
	%\begin{equation}
	%\label{eqn:for2}
	%\begin{aligned}
	%\chi_m = \iota_1{\tilde D}_m + \iota_2{\tilde \kappa}_m + \iota_3{\tilde R}_m.
	%\end{aligned}
	%\end{equation} 
	In this case, $\iota_1, \iota_2, \text{ and }\iota_3$ are the selection weight parameters determined by the MAP to control $\chi_m$, where $\iota_1 + \iota_2 + \iota_3 = 1$~\cite{Song:2014,Wang:2016}. These $\chi_m, \forall m \in \mathcal{M}$, are then shared by the MUs to the MAP. Using such selection metric values, the MAP can select $N$ MUs that will execute the FL process based on their largest $\chi_m$ values, i.e.,
	\begin{equation}
	\label{eqn:for4}
	\begin{aligned}
	\mathcal{N} = \max_{[N]}\{\chi_1,\ldots,\chi_M\}.
	\end{aligned}
	\end{equation}
	Equation (\ref{eqn:for4}) indicates that only MUs in $\mathcal{N}$ can participate in the FL process, i.e., train local datasets and then send local trained models to the MAP for the global model update.
	
	%\section{Sensing Data Coding and Privacy}
	%\label{sec:MU_coding}
	
	%\subsection{Sensing Data Coding}
	%
	%\subsection{Sensing Data Privacy}

	\section{MPOA-Based FL Contract Optimization Problem}
	\label{sec:PF}
	
	Upon completing the MU selection process, the MPOA-based FL contract optimization problem can be formulated. %(in which the MAP and participating MUs act as the agent and principals, respectively). %Unlike the contract problem for the conventional FL, in our contract optimization, there are multiple utility functions required to be addressed by both the MUs and the MAP with respect to the encrypted data sharing (in addition to the local training process) of the MUs and the encrypted training process (in addition to the encrypted trained model aggregation and update) at the MAP. This aims to maximize utilities of the MAP and participating MUs by finding each MU-$n$'s optimal sizes of encrypted and local/unencrypted datasets for the privacy-preserving FL process. 
	Specifically, %from the observation of interested MUs in the MEC network for the FL process, the MUs through the service application can first inform their total sizes of available sensing datasets and current computing resources to the MAP. Then, 
	the selected MUs can send initial contracts including the sizes of local datasets $D^l_n, \forall n \in \mathcal{N}$ (where $0 \leq D^l_n \leq D_n$), and requested payments $\rho^l_n = \alpha_l D^l_n, \forall n \in \mathcal{N}$, as well as the sizes of encrypted datasets $D^o_{i,n}, \forall n \in \mathcal{N}, \forall i \in \mathcal{I}$ (where $0 \leq D^o_{i,n} \leq D_n$), and requested payments $\rho^o_{i,n} = \alpha_o D^o_{i,n}, \forall n \in \mathcal{N}, \forall i \in \mathcal{I}$. These $\alpha_l$ and $\alpha_o$ specify unit prices of using a local sensing sample for local training at an MU and an encrypted sensing sample for training process at the MAP, respectively. Using these initial contracts, the MAP can help the MUs find the optimal contracts under the competition among the MUs in the proposed FL process. %In this section, we first describe the MPOA-based contract optimization problem for the VSP (as the agent) and learning SVs in $\mathcal{N}(t)$ (as the principals).  After that, a transformation method is introduced to reduce the complexity of the problem and an iterative algorithm is developed to find the equilibrium contract for the selected SVs.
	
	\subsection{Utility Optimization for the MAP}
	
	Accounting for the encrypted dataset proportion vector for the MAP with type $\pi_i$, i.e., $\boldsymbol{\eta}_i$, the size of offered local dataset, i.e., $\boldsymbol{D}^l$, and its payment vector, i.e., $\boldsymbol{\rho}^l$, the size of offered encrypted dataset for the MAP with type $\pi_i$, i.e., $\boldsymbol{D}^o_i$, and its payment vector for the MAP with type $\pi_i$, i.e., $\boldsymbol{\rho}^o_i$, the utility of the MAP with type $\pi_i$ (for executing the proposed FL process along with the MU-$n, \forall n \in \mathcal{N}$) can be expressed as follows:
	%\begin{equation}
	%\setlength{\belowdisplayskip}{5pt} 
	\begin{align}
	&{U}^i_{MAP} = \label{eqn:for30} \\
	&\underbrace{\pi_iG_o(\boldsymbol{\eta}_i,\boldsymbol{D}^o_i) - C_o(\boldsymbol{\eta}_i,\boldsymbol{\rho}^o_i,\boldsymbol{D}^o_i)}_{\text{\emph{Utility from training encrypted data}}} + \underbrace{G_l(\boldsymbol{D}^l) - C_l(\boldsymbol{\rho}^l)}_{\text{\emph{Utility from local training at MUs}}}. \nonumber
	\end{align}
	%\end{equation}
	Specifically, the utility of the MAP with type $\pi_i$ can be decomposed into two components. The first component is the difference between the gain function $G_o(\boldsymbol{\eta}_i,\boldsymbol{D}^o_i)$ and cost function $C_o(\boldsymbol{\eta}_i,\boldsymbol{\rho}^o_i,\boldsymbol{D}^o_i)$ for collecting and training the encrypted dataset at the MAP, where $\pi_i$ represents the weight of $G_o(\boldsymbol{\eta}_i,\boldsymbol{D}^o_i)$ for the MAP with type $\pi_i$. %In this case, the MAP with a higher type employs a higher weight due to the MAP's willingness to train more encrypted datasets in the FL process~\cite{Xu:2015}. 
	Meanwhile, the second component is the difference between the gain function $G_l(\boldsymbol{D}^l)$ and cost function $C_l(\boldsymbol{\rho}^l)$ for local training process at the MUs. 
	
	For both gain functions, we use a squared-root function as that of~\cite{Xu:2015} since in terms of economic perspective, the gain increases when a dataset with larger size is used in the FL process. However, the MAP may have less interest to further enhance the gain when much larger size of dataset incurs less global model accuracy improvement~\cite{Samuelson:2005}. %Note that the gain function can be defined using another function, e.g., natural logatihm function, which provides the same properties as the squared-root function. 
	To this end, the gain functions for training encrypted and local datasets can be respectively formulated by
	\begin{equation}
	\label{eqn:for30a}
	\begin{aligned}
	G_o(\boldsymbol{\eta}_i,\boldsymbol{D}^o_i) = \upsilon_o\sqrt{\sum_{n \in \mathcal{N}}\eta_i^n D^o_{i,n}} \text{, and } G_l(\boldsymbol{D}^l) = \upsilon_l\sqrt{\sum_{n \in \mathcal{N}}D^l_n},
	\end{aligned}
	\end{equation} 
	%and
	%\begin{equation}
	%\label{eqn:for30b}
	%\begin{aligned}
	%G_l(\boldsymbol{D}^l) = \upsilon_l\sqrt{\sum_{n \in \mathcal{N}}D^l_n},
	%\end{aligned}
	%\end{equation} 
	where $\upsilon_o > 0$ and $\upsilon_l > 0$ are the conversion parameters indicating the monetary unit of utilizing the encrypted and local datasets, respectively, according to the current data trading market~\cite{Xu:2015}. 
	The MAP's cost function $C_o(\boldsymbol{\eta}_i,\boldsymbol{\rho}^o_i,\boldsymbol{D}^o_i)$ is the sum of total payments for participating MUs in $\mathcal{N}$ (in regards to their encrypted dataset sharing) and energy consumption cost for training the combined encrypted dataset, i.e.,
	\begin{equation}
	\label{eqn:for30a2}
	\begin{aligned}
	C_o(\boldsymbol{\eta}_i,\boldsymbol{\rho}^o_i,\boldsymbol{D}^o_i) = \sum_{n \in \mathcal{N}}\eta_i^n\rho_{i,n}^o + \zeta c f^2 \sum_{n \in \mathcal{N}}\eta_i^n D^o_{i,n} ,
	\end{aligned}
	\end{equation} 
	where $\zeta$, $c$, and $f$ are the efficient capacitance constant of computing chipset, the number of CPU cycles to execute one sample of encrypted dataset, and utilized computing resources for the MAP, respectively. Moreover, the MAP's cost function $C_l(\boldsymbol{\rho}^l)$ is the total payments for the participating MUs with respect to the local training process at the MUs, i.e.,
	\begin{equation}
	\label{eqn:for30a3}
	\begin{aligned}
	C_l(\boldsymbol{\rho}^l) = \sum_{n \in \mathcal{N}}\rho_n^l.
	\end{aligned}
	\end{equation} 
	
	Based on (\ref{eqn:for30})-(\ref{eqn:for30a3}), the MAP's utility maximization problem for type $\pi_i$ can be written as:
	\begin{equation}
	\label{eqn:for3b}
	\begin{aligned}
	(\mathbf{P}_1) \phantom{10} & \max_{\boldsymbol{\eta}_i} \phantom{5} U^i_{MAP},
	\end{aligned}
	\end{equation}
	%\vspace{-0.7cm}
	\begin{align}
	&\text{ s.t. } \quad \sum_{n \in \mathcal{N}}\eta_i^n D^o_{i,n} \leq {\hat D}^o_i, \label{eqn:for3c}  \\
	&0 \leq \eta_i^n \leq 1, \forall n \in \mathcal{N}, \label{eqn:for3c2} 
	\end{align}
	where constraint (\ref{eqn:for3c}) ensures that the total encrypted dataset collected by the MAP cannot exceed the maximum trainable encrypted dataset at the MAP with type $\pi_i$, i.e., ${\hat D}^o_i$. As the objective function (\ref{eqn:for3b}) is a convex function, i.e., the gain and cost functions are respectively concave and linear, and the constraints~(\ref{eqn:for3c})-(\ref{eqn:for3c2}) in $(\mathbf{P}_1)$ are linear, the optimal $\boldsymbol{\hat \eta}_i,\forall i \in \mathcal{I}$, can be found using popular convex solvers.
	
	\subsection{Utility Optimization for Participating MUs}
	
	%In addition to the MAP's utility optimization, we formulate the utility optimization for the participating MUs, aiming at maximizing expected utility of each MU-$n$ independently under the constraints from the MAP and the MUs' limited computing resources. In particular, 
	Using the optimal encrypted dataset proportion $\boldsymbol{\hat \eta}_i, \forall i \in \mathcal{I}$, each MU-$n$ can derive its expected utility by considering all the MAP's possible types, which is
	\begin{align}
	&U_n(\boldsymbol{D}^o, \boldsymbol{\rho}^o,\boldsymbol{D}^l, \boldsymbol{\rho}^l) = 	\label{eqn:for6} \\ &\sum_{i=1}^{I}\Bigg(\underbrace{{\hat \eta}_i^n\rho^o_{i,n} - \xi_{i,n} - \gamma{\hat \eta}_i^n D^o_{i,n}}_{\text{\emph{Utility of training encrypted data at MAP}}} + \underbrace{\rho^l_n - \zeta_n c_n f_n^2 D^l_{n} }_{\text{\emph{Utility from local training}}} \Bigg) \phi_i. \nonumber
	\end{align}
	Similar to the MAP's utility, the utility of MU-$n$ can be divided into two components. The first component is the difference between the received payment ${\hat \eta}_i^n\rho^o_{i,n}$ for sharing the encrypted dataset and the costs for encrypted dataset privacy protection and transmission, where $\xi_{i,n} = \frac{\beta}{2}\log_2\bigg(1 + \frac{\varepsilon_n{\hat \eta}_i^n D^o_{i,n}}{A_n^2}\bigg)$ indicates the privacy cost for the encrypted dataset generation of MU-$n$ when the MAP has type $\pi_i$. Additionally, $\gamma{\hat \eta}_i^n D^o_{i,n}$ denotes the transmission cost of encrypted dataset, where $\gamma$ is the unit cost of sending an encrypted sample to the MAP. Meanwhile, the second component is the difference between the received payment and the cost of energy consumption $\zeta_n c_n f_n^2 D^l_{n}$ for training the local dataset, where $\zeta_n$, $c_n$, and $f_n$ are the efficient capacitance constant of computing chipset, the number of CPU cycles to execute one sample of local dataset, and utilized computing resources for the MU-$n$, respectively. In addition, the use of $\sum_{i=1}^{I}(.)\phi_i$ means that the expected utility of each MU-$n$ depends on the type distribution of the MAP.%, i.e., $\phi_i, \forall i \in \mathcal{I}$.
	
	In this utility optimization, we can find optimal contracts $(\boldsymbol{D}^o, \boldsymbol{\rho}^o,\boldsymbol{D}^l, \boldsymbol{\rho}^l)$ that satisfy the MAP's individual rationality (IR) and incentive compatibility (IC) constraints. Particularly, the IR constraints ensure that the MAP with all possible types will produce non-negative utilities as stated in Definition~\ref{DefV1}.
	\begin{definition}{}\label{DefV1}
		IR constraint: The MAP with type $\pi_i$ must achieve a positive utility, i.e.,
		\begin{align}
		&\pi_iG_o(\boldsymbol{\hat \eta}_i,\boldsymbol{D}^o_i) - C_o(\boldsymbol{\hat \eta}_i,\boldsymbol{\rho}^o_i,\boldsymbol{D}^o_i) + G_l(\boldsymbol{D}^l) - C_l(\boldsymbol{\rho}^l) \geq 0, \nonumber\\
		&\forall i \in \mathcal{I}, \label{eqn:for7a}
		\end{align}
		to join in the contract optimization.
	\end{definition}
	Using the IR constraints, the utility maximization problem for each MU-$n$ when all MUs have full information of the MAP with true type $\pi_i$ (information-symmetry contract) can be derived as follows:
	\begin{equation}
	\label{eqn:for90}
	\begin{aligned}
	(\mathbf{P}_2) \phantom{10} \underset{\boldsymbol{D}^o, \boldsymbol{\rho}^o,\boldsymbol{D}^l, \boldsymbol{\rho}^l}{\text{max}} \phantom{5} &{\hat \eta}_i^n\rho^o_{i,n} - \xi_{i,n} - \gamma{\hat \eta}_i^n D^o_{i,n} \\&+ \rho^l_n - \zeta_n c_n f_n^2 D^l_{n}, \forall n \in \mathcal{N},
	\end{aligned}
	\end{equation}
	%\vspace{-0.7cm}
	\begin{align}
	&\text{ s.t. } \quad \sum_{n \in \mathcal{N}}{\hat \eta}_i^n D^o_{i,n} \leq {\hat D}^o_i, \label{eqn:for90a} \\
	& D^l_n \leq {\hat D}^l_n, \forall n \in \mathcal{N}, \label{eqn:for90b} \\
	& {\hat \eta}_i^n D^o_{i,n} + D^l_n \leq D_n, \forall n \in \mathcal{N}, \label{eqn:for90c} \\
	&\pi_iG_o(\boldsymbol{\hat \eta}_i,\boldsymbol{D}^o_i) - C_o(\boldsymbol{\hat \eta}_i,\boldsymbol{\rho}^o_i,\boldsymbol{D}^o_i) + G_l(\boldsymbol{D}^l) - C_l(\boldsymbol{\rho}^l) \geq 0, \label{eqn:for90d}
	\end{align}
	where the constraint (\ref{eqn:for90b}) represents that the dataset which will be trained locally cannot exceed the maximum trainable local dataset at the MU-$n$. Then, the constraint (\ref{eqn:for90c}) specifies that the total trainable datasets of MU-$n$, i.e., encrypted and local datasets, must be less than or equal to the available data at MU-$n$. 
	
	Since there exists the information asymmetry between the MAP and MUs in practice, the optimal contracts must also satisfy the IC constraints to guarantee that the contracts are feasible. The IC constraint ensures that the MAP with a certain type can achieve the maximum utility when the appropriate contract originated for that type of the MAP is applied. This condition can be described in the following Definition~\ref{DefV2}. 
	
	\begin{definition}{}\label{DefV2}
		IC constraint: The MAP with current type $\pi_i$ will select a contract designed for its current type $\pi_i$ rather than with another type ${\pi_{i^*}}$ to maximize its utility, i.e.,
		\begin{align}
		&\pi_iG_o(\boldsymbol{\hat \eta}_i,\boldsymbol{D}^o_i) - C_o(\boldsymbol{\hat \eta}_i,\boldsymbol{\rho}^o_i,\boldsymbol{D}^o_i) + G_l(\boldsymbol{D}^l) - C_l(\boldsymbol{\rho}^l) \geq \nonumber \\
		&\pi_{i}G_o(\boldsymbol{\hat \eta}_{i^*},\boldsymbol{D}^o_{i^*}) - C_o(\boldsymbol{\hat \eta}_{i^*},\boldsymbol{\rho}^o_{i^*},\boldsymbol{D}^o_{i^*}) + G_l(\boldsymbol{D}^l) - C_l(\boldsymbol{\rho}^l), \nonumber \\
		& i \neq {i^*},\forall i, {i^*} \in \mathcal{I} \label{eqn:for7b}, \\
		&\pi_iG_o(\boldsymbol{\hat \eta}_i,\boldsymbol{D}^o_i) - C_o(\boldsymbol{\hat \eta}_i,\boldsymbol{\rho}^o_i,\boldsymbol{D}^o_i) \geq \nonumber \\
		&\pi_{i}G_o(\boldsymbol{\hat \eta}_{i^*},\boldsymbol{D}^o_{i^*}) - C_o(\boldsymbol{\hat \eta}_{i^*},\boldsymbol{\rho}^o_{i^*},\boldsymbol{D}^o_{i^*}), i \neq {i^*}, \forall i, {i^*} \in \mathcal{I}. \label{eqn:for7}
		\end{align}
	\end{definition}
	
	From the IR and IC constraint definitions, we can formulate the optimization problem for each MU-$n$ which can maximize its expected utility individually at the MAP under the competition with other MUs in $\mathcal{N}$ and the MAP's unknown limited computing resources as follows:
	\begin{equation}
	\label{eqn:for9}
	\begin{aligned}
	(\mathbf{P}_3) \phantom{10} & \underset{\boldsymbol{D}^o, \boldsymbol{\rho}^o,\boldsymbol{D}^l, \boldsymbol{\rho}^l}{\text{max}} \phantom{5} U_n(\boldsymbol{D}^o, \boldsymbol{\rho}^o,\boldsymbol{D}^l, \boldsymbol{\rho}^l), \forall n \in \mathcal{N},
	\end{aligned}
	\end{equation}
	%\vspace{-0.7cm}
	\begin{align}
	&\text{ s.t. } \quad \sum_{n \in \mathcal{N}}{\hat \eta}_i^n D^o_{i,n} \leq {\hat D}^o_i, \forall i \in \mathcal{I}, \label{eqn:for9a} \\
	& D^l_n \leq {\hat D}^l_n, \forall n \in \mathcal{N}\label{eqn:for9b}, \\
	& {\hat \eta}_i^n D^o_{i,n} + D^l_n \leq D_n, \forall i \in \mathcal{I}, \forall n \in \mathcal{N}, \label{eqn:for9c} \\
	&\pi_iG_o(\boldsymbol{\hat \eta}_i,\boldsymbol{D}^o_i) - C_o(\boldsymbol{\hat \eta}_i,\boldsymbol{\rho}^o_i,\boldsymbol{D}^o_i) + G_l(\boldsymbol{D}^l) - C_l(\boldsymbol{\rho}^l) \geq 0, \nonumber\\
	&\forall i \in \mathcal{I}, \label{eqn:for9d} \\
	&\pi_iG_o(\boldsymbol{\hat \eta}_i,\boldsymbol{D}^o_i) - C_o(\boldsymbol{\hat \eta}_i,\boldsymbol{\rho}^o_i,\boldsymbol{D}^o_i) \geq \nonumber \\
	&\pi_{i}G_o(\boldsymbol{\hat \eta}_{i^*},\boldsymbol{D}^o_{i^*}) - C_o(\boldsymbol{\hat \eta}_{i^*},\boldsymbol{\rho}^o_{i^*},\boldsymbol{D}^o_{i^*}), i \neq {i^*}, \forall i, {i^*} \in \mathcal{I}. \label{eqn:for9e}
	\end{align}
	
	%\emph{\color{red}Probably, we can add maximum number of SVs that can connect to a single VSP after constraints (12), using $\sum_{j=1}^J \eta^i_j \leq C_i$, where $C_i$ is maximum number of workers that can be processed by the VSP for certain period}.
	
	\subsection{Social Welfare of the Mobile Network}
	
	To devise the social welfare, i.e, the total actual utilities of the MAP and all participating MUs in the proposed FL process, we can first obtain the actual utility of each MU-$n$ when the MAP has type $\pi_i$ as ${\hat \eta}_i^n\rho^o_{i,n} - \frac{\beta}{2}\log_2\bigg(1 + \frac{\varepsilon_n{\hat \eta}_i^n D^o_{i,n}}{A_n^2}\bigg) - \gamma{\hat \eta}_i^n D^o_{i,n} + \rho^l_n - \zeta_n c_n f_n^2 D^l_{n}$.
	%\begin{equation}
	%\label{eqn:for6a}
	%\begin{aligned}
	%{\hat \varrho}_j^n(t)\varphi_j^n(t) - {\hat \varrho}_j^n(t)\zeta_j^n(t)\xi_n(t).
	%\end{aligned}
	%\end{equation}
	Given the MAP with type $\pi_i$, the total actual utility of all participating MUs can be derived by
	\begin{equation}
	\label{eqn:for6b}
	\begin{aligned}
	U^i_{MU_{tot}} = &\sum_{n \in \mathcal{N}}\Bigg({\hat \eta}_i^n\rho^o_{i,n} - \frac{\beta}{2}\log_2\bigg(1 + \frac{\varepsilon_n{\hat \eta}_i^n D^o_{i,n}}{A_n^2}\bigg) - \\
	&\gamma{\hat \eta}_i^n D^o_{i,n} + \rho^l_n - \zeta_n c_n f_n^2 D^l_{n}\Bigg).
	\end{aligned}
	\end{equation}
	Thus, the social welfare of the network when the MAP has type $\pi_i$ based on~(\ref{eqn:for30}) and (\ref{eqn:for6b}) can be computed by
	\begin{align}
	U^i_{SW} &= U^i_{MAP} + U^i_{MU_{tot}}.\label{eqn:for6c}
	%\Bigg(\pi_i\upsilon_o\sqrt{\sum_{n \in \mathcal{N}}{\hat \eta}_i^n D^o_{i,n}} - \sum_{n \in \mathcal{N}}{\hat \eta}_i^n\rho_{i,n}^o - \zeta c f^2 \sum_{n \in \mathcal{N}}{\hat \eta}_i^n D^o_{i,n}\Bigg) + \Bigg(\upsilon_l\sqrt{\sum_{n \in \mathcal{N}}D^l_n} - \sum_{n \in \mathcal{N}}\rho_n^l\Bigg) \nonumber \\ 
	%&+ \sum_{n \in \mathcal{N}}\Bigg({\hat \eta}_i^n\rho^o_{i,n} - \frac{\beta}{2}\log_2\bigg(1 + \frac{\varepsilon_n{\hat \eta}_i^n D^o_{i,n}}{A_n^2}\bigg) - \gamma{\hat \eta}_i^n D^o_{i,n} + \rho^l_n - \zeta_n c_n f_n^2 D^l_{n}\Bigg) \label{eqn:for6c} \\
	%&= \Bigg(\pi_i\upsilon_o\sqrt{\sum_{n \in \mathcal{N}}{\hat \eta}_i^n D^o_{i,n}} + \upsilon_l\sqrt{\sum_{n \in \mathcal{N}}D^l_n} - \zeta c f^2 \sum_{n \in \mathcal{N}}{\hat \eta}_i^n D^o_{i,n} - \sum_{n \in \mathcal{N}}\rho_n^l\Bigg) \nonumber \\
	%&+\sum_{n \in \mathcal{N}}\Bigg(\rho^l_n - \frac{\beta}{2}\log_2\bigg(1 + \frac{\varepsilon_n{\hat \eta}_i^n D^o_{i,n}}{A_n^2}\bigg) - \gamma{\hat \eta}_i^n D^o_{i,n} - \zeta_n c_n f_n^2 D^l_{n}\Bigg). \nonumber
	\end{align}
	
	\section{MPOA-Based FL Contract Optimization Solution}
	\label{sec:PS}
	
	\subsection{Contract Problem Transformation}
	
	The following Proposition~\ref{lemma0a} shows that the computational complexity of solving $(\mathbf{P}_3)$ increases quadratically when the number of MAP's possible types increases, i.e., $O(I^2)$, leading to a high computing resources to solve the $(\mathbf{P}_3)$ at the MAP.  
	
	\begin{proposition}
		\label{lemma0a}
		The computational complexity of solving $(\mathbf{P}_3)$ is $O(I^2)$. 
	\end{proposition}
	\begin{proof}
		See Appendix~\ref{appx:pro1}.
	\end{proof}
	
	To reduce the complexity of solving $(\mathbf{P}_3)$, we transform $(\mathbf{P}_3)$ into an equivalent problem using IR and IC constraint transformation. In particular, the Lemma~\ref{lemma1} below shows that all the participating MUs will offer larger encrypted datasets to the MAP when the MAP's type is higher, i.e., the higher willingness to train more encrypted dataset because of higher computing resources.
	\begin{lemma}
		\label{lemma1}
		Given $\big(\boldsymbol{D}^o, \boldsymbol{\rho}^o,\boldsymbol{D}^l, \boldsymbol{\rho}^l\big)$ to be any feasible contracts from the MUs to the MAP such that if $\pi_i \geq \pi_{i^*}$, then $\boldsymbol{D}_i^o \geq \boldsymbol{D}_{i^*}^o$, where $i \neq {i^*}, i, {i^*} \in \mathcal{I}$.
	\end{lemma}
	%request larger data as well as computation sizes and 
	\begin{proof}
		See Appendix~\ref{appx:lemma1}.
	\end{proof}
	
	As a result, the MUs will ask for higher payments to the MAP when the MAP has type $\pi_i$ since $\boldsymbol{D}_i^o \geq \boldsymbol{D}_{i^*}^o$ from Lemma~\ref{lemma1}. This condition is stated in Proposition~\ref{prop1}.
	\begin{proposition}
		\label{prop1}
		If $\boldsymbol{D}_i^o \geq \boldsymbol{D}_{i^*}^o$, then $\boldsymbol{\rho}_i^o \geq \boldsymbol{\rho}_{i^*}^o$,	where $i \neq {i^*}, i, {i^*} \in \mathcal{I}$.
	\end{proposition}
	\begin{proof}
		See Appendix~\ref{appx:pro2}.
	\end{proof}

	From Lemma~\ref{lemma1} and Proposition~\ref{prop1}, the following Proposition~\ref{prop2} shows that the utility of the MAP is monotonically increasing with its type.
	\begin{proposition}
		\label{prop2}
		For any feasible contract $\big(\boldsymbol{D}^o, \boldsymbol{\rho}^o,\boldsymbol{D}^l, \boldsymbol{\rho}^l\big)$, the MAP's utility must hold
		\begin{align}
		&\pi_iG_o(\boldsymbol{\hat \eta}_i,\boldsymbol{D}^o_i) - C_o(\boldsymbol{\hat \eta}_i,\boldsymbol{\rho}^o_i,\boldsymbol{D}^o_i) + G_l(\boldsymbol{D}^l) - C_l(\boldsymbol{\rho}^l) \geq \nonumber \\
		&\pi_{i^*}G_o(\boldsymbol{\hat \eta}_{i^*},\boldsymbol{D}^o_{i^*}) - C_o(\boldsymbol{\hat \eta}_{i^*},\boldsymbol{\rho}^o_{i^*},\boldsymbol{D}^o_{i^*}) + G_l(\boldsymbol{D}^l) - C_l(\boldsymbol{\rho}^l) \nonumber \\
		&\pi_iG_o(\boldsymbol{\hat \eta}_i,\boldsymbol{D}^o_i) - C_o(\boldsymbol{\hat \eta}_i,\boldsymbol{\rho}^o_i,\boldsymbol{D}^o_i) \geq \label{eqn:for10a} \\ &\pi_{i^*}G_o(\boldsymbol{\hat \eta}_{i^*},\boldsymbol{D}^o_{i^*}) - C_o(\boldsymbol{\hat \eta}_{i^*},\boldsymbol{\rho}^o_{i^*},\boldsymbol{D}^o_{i^*}) \nonumber,
		\end{align}
		where $\pi_i \geq \pi_{i^*}$, $i \neq {i^*}, \forall i, {i^*} \in \mathcal{I}$.
	\end{proposition}
	%request larger data as well as computation sizes and 
	\begin{proof}
		See Appendix~\ref{appx:pro3}.
	\end{proof}
	The Proposition~\ref{prop2} helps reducing the number of IR constraints through applying the MAP's minimum type, i.e., $\pi_1$, such that $\pi_iG_o(\boldsymbol{\hat \eta}_i,\boldsymbol{D}^o_i) - C_o(\boldsymbol{\hat \eta}_i,\boldsymbol{\rho}^o_i,\boldsymbol{D}^o_i) + G_l(\boldsymbol{D}^l) - C_l(\boldsymbol{\rho}^l) \geq \pi_{i}G_o(\boldsymbol{\hat \eta}_{1},\boldsymbol{D}^o_{1}) - C_o(\boldsymbol{\hat \eta}_{1},\boldsymbol{\rho}^o_{1},\boldsymbol{D}^o_{1}) + G_l(\boldsymbol{D}^l) - C_l(\boldsymbol{\rho}^l) \geq \pi_{1}G_o(\boldsymbol{\hat \eta}_{1},\boldsymbol{D}^o_{1}) - C_o(\boldsymbol{\hat \eta}_{1},\boldsymbol{\rho}^o_{1},\boldsymbol{D}^o_{1}) + G_l(\boldsymbol{D}^l) - C_l(\boldsymbol{\rho}^l) \geq 0.$
	%\begin{align}
	%\pi_iG_o(\boldsymbol{\hat \eta}_i,\boldsymbol{D}^o_i) - C_o(\boldsymbol{\hat \eta}_i,\boldsymbol{\rho}^o_i,\boldsymbol{D}^o_i) + G_l(\boldsymbol{D}^l) - C_l(\boldsymbol{\rho}^l) &\geq \pi_{i}G_o(\boldsymbol{\hat \eta}_{1},\boldsymbol{D}^o_{1}) - C_o(\boldsymbol{\hat \eta}_{1},\boldsymbol{\rho}^o_{1},\boldsymbol{D}^o_{1}) + G_l(\boldsymbol{D}^l) - C_l(\boldsymbol{\rho}^l)  \nonumber \\
	%& \geq \pi_{1}G_o(\boldsymbol{\hat \eta}_{1},\boldsymbol{D}^o_{1}) - C_o(\boldsymbol{\hat \eta}_{1},\boldsymbol{\rho}^o_{1},\boldsymbol{D}^o_{1}) + G_l(\boldsymbol{D}^l) - C_l(\boldsymbol{\rho}^l) \nonumber \\ &\geq 0. \label{eqn:for12}
	%\end{align}
	To this end, the IR constraints for other types $\pi_i$, $i > 1$, are satisfied if and only if the IR constraint for $\pi_1$ is held. Thus, the IR constraints in~(\ref{eqn:for9b}) can be transformed as follows:
	%\begin{equation}
	%\label{eqn:for13}
	%\begin{aligned}
	%U_{SGP}(\phi_{min},\rho(\phi_{min}),\xi(\phi_{min})) \geq 0,
	%\end{aligned}
	%\end{equation}
	%and
	\begin{equation}
	\label{eqn:for13}
	\begin{aligned}
	\pi_{1}G_o(\boldsymbol{\hat \eta}_{1},\boldsymbol{D}^o_{1}) - C_o(\boldsymbol{\hat \eta}_{1},\boldsymbol{\rho}^o_{1},\boldsymbol{D}^o_{1}) + G_l(\boldsymbol{D}^l) - C_l(\boldsymbol{\rho}^l) \geq 0.
	\end{aligned}
	\end{equation}
	%respectively.
	
	Based on the following Lemma~\ref{lemma2}, we can also reduce the number of IC constraints in~(\ref{eqn:for9c}).
	\begin{lemma}
		\label{lemma2}
		The IC constraints in~(\ref{eqn:for9c}) of $(\mathbf{P}_3)$ can be transformed into the local downward incentive constraints (LDIC) as follows:
		\begin{align}
		&\pi_iG_o(\boldsymbol{\hat \eta}_i,\boldsymbol{D}^o_i) - C_o(\boldsymbol{\hat \eta}_i,\boldsymbol{\rho}^o_i,\boldsymbol{D}^o_i) \geq \label{eqn:for14} \\
		&\pi_{i}G_o(\boldsymbol{\hat \eta}_{i-1},\boldsymbol{D}^o_{i-1}) - C_o(\boldsymbol{\hat \eta}_{i-1},\boldsymbol{\rho}^o_{i-1},\boldsymbol{D}^o_{i-1}), \forall i \in \{2, \ldots, I\}, \nonumber
		\end{align}
		where $\boldsymbol{D}^o_i \geq \boldsymbol{D}^o_{i-1}, \forall i \in \{2, \ldots, I\}$.
		
	\end{lemma}
	\begin{proof}
		See Appendix~\ref{appx:lemma2}.
	\end{proof}
	Equation~(\ref{eqn:for14}) indicates that if the IC constraint for type $\pi_{i-1}$ holds, then all other IC constraints are also satisfied as long as the conditions in Lemma~\ref{lemma1} hold. Finally, using (\ref{eqn:for13}) and (\ref{eqn:for14}), the optimization problem $(\mathbf{P}_3)$ can be transformed into the problem $(\mathbf{P}_4)$ as follows:
	\begin{equation}
	\label{eqn:for9rev}
	\begin{aligned}
	(\mathbf{P}_4) \phantom{10} & \underset{\boldsymbol{D}^o, \boldsymbol{\rho}^o,\boldsymbol{D}^l, \boldsymbol{\rho}^l}{\text{max}} \phantom{5} U_n(\boldsymbol{D}^o, \boldsymbol{\rho}^o,\boldsymbol{D}^l, \boldsymbol{\rho}^l), \forall n \in \mathcal{N},
	\end{aligned}
	\end{equation}
%	\vspace{-0.7cm}
	\begin{align}
	&\text{ s.t. } \quad \text{ 
		(\ref{eqn:for9a})-(\ref{eqn:for9c}), and}, \nonumber \\ 
	& \pi_{1}G_o(\boldsymbol{\hat \eta}_{1},\boldsymbol{D}^o_{1}) - C_o(\boldsymbol{\hat \eta}_{1},\boldsymbol{\rho}^o_{1},\boldsymbol{D}^o_{1}) + G_l(\boldsymbol{D}^l) - C_l(\boldsymbol{\rho}^l) \geq 0, \label{eqn:for9arev} \\
	&\pi_iG_o(\boldsymbol{\hat \eta}_i,\boldsymbol{D}^o_i) - C_o(\boldsymbol{\hat \eta}_i,\boldsymbol{\rho}^o_i,\boldsymbol{D}^o_i) \geq \label{eqn:for9brev} \\
	&\pi_{i}G_o(\boldsymbol{\hat \eta}_{i-1},\boldsymbol{D}^o_{i-1}) - C_o(\boldsymbol{\hat \eta}_{i-1},\boldsymbol{\rho}^o_{i-1},\boldsymbol{D}^o_{i-1}),  \forall i \in \{2, \ldots, I\}, \nonumber \\
	& \boldsymbol{D}^o_i \geq \boldsymbol{D}^o_{i-1}, \forall i \in \{2, \ldots, I\}. \label{eqn:for9crev}
	\end{align}
	This $(\mathbf{P}_4)$ has the complexity $O(I)$ as stated in Proposition~\ref{lemma0b}.
	\begin{proposition}
		\label{lemma0b}
		The computational complexity of solving $(\mathbf{P}_4)$ is $O(I)$. 
	\end{proposition}
	\begin{proof}
		See Appendix~\ref{appx:pro4}.
	\end{proof}
	
	\subsection{Proposed FL Contract Iterative Algorithm}

	To obtain the optimal contracts $\Big(\boldsymbol{\hat D}^o, \boldsymbol{\hat \rho}^o,\boldsymbol{\hat D}^l, \boldsymbol{\hat \rho}^l\Big)$ from problem $(\mathbf{P}_4)$, we implement an iterative process as shown in Algorithm~\ref{ICEA}. In this case,
	the participating MUs in $\mathcal{N}$ first offer the initial contracts prior to the iterative process, i.e., $\Big(\boldsymbol{D}^{o,(\tau)}_n, \boldsymbol{\rho}^{o,(\tau)}_n,\boldsymbol{D}^{l,(\tau)}_n, \boldsymbol{\rho}^{l,(\tau)}_n\Big)$, where $\boldsymbol{D}^{o,(\tau)}_n = [D^{o,(\tau)}_{1,n},\ldots,D^{o,(\tau)}_{i,n},\ldots,D^{o,(\tau)}_{I,n}]$, $\boldsymbol{\rho}^{o,(\tau)}_n = [\rho^{o,(\tau)}_{1,n},\ldots,\rho^{o,(\tau)}_{i,n},\ldots,\rho^{o,(\tau)}_{I,n}]$, and $\tau=0$. Using the initial contracts, the MAP can find the optimal encrypted dataset proportion values $\boldsymbol{\hat \eta}^{(\tau)}$ that maximize the objective function in $(\mathbf{P}_1)$. Then, the MAP can find a new contract for the MU-$n$, $n \in \mathcal{N}$, iteratively using the $\boldsymbol{\hat \eta}^{(\tau)}$ and other MUs' current contracts $\Big(\boldsymbol{D}^{o,(\tau)}_{-n}, \boldsymbol{\rho}^{o,(\tau)}_{-n},\boldsymbol{D}^{l,(\tau)}_{-n}, \boldsymbol{\rho}^{l,(\tau)}_{-n}\Big)$~\cite{Zhu:2015}, aiming at maximizing the objective function in $(\mathbf{P}_4)$ at iteration $\tau+1$. %For each iteration, we may find contracts which can improve the expected profits for all the SVs in $\mathcal{N}(t)$.
	
	The above process continues iteratively and then terminates when the gaps between the expected utilities of MU-$n$, $\forall n \in \mathcal{N}$, at iteration $\tau$ and $\tau+1$ are equal or less than the optimality tolerance $\sigma$. In this way, the algorithm converges and the equilibrium contract solution can be found (as described in Theorems~\ref{theorem_equi1} and \ref{theorem_equi2}). Alternatively, considering the equilibrium contract solution of other MUs $\Big(\boldsymbol{\hat D}_{-n}^o, \boldsymbol{\hat \rho}_{-n}^o,\boldsymbol{\hat D}_{-n}^l, \boldsymbol{\hat \rho}_{-n}^l\Big)$, the expected utility of the MU-$n$ at the equilibrium contract solution $\Big(\boldsymbol{\hat D}_n^o, \boldsymbol{\hat \rho}_n^o,\boldsymbol{\hat D}_n^l, \boldsymbol{\hat \rho}_n^l\Big)$ will produce the highest value compared with the ones using all other contract solutions. As such, all the MU-$n$, $\forall n \in \mathcal{N}$, cannot enhance their expected utilities through deviating from its equilibrium solution unilaterally. This condition is formally described in Definition~\ref{DefV3}. %More details of the convergence, equilibrium, and complexity analysis of the learning contract iterative algorithm are provided in Appendix~\ref{appx:exp2}.
	
	\begin{definition}{}\label{DefV3}
		Equilibrium contract solution for $(\mathbf{P}_4)$: The optimal contracts $\Big(\boldsymbol{\hat D}^o, \boldsymbol{\hat \rho}^o,\boldsymbol{\hat D}^l, \boldsymbol{\hat \rho}^l\Big)$ are the equilibrium solution of the $(\mathbf{P}_4)$ if and only if the conditions
		\begin{align}
		&U_n\Big(\boldsymbol{\hat D}_n^o, \boldsymbol{\hat \rho}_n^o,\boldsymbol{\hat D}_n^l, \boldsymbol{\hat \rho}_n^l,\boldsymbol{\hat D}_{-n}^o, \boldsymbol{\hat \rho}_{-n}^o,\boldsymbol{\hat D}_{-n}^l, \boldsymbol{\hat \rho}_{-n}^l\Big) \geq \nonumber \\
		&U_n\Big(\boldsymbol{D}_n^o, \boldsymbol{\rho}_n^o,\boldsymbol{D}_n^l, \boldsymbol{\rho}_n^l,\boldsymbol{\hat D}_{-n}^o, \boldsymbol{\hat \rho}_{-n}^o,\boldsymbol{\hat D}_{-n}^l, \boldsymbol{\hat \rho}_{-n}^l\Big), \label{eqn:for20}
		\end{align}
		$\forall n \in \mathcal{N}$, satisfy and the optimal contracts $\Big(\boldsymbol{\hat D}^o, \boldsymbol{\hat \rho}^o,\boldsymbol{\hat D}^l, \boldsymbol{\hat \rho}^l\Big)$ still hold the constraints (\ref{eqn:for9arev})-(\ref{eqn:for9crev}).
	\end{definition}
	
	\subsection{Convergence, Equilibrium, and Complexity Analysis of the Proposed FL Contract Algorithm}
	
	In this section, we first can show that the Algorithm~\ref{ICEA} converges to the equilibrium contract solution by adopting the best response method~\cite{Zhu:2015}. In this case, we can define the best response of MU-$n$ at iteration $\tau+1$ using $\Big(\boldsymbol{D}^{o,(\tau)}_{-n}, \boldsymbol{\rho}^{o,(\tau)}_{-n},\boldsymbol{D}^{l,(\tau)}_{-n}, \boldsymbol{\rho}^{l,(\tau)}_{-n}\Big)$ as follows:
	\begin{align}
	&\Theta^{(\tau+1)}_n\Big(\boldsymbol{D}^{o,(\tau)}_{-n}, \boldsymbol{\rho}^{o,(\tau)}_{-n},\boldsymbol{D}^{l,(\tau)}_{-n}, \boldsymbol{\rho}^{l,(\tau)}_{-n}\Big) = \nonumber \\ &\underset{\{\boldsymbol{D}^{o,\text{new}}_n, \boldsymbol{\rho}^{o,\text{new}}_n,\boldsymbol{D}^{l,\text{new}}_n, \boldsymbol{\rho}^{l,\text{new}}_n\} \in \mathbb{Q}_n}{\arg \max} U_{n}(.) \label{eqn:for11a},
	\end{align}
	where
	\begingroup\makeatletter\def\f@size{9.2}\check@mathfonts
	\def\maketag@@@#1{\hbox{\m@th\normalsize\normalfont#1}}%
	\begin{align}
	&U_{n}(.) = \label{eqn:for11b} \\
	&U_{n}\Big(\boldsymbol{D}^{o,\text{new}}_n, \boldsymbol{\rho}^{o,\text{new}}_n,\boldsymbol{D}^{l,\text{new}}_n, \boldsymbol{\rho}^{l,\text{new}}_n,\boldsymbol{D}^{o,(\tau)}_{-n}, \boldsymbol{\rho}^{o,(\tau)}_{-n},\boldsymbol{D}^{l,(\tau)}_{-n}, \boldsymbol{\rho}^{l,(\tau)}_{-n}\Big) \nonumber,
	\end{align}\endgroup
	and $\mathbb{Q}_n$ is the non-empty contract space~\cite{Fraysse:1993} for MU-$n$ with $\mathbb{Q} = {\prod}_{n \in \mathcal{N}}\mathbb{Q}_n$. According to the Algorithm~\ref{ICEA}, we can apply the current contract $\Big(\boldsymbol{D}^{o,\text{new}}_n, \boldsymbol{\rho}^{o,\text{new}}_n,\boldsymbol{D}^{l,\text{new}}_n, \boldsymbol{\rho}^{l,\text{new}}_n\Big) \in \Theta^{(\tau+1)}_n$ to be the new contract at $\tau + 1$, i.e., $\Big(\boldsymbol{D}^{o,(\tau+1)}_n, \boldsymbol{\rho}^{o,(\tau+1)}_n,\boldsymbol{D}^{l,(\tau+1)}_n, \boldsymbol{\rho}^{l,(\tau+1)}_n\Big)$, when the following condition exists
	\begingroup\makeatletter\def\f@size{9.1}\check@mathfonts
	\def\maketag@@@#1{\hbox{\m@th\normalsize\normalfont#1}}%
	\begin{align}
	&\bigg[U_{n}\Big(\boldsymbol{D}^{o,\text{new}}_n, \boldsymbol{\rho}^{o,\text{new}}_n,\boldsymbol{D}^{l,\text{new}}_n, \boldsymbol{\rho}^{l,\text{new}}_n,\boldsymbol{D}^{o,(\tau)}_{-n}, \boldsymbol{\rho}^{o,(\tau)}_{-n},\boldsymbol{D}^{l,(\tau)}_{-n}, \boldsymbol{\rho}^{l,(\tau)}_{-n}\Big) \nonumber \\ 
	&- U_{n}\Big(\boldsymbol{D}^{o,(\tau)}_n, \boldsymbol{\rho}^{o,(\tau)}_n,\boldsymbol{D}^{l,(\tau)}_n, \boldsymbol{\rho}^{l,(\tau)}_n,\boldsymbol{D}^{o,(\tau)}_{-n}, \boldsymbol{\rho}^{o,(\tau)}_{-n},\boldsymbol{D}^{l,(\tau)}_{-n}, \boldsymbol{\rho}^{l,(\tau)}_{-n}\Big)\bigg] \nonumber\\
	&> \sigma. \label{eqn:for11c}
	\end{align}\endgroup
	The algorithm stops when the condition in~(\ref{eqn:for11c}) does not hold for all MU-$n$, $\forall n \in \mathcal{N}$. For that, the algorithm converges as stated in the following Theorem~\ref{theorem_equi1}. %Suppose that the total best response is described by 
	%\begin{equation}
	%\label{eqn:for11b}
	%\begin{aligned}
	%\Gamma\Big(\boldsymbol{\rho}(\phi), \boldsymbol{\xi}(\phi)\Big) = \bigg[\Gamma_i\Big(\boldsymbol{\rho}_{-i}(\phi), \boldsymbol{\xi}_{-i}(\phi)\Big)\bigg]_{i \in \mathcal{I}},
	%\end{aligned}
	%\end{equation}
	
	\begin{theorem}\label{theorem_equi1}
		The Algorithm~\ref{ICEA} converges under the optimality tolerance $\sigma$.
	\end{theorem}
	\begin{proof}
		%We first prove that the best response of SV-$n$, i.e., $\Phi^n$, is an increasing function. Then, we show that there exists a condition where the optimal contracts do not change anymore under the optimality tolerance $\gamma$. More details are provided in Appendix~\ref{appx:theorem1}.
		See Appendix~\ref{appx:theorem1}.
	\end{proof}
	Upon proving that the Algorithm~\ref{ICEA} converges under $\sigma$, we need to guarantee that the Algorithm~\ref{ICEA} also converges to the equilibrium contract solution $\Big(\boldsymbol{\hat D}^o, \boldsymbol{\hat \rho}^o,\boldsymbol{\hat D}^l, \boldsymbol{\hat \rho}^l\Big)$. In particular, we first analyze that the equilibrium solution exists by finding a fixed point in a set-valued function $\Theta$, $\Theta: \mathbb{Q} \rightarrow 2^\mathbb{Q}$, i.e., 
	\begin{equation}
	\label{eqn:for11d}
	\begin{aligned}
	\Theta = \Big[\Theta_n\Big(\boldsymbol{D}_{-n}^o, \boldsymbol{\rho}_{-n}^o,\boldsymbol{D}_{-n}^l, \boldsymbol{\rho}_{-n}^l\Big),\Theta_{-n}\Big(\boldsymbol{D}_{n}^o, \boldsymbol{\rho}_{n}^o,\boldsymbol{ D}_{n}^l, \boldsymbol{\rho}_{n}^l\Big)\Big].
	\end{aligned}
	\end{equation} 
	The finding of this fixed point is equivalent to the equilibrium solution~\cite{Bernheim:1986, Fraysse:1993}, where all the MUs in $\mathcal{N}$ obtain the maximum expected utilities where $\Big(\boldsymbol{\hat D}^o, \boldsymbol{\hat \rho}^o,\boldsymbol{\hat D}^l, \boldsymbol{\hat \rho}^l\Big)$ are found. As a result, the Algorithm~\ref{ICEA} converges to the equilibrium contract solution as formally stated in Theorem~\ref{theorem_equi2}.
	
		\setlength{\textfloatsep}{5pt}% Remove \textfloatsep
	\begin{algorithm}[h]
		\caption{Proposed FL Contract Iterative Algorithm} \label{ICEA}
		{\fontsize{9}{10}\selectfont
			\begin{algorithmic}[1] 
				\STATE Initialize $\tau = 0$ and $\sigma$
				
				\STATE All MU-$n$ send initial contracts $\Big(\boldsymbol{D}^{o,(\tau)}_n, \boldsymbol{\rho}^{o,(\tau)}_n,\boldsymbol{D}^{l,(\tau)}_n, \boldsymbol{\rho}^{l,(\tau)}_n\Big), \forall n \in \mathcal{N}$, to the MAP
				
				\REPEAT
				
				\STATE Find $\boldsymbol{\hat \eta}^{(\tau)}$ values which maximize $(\mathbf{P}_1)$ using $\Big(\boldsymbol{D}^{o,(\tau)}, \boldsymbol{\rho}^{o,(\tau)},\boldsymbol{D}^{l,(\tau)}, \boldsymbol{\rho}^{l,(\tau)}\Big)$
				
				\FOR{$\forall n \in \mathcal{N}$}
				
				\STATE Produce a new contract $\Big(\boldsymbol{D}^{o,\text{new}}_n, \boldsymbol{\rho}^{o,\text{new}}_n,\boldsymbol{D}^{l,\text{new}}_n, \boldsymbol{\rho}^{l,\text{new}}_n\Big)$, which maximizes $(\mathbf{P}_4)$ given $\boldsymbol{\hat \eta}^{(\tau)}$ and $\Big(\boldsymbol{D}^{o,(\tau)}_{-n}, \boldsymbol{\rho}^{o,(\tau)}_{-n},\boldsymbol{D}^{l,(\tau)}_{-n}, \boldsymbol{\rho}^{l,(\tau)}_{-n}\Big)$
				
				\IF{$\text{the condition in~(\ref{eqn:for11c}) is satisfied}$} 
				
				\STATE Set $\Big(\boldsymbol{D}^{o,(\tau+1)}_n, \boldsymbol{\rho}^{o,(\tau+1)}_n,\boldsymbol{D}^{l,(\tau+1)}_n, \boldsymbol{\rho}^{l,(\tau+1)}_n\Big) = \Big(\boldsymbol{D}^{o,\text{new}}_n, \boldsymbol{\rho}^{o,\text{new}}_n,\boldsymbol{D}^{l,\text{new}}_n, \boldsymbol{\rho}^{l,\text{new}}_n\Big)$
				
				\ELSE
				
				\STATE Set $\Big(\boldsymbol{D}^{o,(\tau+1)}_n, \boldsymbol{\rho}^{o,(\tau+1)}_n,\boldsymbol{D}^{l,(\tau+1)}_n, \boldsymbol{\rho}^{l,(\tau+1)}_n\Big) = \Big(\boldsymbol{D}^{o,(\tau)}_n, \boldsymbol{\rho}^{o,(\tau)}_n,\boldsymbol{D}^{l,(\tau)}_n, \boldsymbol{\rho}^{l,(\tau)}_n\Big)$
				
				\ENDIF
				
				\ENDFOR
				
				\STATE $\tau = \tau +1$
				
				\UNTIL \\$U_{n}\Big(\boldsymbol{D}^{o,(\tau)}_n, \boldsymbol{\rho}^{o,(\tau)}_n,\boldsymbol{D}^{l,(\tau)}_n, \boldsymbol{\rho}^{l,(\tau)}_n,\boldsymbol{D}^{o,(\tau)}_{-n}, \boldsymbol{\rho}^{o,(\tau)}_{-n},\boldsymbol{D}^{l,(\tau)}_{-n}, \boldsymbol{\rho}^{l,(\tau)}_{-n}\Big)$ \\ $,\forall n \in \mathcal{N}$, remain unchanged
				
				\STATE Obtain the optimal contracts $\Big(\boldsymbol{\hat D}^o, \boldsymbol{\hat \rho}^o,\boldsymbol{\hat D}^l, \boldsymbol{\hat \rho}^l\Big)$
				
		\end{algorithmic}}
	\end{algorithm}
	
	\begin{theorem}\label{theorem_equi2}
		The best response iterative process in Algorithm~\ref{ICEA} converges to its equilibrium contract $\Big(\boldsymbol{\hat D}^o, \boldsymbol{\hat \rho}^o,\boldsymbol{\hat D}^l, \boldsymbol{\hat \rho}^l\Big)$ if and only if $\Big(\boldsymbol{\hat D}^o, \boldsymbol{\hat \rho}^o,\boldsymbol{\hat D}^l, \boldsymbol{\hat \rho}^l\Big)$ is a fixed point in $\Theta$.
	\end{theorem}
	\begin{proof}
		See Appendix~\ref{appx:theorem2}.
	\end{proof}
	Next, the complexity of Algorithm~\ref{ICEA} can be observed by considering $N$ MUs and the iterative process to find the optimal
	contract for each MU in $\mathcal{N}$. In this case, the complexity can be bounded above by polynomial complexity $\log(N) + e^\epsilon +O(1/N)$, where $\epsilon$ is the Euler constant, as formally derived in Theorem~\ref{theorem3}.
	\begin{theorem}\label{theorem3} 
		Algorithm~\ref{ICEA} has polynomial complexity $\log(N) + e^\epsilon +O(1/N)$. %, where $\epsilon$ is the Euler constant.
	\end{theorem}
	\begin{proof}
		%We first show the probability that an SV in $\mathcal{N}(t)$ modifies the contract or not. Then, using the Markov chain, the number of steps of Algorithm~\ref{ICEA} prior to the convergence can be derived considering the current and updated profit of the SV. More details are provided in Appendix~\ref{appx:theorem3}.
		See Appendix~\ref{appx:theorem3}.
	\end{proof} 

	\section{Federated Learning with Straggling Mitigation and Privacy-Awareness Implementation}
	\label{sec:FLS}

	Using the vectors of optimal encrypted dataset size $\boldsymbol{\hat D}^o$ and optimal local dataset size $\boldsymbol{\hat D}^l$ of the selected MUs in $\mathcal{N}$ according to the optimal contracts $\Big(\boldsymbol{\hat D}^o, \boldsymbol{\hat \rho}^o,\boldsymbol{\hat D}^l, \boldsymbol{\hat \rho}^l\Big)$, we then can execute the entire FL process at the MAP and MUs. %, aiming at improving the global HAR model accuracy. 
	Specifically, we apply a deep learning approach utilizing CNN along with bi-directional LSTM (BLSTM) for a classification prediction model. These two approaches are proved to perform well for time-series problems~\cite{Zhang2:2018}. %The CNN is useful to implement feature extraction/detector of the time-series sensing dataset using convolutional processes with filters (i.e., the number of output filters used in the convolution) and kernel size (i.e., the size of the convolutional window), aiming at reducing the number of model parameters significantly. Meanwhile, the BLSTM can be used to extract the relationship among the past, current, and future information through incorporating forward and backward sequences.
	Let $S^o_i = (\mathbf{F}^o_i, \mathbf{L}^o_i)$ with optimal size $\sum_{n \in \mathcal{N}} {\hat D}_{i,n}^o$ denote the encrypted dataset of all MUs in $\mathcal{N}$ at the MAP with type $\pi_i$, where $\mathbf{F}^o_i$ and $\mathbf{L}^o_i$ are the training feature and training label matrices of encrypted dataset at the MAP with type $\pi_i$, respectively. Additionally, we denote $S^l_n = (\mathbf{F}^l_n, \mathbf{L}^l_n)$ with optimal size ${\hat D}_n^l$ to be the local dataset at each MU-$n$, where $\mathbf{F}^l_n$ and $\mathbf{L}^l_n$ are the training feature and label matrices of local dataset at each MU-$n$, respectively. %Upon generating all the above data matrices, they can be fed into the CNN for the training process. In the CNN, we consider multiple layers including three-stacked convolutional layers, in which each layer is followed by a max pooling layer to reduce the trained features' sizes. The output from the CNN process is then stored as the input of the BLSTM layer. To further minimize overfitting and generalization error, a dropout layer is added after the BLSTM layer. Finally, we add several fully connected layers which can learn features deeply from all the feature combinations of the previous layers and then produce the final classification output. 
	
	For the CNN, we have $\mathbf{F}^{o,a}_i$ and $\mathbf{F}^{l,a}_n$ such that $\mathbf{F}^{o,1}_i = \mathbf{F}^{o}_i$ and $\mathbf{F}^{l,1}_n = \mathbf{F}^{l}_n$, respectively, where $a$ is the training layer, $a \in [1,2,\ldots,a_{max}]$. For each convolutional layer, the training outputs $\mathbf{F}^{o,(a+1)}_i$ and $\mathbf{F}^{l,(a+1)}_n$ can be respectively produced as expressed by~\cite{Zhang2:2018}
	\begin{equation}
	\label{eqn4a}
	\begin{aligned}
	&\mathbf{F}^{o,(a+1)}_i = \alpha_i^a \Big(\sum \mathbf{F}^{o,a}_i \ast \mathbf{Q}\Big), \text{ and } \\ &\mathbf{F}^{l,(a+1)}_n = \alpha_n^a \Big(\sum \mathbf{F}^{l,a}_i \ast \mathbf{Q}\Big),
	\end{aligned}
	\end{equation} 
	where $\mathbf{Q}$ is the squared kernel matrix with a certain number of filters and $\alpha_i^a(.), \alpha_n^a(.)$ are the activation functions~\cite{Zhang2:2018}. These outputs are then fed into a max pooling layer to obtain $\mathbf{F}^{o,(a+2)}_i = \omega \Big(\mathbf{F}^{o,(a+1)}_i\Big) \text{ and } \mathbf{F}^{l,(a+2)}_n = \omega \Big(\mathbf{F}^{l,(a+1)}_n\Big)$,
	%\begin{equation}
	%\label{eqn4b}
	%\begin{aligned}
	%\mathbf{F}^{o,(a+2)}_i = \omega \Big(\mathbf{F}^{o,(a+1)}_i\Big), \text{ and } \mathbf{F}^{l,(a+2)}_n = \omega \Big(\mathbf{F}^{l,(a+1)}_i\Big),
	%\end{aligned}
	%\end{equation} 
	where $\omega$ is the pooling function representing the maximum element values of $\mathbf{F}^{o,(a+1)}_i$ and $\mathbf{F}^{l,(a+1)}_n$ based on the pre-defined pooling size. The above process between convolutional and max pooling layers is repeated for certain times, and then the training outputs from the last max pooling layer can be considered as the inputs of the BLSTM layer, i.e., $\mathbf{F}^{o,a^*}_i$ and $\mathbf{F}^{l,a^*}_n$. Using $\mathbf{F}^{o,a^*}_i$, $\mathbf{F}^{l,a^*}_n$, and the BLSTM method in~\cite{Yu:2015}, we can generate the output matrices $\mathbf{F}^{o,(a^*+1)}_i$ and $\mathbf{F}^{l,(a^*+1)}_n$.

	Upon passing $\mathbf{F}^{o,(a^*+1)}_i$ and $\mathbf{F}^{l,(a^*+1)}_n$ to a dropout layer, and obtaining $\mathbf{F}^{o,(a_{drop})}_i$ and $\mathbf{F}^{l,(a_{drop})}_n$, both matrices can enter the fully connected layers to obtain the final output matrix of layer $a_{max}$, i.e., $\mathbf{L}^{o,a_{max}}_i$ at the MAP and $\mathbf{L}^{l,(a_{max})}_n$ at each MU-$n$, which can be expressed by~\cite{Zhang2:2018}
	\begin{equation}
	\label{eqn4d5}
	\begin{aligned}
	&\mathbf{L}^{o,a_{max}}_i = \alpha_i^{a_{max}} \Big(\mathbf{F}^{o,(a_{max}-1)}_i\mathbf{\Upsilon}_i^{a_{max}}\Big), \text{ and } \\
	&\mathbf{L}^{l,a_{max}}_n = \alpha_n^{a_{max}} \Big(\mathbf{F}^{l,(a_{max}-1)}_n\mathbf{\Upsilon}_n^{a_{max}}\Big),
	\end{aligned}
	\end{equation} 
	where $\alpha_i^{a_{max}}(.)$ and $\alpha_n^{a_{max}}(.)$ are the \emph{softmax} activation functions used to provide a probability distribution for the classification results~\cite{Zhang2:2018}. Moreover, $\mathbf{F}^{o,(a_{max}-1)}_i$ and $\mathbf{F}^{l,(a_{max}-1)}_i$ are the output matrices from the previous fully connected layer $a_{max}-1$, while $\mathbf{\Upsilon}_i^{a_{max}}$ and $\mathbf{\Upsilon}_n^{a_{max}}$ are the weight matrices at the final layer $a_{max}$. 
	
	Based on $\mathbf{L}^o_i$, $\mathbf{L}^l_n$, $\mathbf{L}^{o,a_{max}}_i$, and $\mathbf{L}^{l,a_{max}}_n$, the loss functions of the MAP and each MU-$n$, $n \in \mathcal{N}$, at learning round $\theta$ can be obtained through using a squared Frobenius norm, i.e.,
	\begin{equation}
	\label{eqn4e1}
	\begin{aligned}
	\varphi_i\Big(\mathbf{\tilde \Upsilon}^{(\theta)}\Big) &= \frac{1}{2\sum_{n \in \mathcal{N}} {\hat D}_{i,n}^o}\Big\|\mathbf{L}_i^{o,a_{max}} - \mathbf{L}^o_i\Big\|_F^2 \\&= \frac{1}{2\sum_{n \in \mathcal{N}} {\hat D}_{i,n}^o}\sum_{j=1}^{\sum_{n \in \mathcal{N}} {\hat D}_{i,n}^o}\Big({\hat \ell}_{i,j} - \ell_{i,j}\Big)^2,
	\end{aligned}
	\end{equation}
	and
	\begin{equation}
	\label{eqn4e2}
	\begin{aligned}
	\varphi_n\Big(\mathbf{\Upsilon}^{(\theta)}\Big) &= \frac{1}{2{\hat D}_{n}^l}\Big\|\mathbf{L}_n^{l,a_{max}} - \mathbf{L}^l_n\Big\|_F^2 \\&= \frac{1}{2{\hat D}_{n}^l}\sum_{j=1}^{{\hat D}_{n}^l}\Big({\hat \ell}_{n,j} - \ell_{n,j}\Big)^2,
	\end{aligned}
	\end{equation}
	respectively, where $\mathbf{\tilde \Upsilon}^{(\theta)}$ and $\mathbf{\Upsilon}^{(\theta)}$ are the encrypted and decrypted global model matrices (i.e., $\mathbf{\Upsilon}^{(\theta)} = \emph{\mbox{Dec}}(g_{sk},\mathbf{\tilde \Upsilon}^{(\theta)})$), respectively, containing weights for the BLSTM and fully connected layers at round $\theta$, $\ell_{i,j}$ and $\ell_{n,j}$ are the sample points of ground-truth label data $\mathbf{L}^o_i$ and $\mathbf{L}^l_n$, respectively, while ${\hat \ell}_{i,j}$ and ${\hat \ell}_{n,j}$ are the elements of predicted label data $\mathbf{L}^{o,a_{max}}_i$, and $\mathbf{L}^{l,a_{max}}_n$, respectively. 
	
	From (\ref{eqn4e1}) and (\ref{eqn4e2}), the encrypted gradients at the MAP and local gradient at each MU-$n$ can be respectively expressed by	
	\begin{equation}
	\label{eqn3c}
	\begin{aligned}
	\nabla \mathbf{\tilde \Upsilon}^{(\theta)}_i = \frac{\partial \varphi_i\Big(\mathbf{\tilde \Upsilon}^{(\theta)}\Big)}{\partial \mathbf{\tilde \Upsilon}^{(\theta)}}, \text{ and } \nabla \mathbf{\Upsilon}^{(\theta)}_n = \frac{\partial \varphi_n\Big(\mathbf{\Upsilon}^{(\theta)}\Big)}{\partial \mathbf{\Upsilon}^{(\theta)}}.
	\end{aligned}
	\end{equation}
	Suppose that the local loss gradient at MU-$n$ when we use baseline FL without encrypted training process can be simplified as
	\begin{equation}
	\label{eqn11}
	\begin{aligned}
	\nabla \mathbf{\hat \Upsilon}^{(\theta)}_n &= \frac{\partial \varphi_n\Big(\mathbf{\Upsilon}^{(\theta)}\Big)}{\partial \mathbf{\Upsilon}^{(\theta)}} = \frac{\partial\Bigg[\frac{1}{2{D}_{n}}\Big\|\mathbf{\hat L}_n - \mathbf{L}_n\Big\|_F^2\Bigg]}{\partial \mathbf{\Upsilon}^{(\theta)}}
	%\\&=\frac{\partial\Bigg[\frac{1}{2{D}_{n}}\Big\|\alpha_n\Big(\mathbf{F}_n\mathbf{\Upsilon}^{(\theta)}\Big) - \mathbf{L}_n \Big\|_F^2\Bigg]}{\partial \mathbf{\Upsilon}^{(\theta)}} 
	\\&= \frac{1}{{D}_{n}}\mathbf{F}_n^T\Big(\mathbf{F}_n\mathbf{\Upsilon}^{(\theta)} - \mathbf{L}_n\Big),
	\end{aligned}
	\end{equation}
	where $\mathbf{F}_n$, $\mathbf{L}_n$, and $\mathbf{\hat L}_n$ are the training feature data, ground truth label data, and predicted label data at MU-$n$, respectively. When the encrypted training is used, the encrypted gradient of encrypted dataset at the MAP with type $\pi_i$ in (\ref{eqn3c}) can be derived by
	\begin{equation}
	\label{eqn11a}
	\begin{aligned}
	\nabla \mathbf{\tilde \Upsilon}^{(\theta)}_i &= \frac{1}{\sum_{n \in \mathcal{N}} {\hat D}_{i,n}^o}(\mathbf{F}_{i}^o)^T\Big(\mathbf{F}_{i}^o\mathbf{\tilde \Upsilon}^{(\theta)} - \mathbf{L}_{i}^o\Big).
	%\\&= \mathbf{F}^T\mathbf{H}^T\Bigg(\frac{\mathbf{G}_i^T\mathbf{G}_i}{\sum_{n \in \mathcal{N}} {\hat D}_{i,n}^o}\Bigg)\mathbf{H}\Big(\mathbf{F}\mathbf{\Upsilon}^{(\theta)} - \mathbf{L}\Big),
	\end{aligned}
	\end{equation}
%	where $\mathbf{G}_i$ is the generator matrix for the MAP with type $\pi_i$. With the assumption that $\sum_{n \in \mathcal{N}} {\hat D}_{i,n}^o$ is very large, i.e., $\sum_{n \in \mathcal{N}} {\hat D}_{i,n}^o \rightarrow \infty$~\cite{Prakash:2021}, we then have
%	\begin{equation}
%	\label{eqn11e}
%	\begin{aligned}
%	\frac{\mathbf{G}_i^T\mathbf{G}_i}{\sum_{n \in \mathcal{N}} {\hat D}_{i,n}^o} = \mathbf{I}_{\sum_{n \in \mathcal{N}} D_n},
%	\end{aligned}
%	\end{equation}
%	where $\mathbf{I}_{\sum_{n \in \mathcal{N}} D_n}$ is the identity matrix with size $\sum_{n \in \mathcal{N}} D_n$. Then, we can simplify (\ref{eqn11a}) into
%	\begin{equation}
%	\label{eqn11a2}
%	\begin{aligned}
%	\nabla \mathbf{\Upsilon}^{(\theta)}_i = \mathbf{F}^T\mathbf{H}^T\mathbf{H}\Big(\mathbf{F}\mathbf{\Upsilon}^{(\theta)} - \mathbf{L}\Big).
%	\end{aligned}
%	\end{equation}
	Likewise, the local gradient of the local dataset at each MU-$n$ considering the encrypted training can be written as
	\begin{equation}
	\label{eqn11b}
	\begin{aligned}
	\nabla \mathbf{\Upsilon}^{(\theta)}_n &= \frac{1}{{\hat D}_{n}^l}(\mathbf{F}_{n}^l)^T\Big(\mathbf{F}_{n}^l\mathbf{\Upsilon}^{(\theta)} - \mathbf{L}_{n}^l\Big).
	\end{aligned}
	\end{equation}
	Using (\ref{eqn3c}), each MU-$n$ can encrypt $\nabla \mathbf{\Upsilon}^{(\theta)}_n$ into $\nabla \mathbf{\tilde \Upsilon}^{(\theta)}_n = \emph{\mbox{Enc}}(g_{pk},\nabla \mathbf{\Upsilon}^{(\theta)}_n)$ and share its $\nabla \mathbf{\tilde \Upsilon}^{(\theta)}_n$ to the MAP for the encrypted model aggregation such that the MAP has
	\begin{equation}
	\label{eqn3e}
	\begin{aligned}
	\nabla \mathbf{\tilde\Upsilon}^{(\theta)}_{\mathcal{N}} &= \sum_{n \in \mathcal{N}}{\hat D}_{n}^l\nabla \mathbf{\tilde\Upsilon}^{(\theta)}_n.
	\end{aligned}
	\end{equation}
	As the MAP also has $\nabla \mathbf{\tilde\Upsilon}^{(\theta)}_i$ from the encrypted training, the encrypted global gradient can be updated as follows:
	\begin{equation}
	\label{eqn3f}
	\begin{aligned}
	\nabla\mathbf{\tilde\Upsilon}^{(\theta)} = \frac{1}{\sum_{n \in \mathcal{N}}{\hat D}_n^l + \sum_{n \in \mathcal{N}} {\hat D}_{i,n}^o}\Big(\mathbf{\nabla \tilde\Upsilon}^{(\theta)}_{\mathcal{N}} + \mathbf{\nabla \tilde\Upsilon}^{(\theta)}_i\Big).
	\end{aligned}
	\end{equation}
	Hence, the MAP can renew the encrypted global model for the next learning round as described by
	\begin{equation}
	\label{eqn3g}
	\begin{aligned}
	\mathbf{\tilde\Upsilon}^{(\theta+1)} = \mathbf{\tilde\Upsilon}^{(\theta)} -  \lambda^{(\theta+1)}\Omega\Big(\nabla\mathbf{\tilde\Upsilon}^{(\theta)}\Big),
	\end{aligned}
	\end{equation}
	where $\lambda^{(\theta+1)}$ is the encrypted learning step function of the adaptive learning rate \emph{Adam} optimizer~\cite{Kingma:2015}. %aiming at minimizing the loss functions in (\ref{eqn4e1}) and (\ref{eqn4e2}). 
	Furthermore, $\Omega\Big(\nabla\mathbf{\tilde\Upsilon}^{(\theta)}\Big)$ indicates the local update rules as a function of encrypted gradient $\mathbf{\nabla\tilde\Upsilon}^{(\theta)}$ described in~\cite{Kingma:2015}. 
	To this end, we can also determine the global loss function at learning round $\theta+1$ as
	\begin{equation}
	\label{eqn3i2}
	\begin{aligned}
	\varphi\Big(\mathbf{\tilde\Upsilon}^{(\theta+1)}\Big) = \frac{1}{N+1}\Bigg(\varphi_i\Big(\mathbf{\tilde\Upsilon}^{(\theta)}\Big) + \sum_{n \in \mathcal{N}(t)}\varphi_n\Big(\mathbf{\tilde\Upsilon}^{(\theta)}\Big) \Bigg).
	\end{aligned}
	\end{equation}
	The above global process terminates when the global loss converges or the learning rounds achieve a given threshold $\theta_{\emph{\mbox{th}}}$, and thus the final encrypted global model $\mathbf{\tilde\Upsilon}^*$ and the final global loss $\varphi^*(\mathbf{\tilde\Upsilon}^*)$ are produced. The whole process of proposed FL-based framework is summarized in Algorithm~\ref{DDL-PC} and the convergence analysis is discussed in~Appendix~\ref{appx:theorem4}.

	\begin{algorithm}[t]
		{\fontsize{9}{10}\selectfont
			\caption{Proposed FL-Based Framework Algorithm} \label{DDL-PC}
			
			\begin{algorithmic}[1] 
				
				\STATE Set $\theta_{th}$, $\mathbf{\tilde\Upsilon}^{(0)}$, and $\theta=0$
				
				\STATE The MAP determines MUs in $\mathcal{N} \subset \mathcal{M}$ using (\ref{eqn:for4})
				
				\STATE Perform Algorithm~\ref{ICEA} for all MUs in $\mathcal{N}$ to obtain optimal contracts $\Big(\boldsymbol{\hat D}^o, \boldsymbol{\hat \rho}^o,\boldsymbol{\hat D}^l, \boldsymbol{\hat \rho}^l\Big)$
				
				\FOR{$\forall n \in \mathcal{N}$}
				
				\STATE Implement the data encryption with the optimal size ${\hat D}_{i,n}^o$ considering the MAP with type $\pi_i$
				
				\STATE Send the encrypted dataset to the MAP
				
				\STATE Set $\mathbf{F}_n^l$ and $\mathbf{L}_n^l$ based on the optimal size ${\hat D}_n^l$
				
				\ENDFOR
				
				\STATE The MAP combines the received encrypted datasets into $S_i^o$
				
				\STATE The MAP sets $\mathbf{F}_i^o$ and $\mathbf{L}_i^o$ from $S_i^o$
				
				\WHILE{$\theta \leq \theta_{th} \text{ {\bf and} } \varphi\Big(\mathbf{\tilde\Upsilon}^{(\theta)}\Big) \text{ does not converge }$}
				
				\FOR{$\forall n \in \mathcal{N}$}
				
				\STATE Compute $\mathbf{L}_n^{l,a_{max}}$ using $\mathbf{F}_n^l$ and $\mathbf{\Upsilon}^{(\theta)}$
				
				%			\STATE Find $\varphi_n\Big(\mathbf{\Upsilon}^{(\theta)}\Big)$ and $\mathbf{\Upsilon}^{(\theta,\iota_{\emph{\mbox{th}}})}_n$
				%			
				%			\STATE Send $\mathbf{\Upsilon}^{(\theta,\iota_{\emph{\mbox{th}}})}_n$ to the MAP
				
				\STATE Decrypt $\nabla \mathbf{\tilde\Upsilon}^{(\theta)}$ into $\nabla \mathbf{\Upsilon}^{(\theta)}$
				
				\STATE Find $\varphi_n\Big(\mathbf{\Upsilon}^{(\theta)}\Big)$ and $\nabla \mathbf{\Upsilon}^{(\theta)}_n$
				
				\STATE Encrypt $\nabla \mathbf{\Upsilon}^{(\theta)}_n$ into $\nabla \mathbf{\tilde\Upsilon}^{(\theta)}_n$ and send it to the MAP
				
				\ENDFOR
				
				\STATE The MAP calculates $\mathbf{L}_i^{o,a_{max}}$ using $\mathbf{F}_i^o$ and $\mathbf{\Upsilon}^{(\theta)}$
				
				%			\STATE The MAP obtains $\varphi_i\Big(\mathbf{\Upsilon}^{(\theta)}\Big)$ and $\mathbf{\Upsilon}^{(\theta,\iota_{\emph{\mbox{th}}})}_i$
				
				%			\STATE Aggregate $\mathbf{\Upsilon}^{(\theta,\iota_{\emph{\mbox{th}}})}_n, \forall n \in \mathcal{N}$, using (\ref{eqn3i}) and update the global model $\mathbf{\Upsilon}^{(\theta+1)}$ using (\ref{eqn3j})
				
				\STATE The MAP obtains $\varphi_i\Big(\mathbf{\tilde\Upsilon}^{(\theta)}\Big)$ and $\nabla \mathbf{\tilde\Upsilon}^{(\theta)}_i$
				
				\STATE Aggregate $\nabla \mathbf{\tilde\Upsilon}^{(\theta)}_n, \forall n \in \mathcal{N}$, using (\ref{eqn3e}) and the encrypted global gradient $\nabla\mathbf{\tilde\Upsilon}^{(\theta)}$ using (\ref{eqn3f})
				
				\STATE Update the encrypted global model $\mathbf{\tilde\Upsilon}^{(\theta+1)}$ using (\ref{eqn3g})
				
				\STATE Obtain the global loss $\varphi\Big(\mathbf{\tilde\Upsilon}^{(\theta+1)}\Big)$ using (\ref{eqn3i2})
				
				\STATE $\theta = \theta + 1$
				
				\ENDWHILE
				
				\STATE Finalize the encrypted global model $\mathbf{\tilde\Upsilon}^*$ and global loss $\varphi^*(\mathbf{\tilde\Upsilon}^*)$
				%			
				%			\ENDFOR
				
		\end{algorithmic}}
	\end{algorithm}

	\section{Performance Evaluation}
	\label{sec:PE}
	
	\subsection{Dataset Pre-processing}
	
	We investigate the performance of the proposed FL-based framework utilizing an actual HAR dataset in 2019~\cite{WISDM:2019}. This dataset contains 15M raw sensor samples divided into accelerometer and gyroscope data of smartphones and smartwatches from MUs. To evaluate the efficiency of the FL-based framework, we first combine all $x, y, z$ direction features from data of all smartphones and smartwatches into a single dataset to obtain 12 features with 1M samples (by removing samples with incomplete features). Then, we extract the first six activity labels from total 18 labels including walking, jogging, stairs, sitting, standing, and typing. Using the sampling rate 20 Hz, we convert the raw data into 10-second time-series data to obtain $\sim$ 5K time-series samples with $200$ time periods for each sample.
	
	\subsection{Experiment Setup}
	
	For the contract optimization, we use one agent, i.e., the MAP, and $N$ principals (referring to $N$ selected MUs). Specifically, we set 10 types of the MAP with uniform distribution of the types. Considering the number of HAR samples prior to time-series conversion, we set $D_{max}^o = 5 \times 10^5$ samples. We set $\upsilon_o$ and $\upsilon_l$ at $0.125$ and $3$, respectively, as well as $\alpha_o$ and $\alpha_l$ at $0.001$ and $0.005$, respectively. We also denote $\beta=1$ and $\gamma=0.0001$. We use $\zeta_n = 0.5 \times 10^{-26}, \forall n \in \mathcal{N}$~\cite{Kim:2020}, and $f_n = 2$GHz for practical CPU frequency of today's smartphone. Moreover, we use an average transmission rate $293$ Mbps of 802.11ac WLAN with bandwidth $80$ MHz~\cite{Cisco:2018}. Given that each sample occupies $187$ bytes, we use $c_n = 44880$ cycles/sample considering $30$ cycles/bit. Furthermore, we assume that the cycles per bit at the MAP is much lower than that of the participating MUs due to its much faster computing capability~\cite{Prakash:2021,Kang:2019}. We then compare our proposed contract solution with other baseline and information-symmetry methods. For the baseline method, each MU can offer the proportional size of encrypted dataset to the MAP without using contract policy~\cite{Saputra:2020}. For the information-symmetry method, each MU completely knows the MAP's true type and other MUs' contracts to obtain the optimal contract policy. In this way, we consider this method as the upper bound solution.
	
	To implement the proposed FL process, we utilize \emph{TensorFlow CPU 2.2} containing \emph{TensorFlow Federated} in a shared cluster. Specifically, we consider 100 active MUs which are then reduced (using the MU selection in Section~\ref{sec:MU_select}) to 10 best MUs for the FL process to compare with the conventional FL algorithm (conv-FL), i.e., the MAP waits for a specified number of participating MUs at each round to upload their encrypted local models without using encrypted data training, and proposed framework. For the proposed FL, each MU obtains the size of its optimal encrypted dataset based on the MAP's type. In particular, 10\%, 30\%, and 50\% encrypted datasets of the whole dataset (from all MUs in $\mathcal{N}$), i.e., Enc 10\%, Enc 30\%, and Enc 50\%, represent the MAP with type 2, 6, and 10, respectively. Next, we split the pre-processed dataset into 80\% training set and 20\% testing set. To represent a practical scenario where each mobile device may capture different activities from its MU, we consider a non-i.i.d dataset distribution scenario. As such, we first sort the training set according to the training labels in ascending order. Then, a certain number of samples from the sorted training set is equally distributed to the selected MUs sequentially. For the CNN, we use three-stacked 1D-convolutional layers with 64 filters and 3, 5, and 11 kernel sizes for each corresponding layer. We also utilize a BLSTM and two hidden DNN layers with 64 neurons for each layer followed by a dropout layer with fraction rate 0.25. For the optimizer, we use the Adam optimizer with initial step size 0.01~\cite{Saputra:2020}.
	
	\subsection{Proposed FL Contract Performance}
	
	\subsubsection{IR and IC constraint feasibility of the MAP} 
	
	\begin{figure}[t]
		\begin{center}
			$\begin{array}{cc} 
			\epsfxsize=1.69 in \epsffile{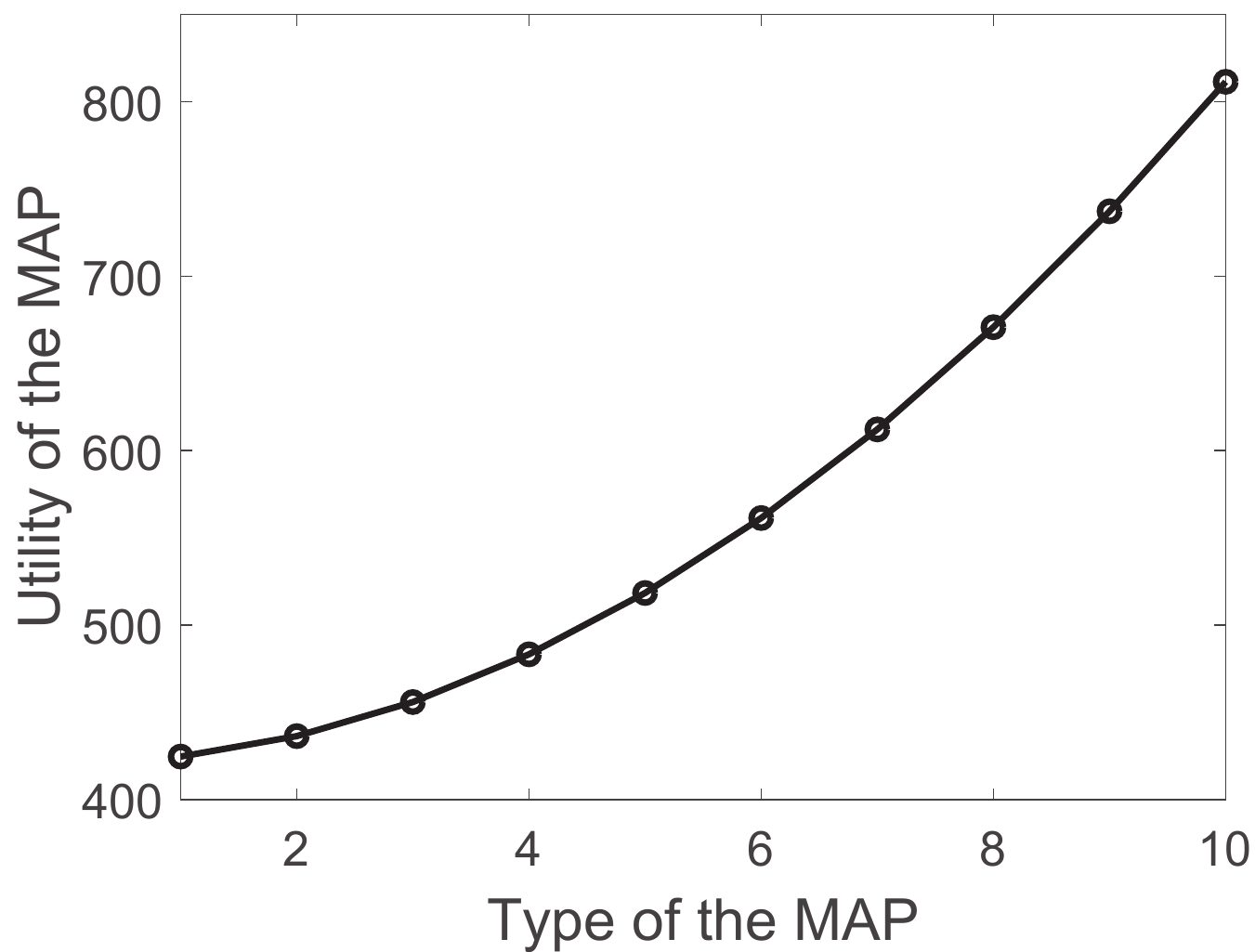} & 
			\hspace*{-0.3cm}
			\epsfxsize=1.69 in \epsffile{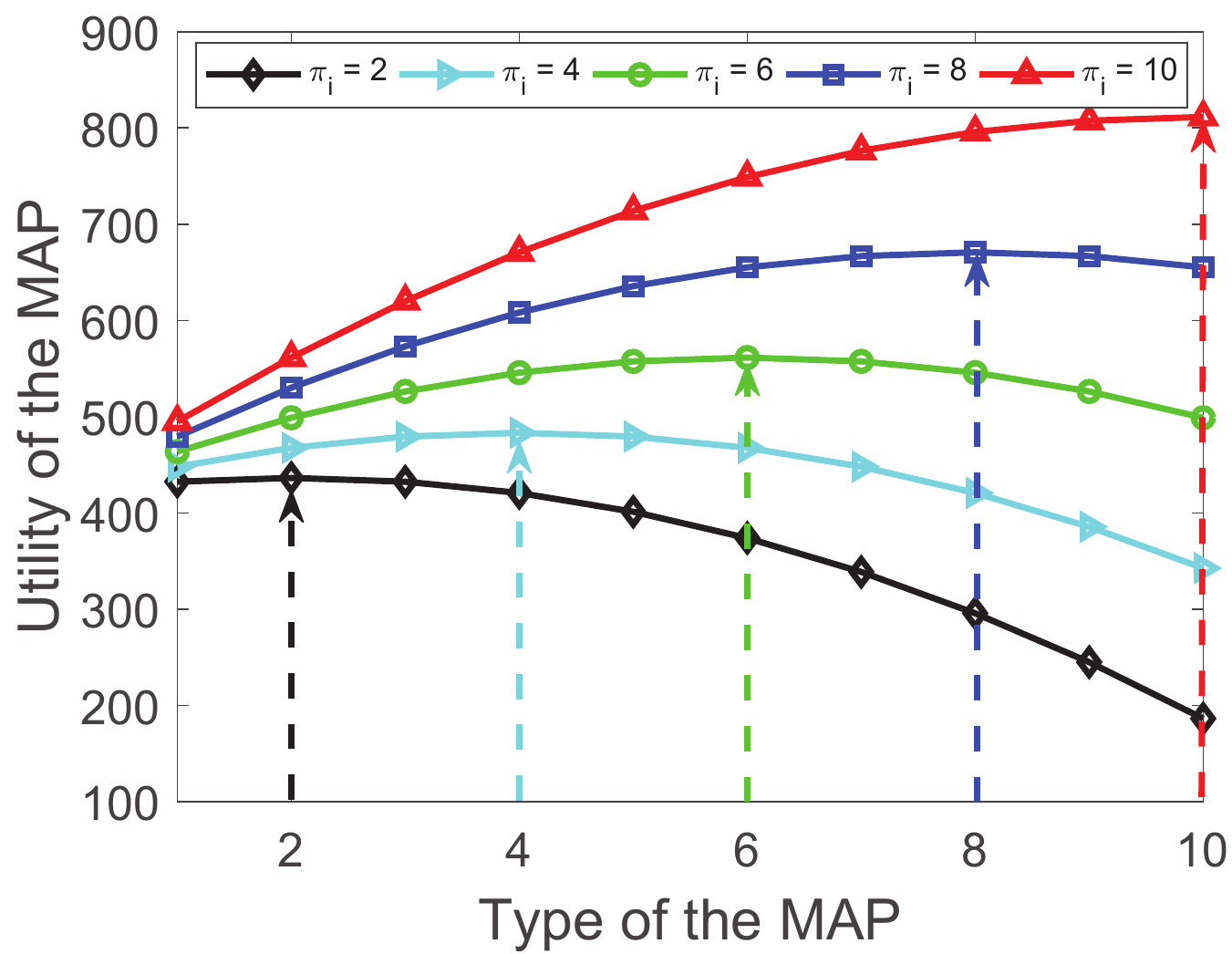} \\ [0.1cm]
			\text{\footnotesize (a) IR constraint} & \text{\footnotesize (b) IC constraint} \\ [0.2cm]
			\vspace*{-0cm}
			\end{array}$
			\vspace{-1.5\baselineskip}
			\caption{The contract feasibility of the MAP.}
			%\vspace{-1.5em}
			\label{fig:IC_IR_const}
		\end{center}
	\end{figure} 

	\begin{figure*}[t]
		\begin{center}
			$\begin{array}{ccc} 
			\epsfxsize=2.14 in \epsffile{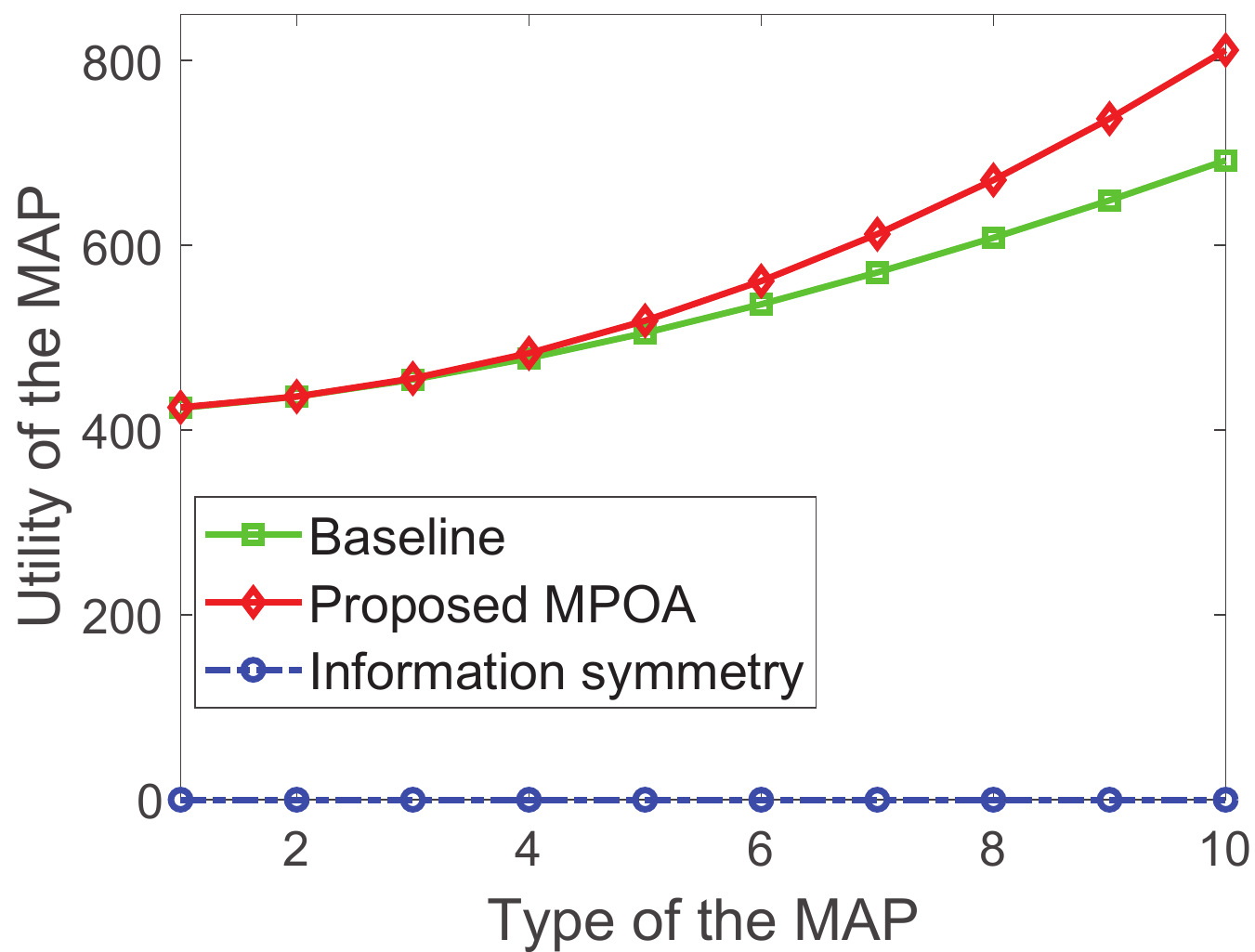} & 
			\hspace*{-.3cm}
			\epsfxsize=2.14 in \epsffile{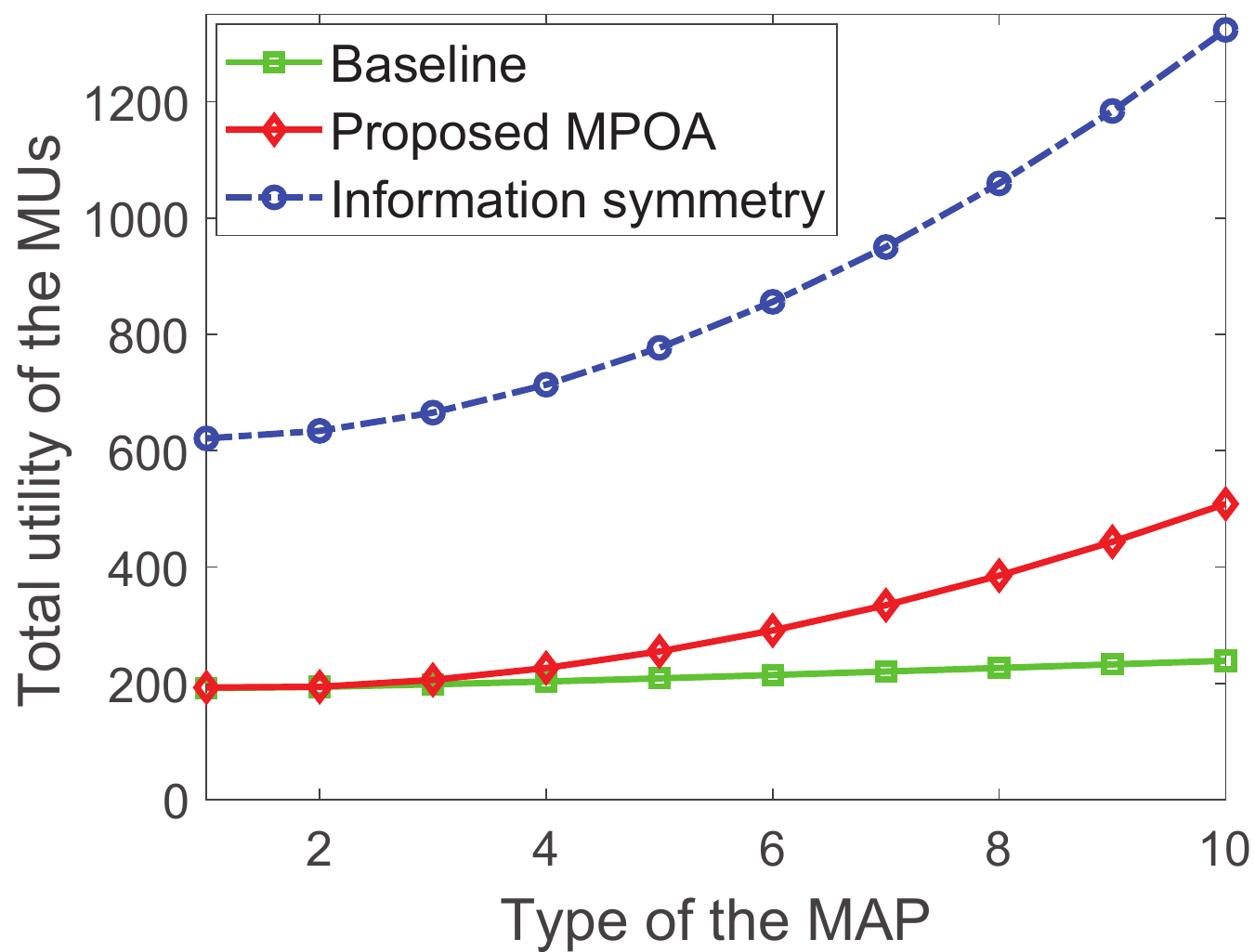} &
			\hspace*{-.3cm}
			\epsfxsize=2.14 in \epsffile{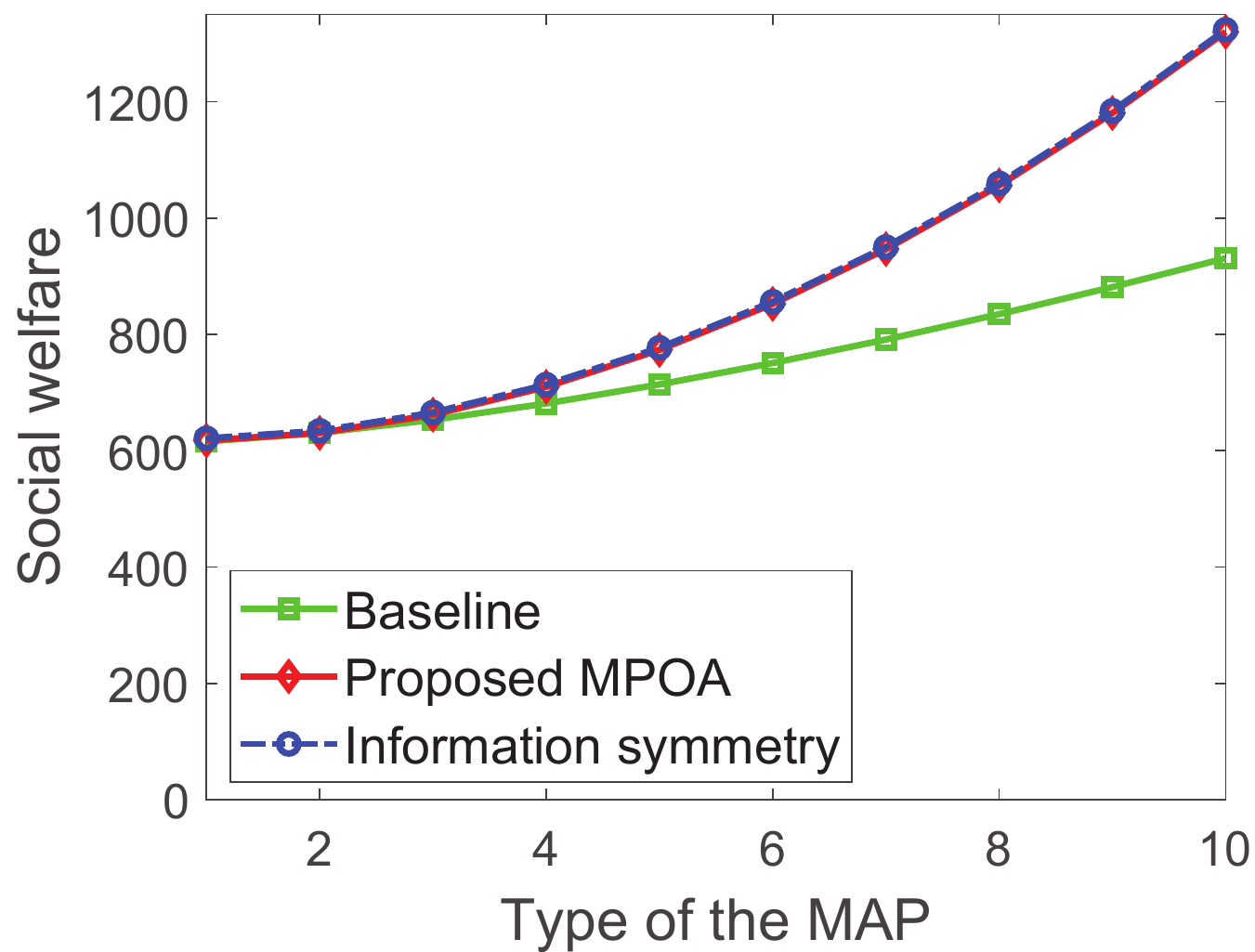} \\ [0.1cm]
			\text{\footnotesize (a) Utility of the MAP} & \text{\footnotesize (b) Total utility of participating MUs} & \text{\footnotesize (c) Social welfare}\\ [0.2cm]
			\vspace*{-0cm}
			\end{array}$
			\vspace{-1.5\baselineskip}
			\caption{The learning contract performance vs type of the MAP.}
			%\vspace{-1.5em}
			\label{fig:utility}
		\end{center}
	\end{figure*} 

	\begin{figure*}[t]
		\begin{center}
			$\begin{array}{ccc} 
			\epsfxsize=2.12 in \epsffile{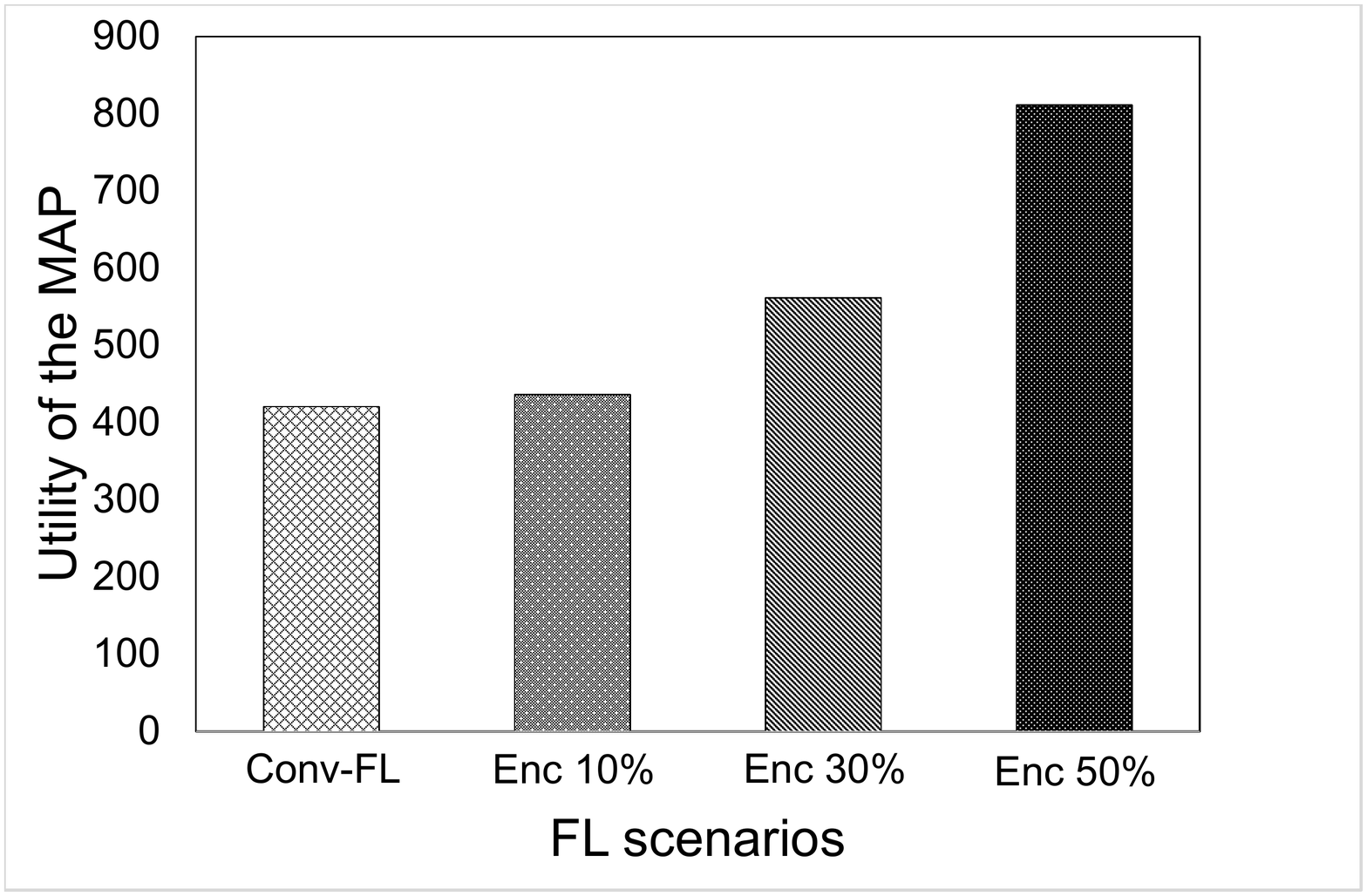} & 
			\hspace*{-.1cm}
			\epsfxsize=2.12 in \epsffile{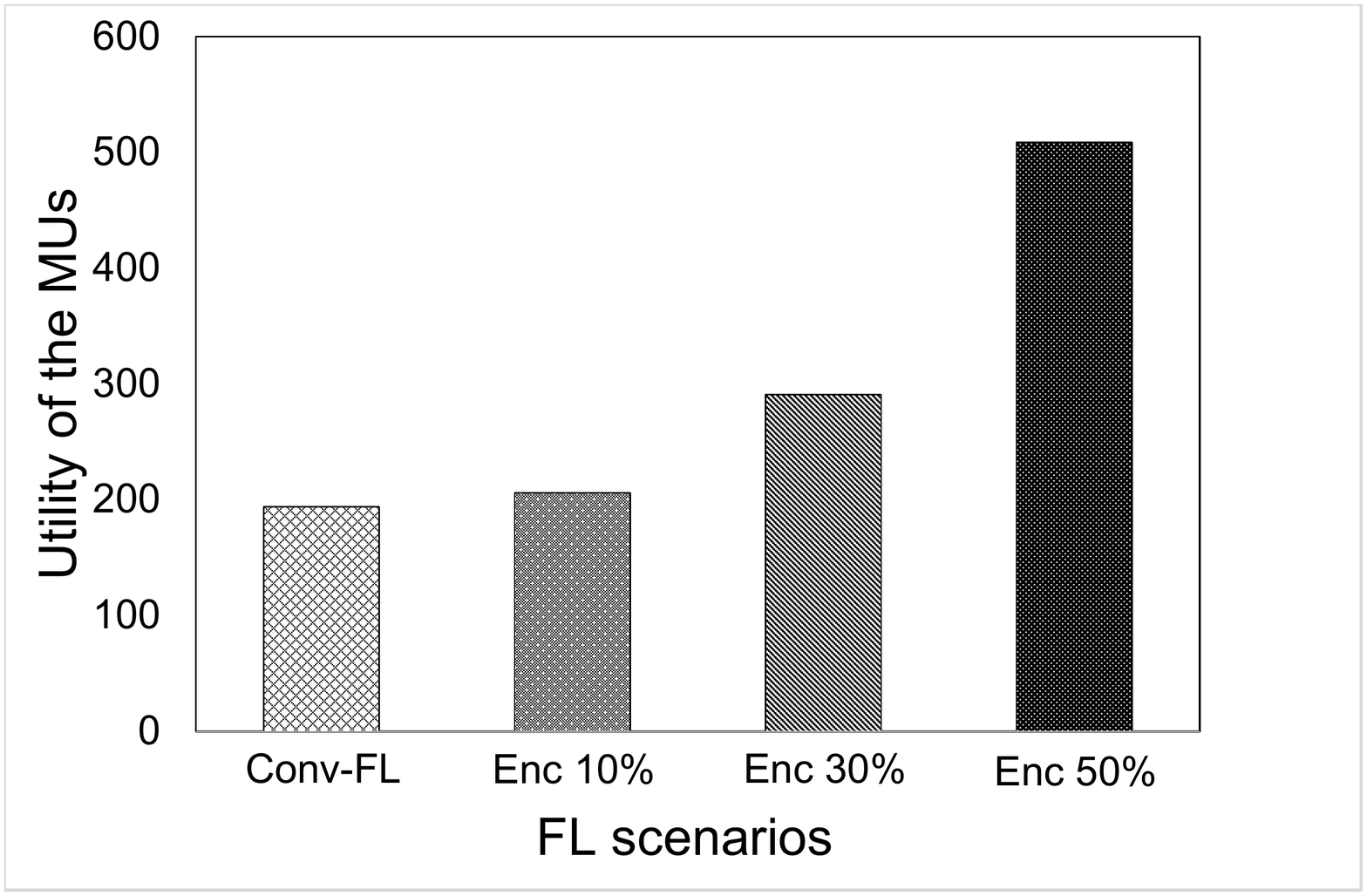} &
			\hspace*{-.1cm}
			\epsfxsize=2.14 in \epsffile{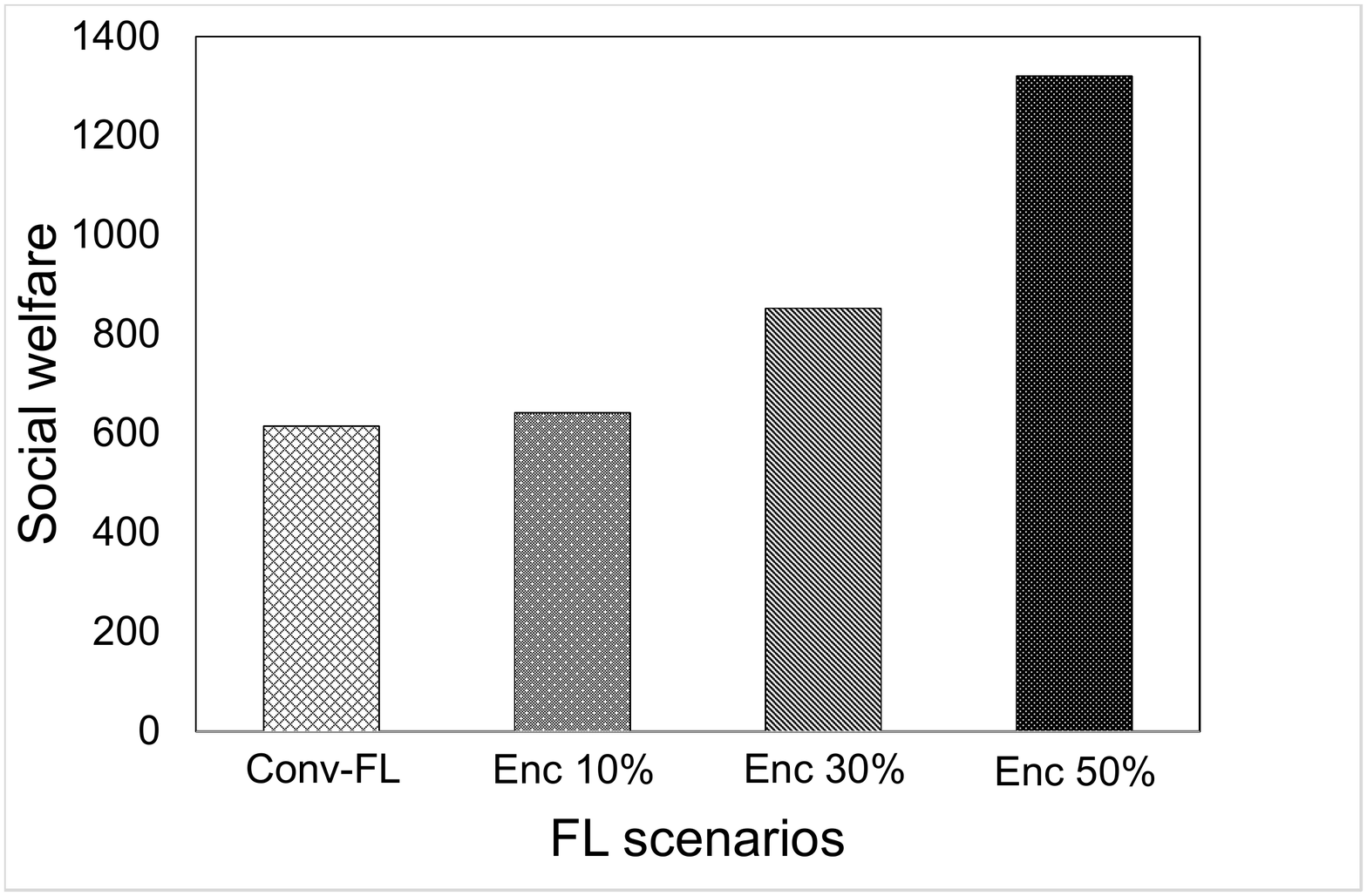} \\ [0.1cm]
			\text{\footnotesize (a) Utility of the MAP} & \text{\footnotesize (b) Total utility of participating MUs} & \text{\footnotesize (c) Social welfare}\\ [0.2cm]
			\vspace*{-0cm}
			\end{array}$
			\vspace{-1.5\baselineskip}
			\caption{The learning contract performance for various FL scenarios.}
			%\vspace{-1.5em}
			\label{fig:utility2}
		\end{center}
	\end{figure*} 
	
	We first demonstrate in Fig.~\ref{fig:IC_IR_const}(a) that the MAP always obtains a non-negative utility to ensure the IR constraint validity as mathematically formulated in~(\ref{eqn:for7a}). This utility follows a monotonic increasing function with respect to the MAP's type because the MAP is willing to train more encrypted datasets due to higher benefit for the MAP in terms of the global model accuracy. Moreover, the MAP can also produce the highest utility when it applies the right contract for its true type, i.e., the IC constraints in~(\ref{eqn:for7a}) are satisfied, as shown in Fig.~\ref{fig:IC_IR_const}(b). For example, the MAP with types 2, 4, 6, 8, and 10 will generate the highest utility when corresponding contracts for the respective types are applied. Since both IR and IC constraints are satisfied, we can obtain the feasible contracts for all participating MUs to maximize their utilities in the proposed FL process. 
	
	\subsubsection{The utility of the MAP/participating MUs and the social welfare} 
	
	\begin{figure*}[t]
		\begin{center}
			$\begin{array}{ccc} 
			\epsfxsize=2 in \epsffile{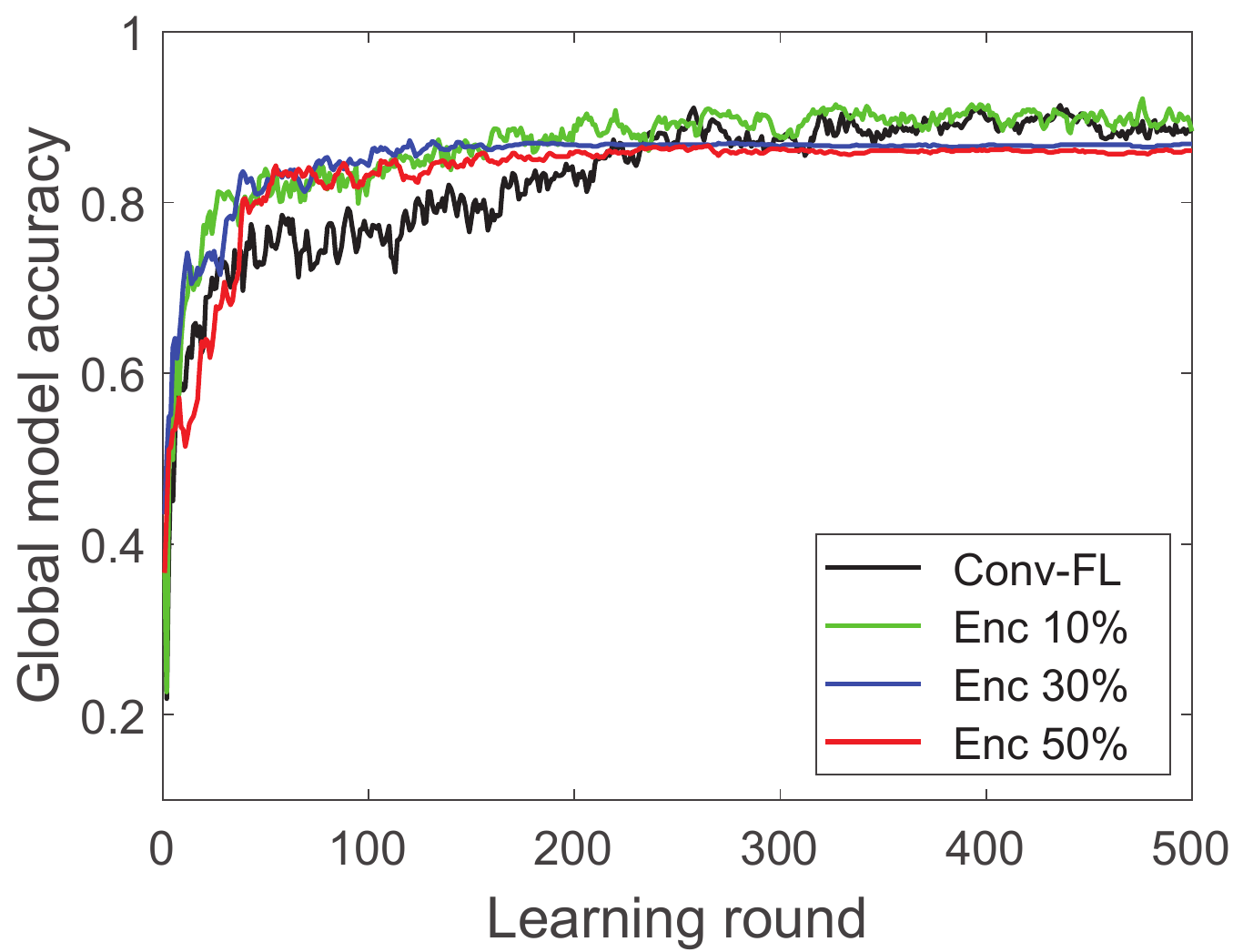} &
			\hspace*{-.3cm}
			\epsfxsize=2 in \epsffile{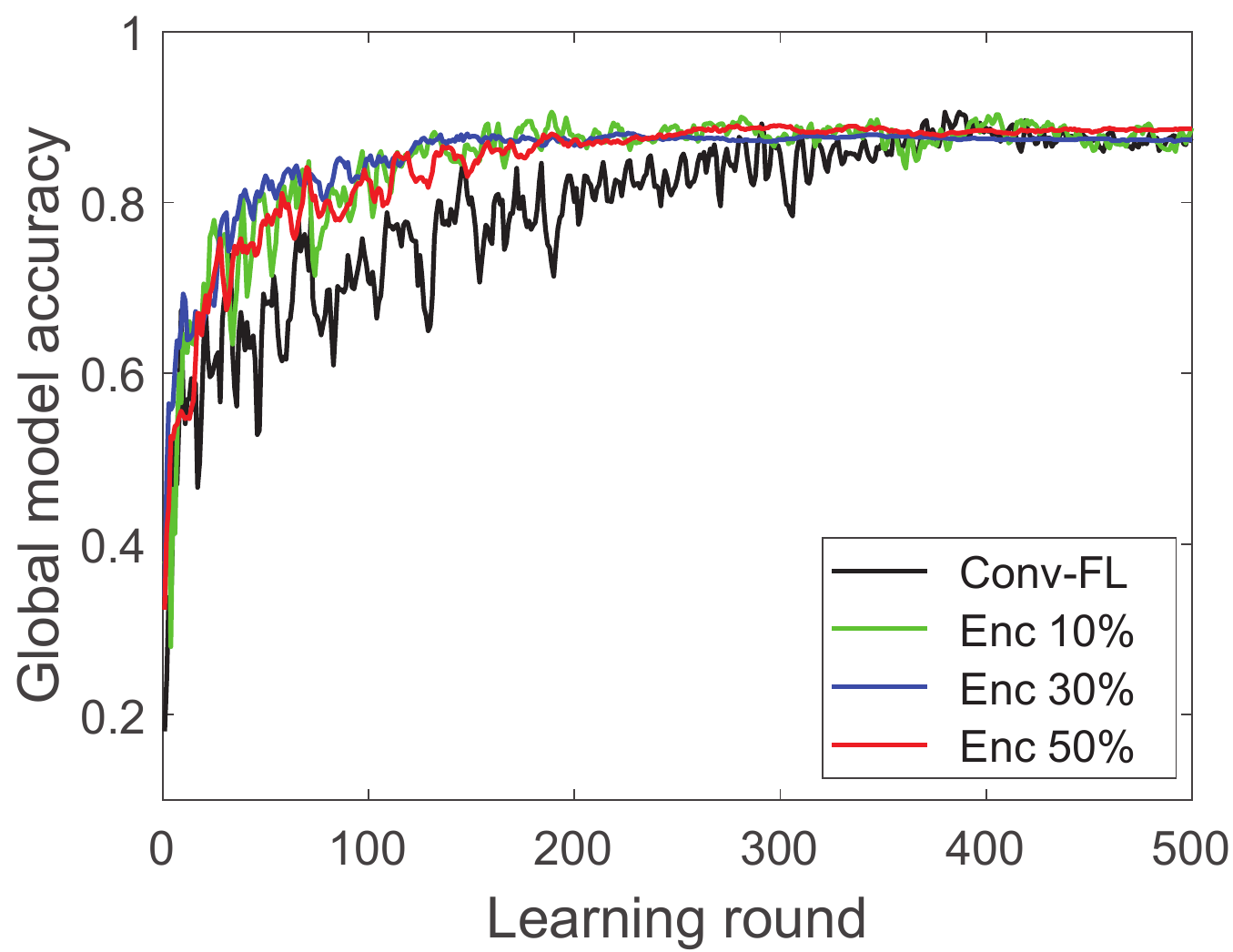} &
			\hspace*{-.3cm}
			\epsfxsize=2 in \epsffile{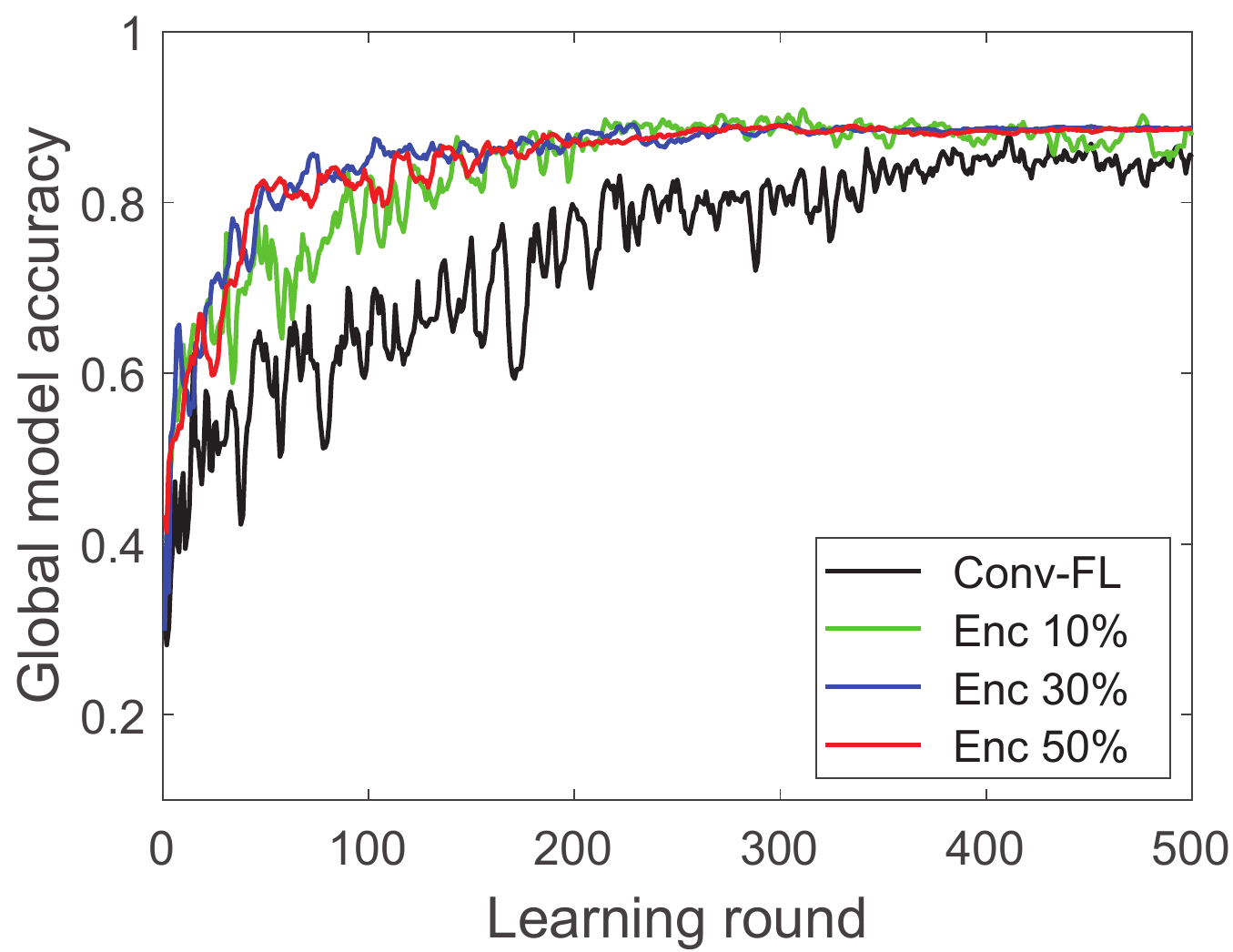} \\ [-0.1cm]
			\text{\footnotesize (a) Accuracy using 10 MUs} & \text{\footnotesize (b) Accuracy using 7 MUs} & \text{\footnotesize (c) Accuracy using 5 MUs}  \\ [0.2cm]
			\end{array}$
			\vspace*{-0cm}
			$\begin{array}{ccc} 
			\epsfxsize=2 in \epsffile{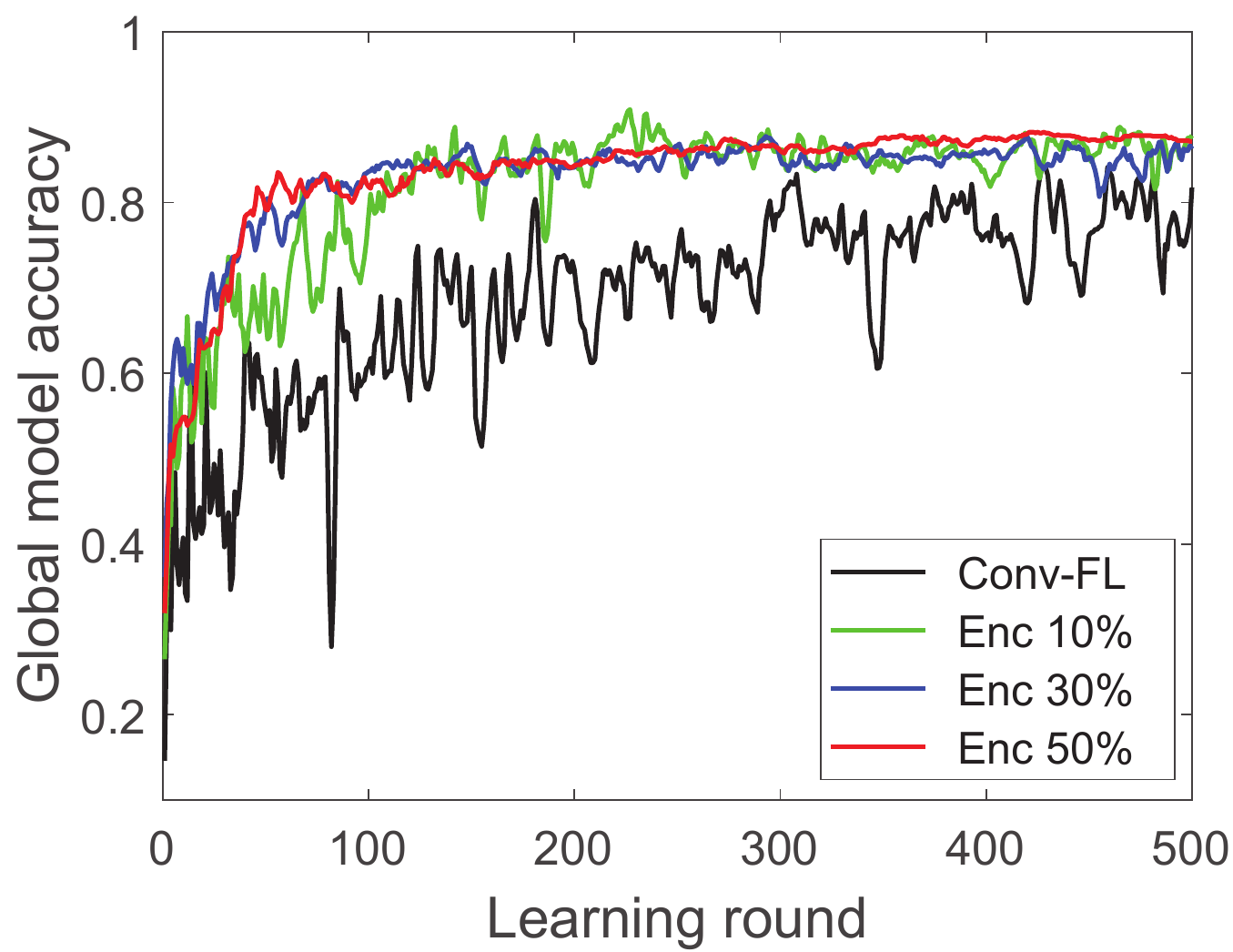} &
			\hspace*{-.3cm}
			\epsfxsize=2 in \epsffile{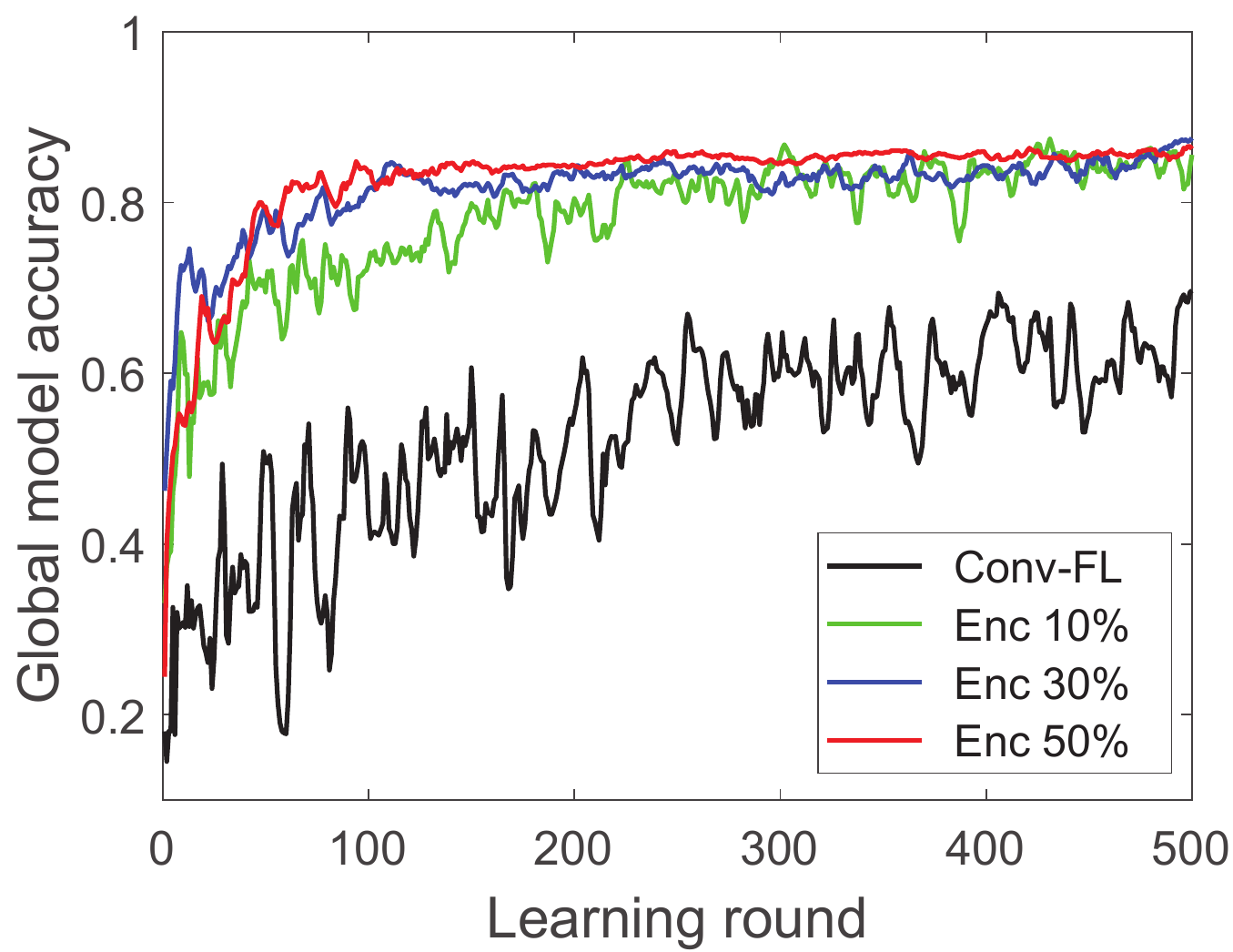} &
			\hspace*{-.3cm}
			\epsfxsize=2.3 in \epsffile{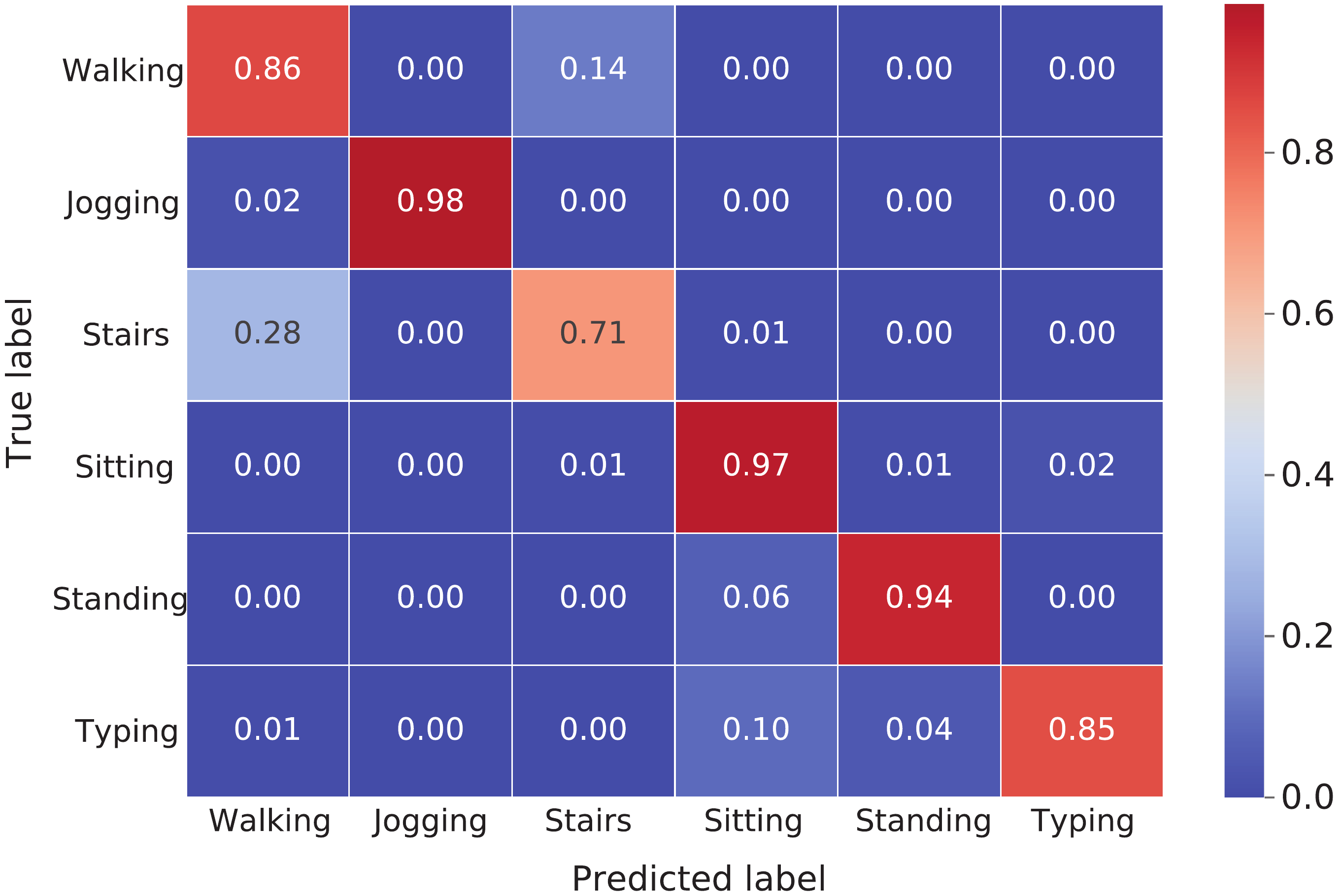}\\ [-0.1cm]
			\text{\footnotesize (d) Accuracy using 3 MUs} & \text{\footnotesize (e) Accuracy using 1 MU} & \text{\footnotesize (f) Confusion matrix} \\ [0.2cm]
			\end{array}$
			\vspace*{-0.3cm}
			\caption{The performance of different FL scenarios using various participating MUs.}
			%\vspace{-0.5em}
			\label{fig:accuracy_niid}
		\end{center}
	\end{figure*}
	
	We then compare the MAP's utility of the proposed framework with those of the baseline and information-symmetry methods. As can be observed in Fig.~\ref{fig:utility}(a), the proposed framework can achieve the highest utility for all types of the MAP especially when the MAP's type gets higher. Particularly, the proposed framework can obtain the utility up to 17\% higher than that of the baseline method. The reason is that the baseline method cannot optimize the encrypted dataset proportion from the MAP since each participating MU only can send the proportional amount of encrypted dataset under the MAP's current computing resources (to train the whole encrypted datasets from all participating MUs). Additionally, the information-symmetry method suffers from zero utility for all types because all participating MUs can collaborate together to obtain maximum utility through completely perceiving the current true type, i.e., the available computing resources, of the MAP. Consequently, as shown in Fig.~\ref{fig:utility}(b), the information-symmetry method can maximize the total utility of participating MUs. In this case, the proposed framework can still outperform the baseline method up to 113\% in terms of the total utility of participating MUs. An interesting point can be observed in Fig.~\ref{fig:utility}(c). Particularly, although the information-symmetry can obtain the total utility of the MUs up to 2.6 times higher than that of the proposed framework, our proposed framework can achieve the social welfare within 0.57\% of the information-symmetry method (which acts as the upper bound solution). Moreover, the proposed framework can obtain the social welfare up to 42\% higher than that of the baseline method. To this end, we can demonstrate that our proposed framework is applicable to implement in the proposed FL process through stabilizing the utility performance of the MAP and participating MUs efficiently. Then, we observe the contract performance of our FL-based framework compared with that of the conventional FL method. As shown in Fig.~\ref{fig:utility2}, the use of encrypted training in our FL-based framework, i.e., 10\%, 30\%, and 50\% encrypted datasets, can improve the utility of the MAP, the total utility of the participating MUs, and social welfare of the network up to 93\%, 162\%, and 114\%, respectively, compared with those of the conventional FL method. The reason is that the MUs can reduce remarkable training tasks through offloading their encrypted datasets to the MAP for the encrypted training, which then leads to better learning quality, especially when some MUs suffer from straggling computing resources. These results reflect that a higher proportion of encrypted dataset at the MAP may result in higher utilities and social welfare for both the participating MUs and the MAP at the expense of a higher privacy protection cost (which is further discussed in Section~\ref{subsec:tra-off}).  
	
	\subsection{Proposed FL Accuracy Performance}
	
	\subsubsection{Accuracy with various number of participating MUs}

		\begin{figure*}[t]
		\begin{center}
			$\begin{array}{ccc} 
			\epsfxsize=2 in \epsffile{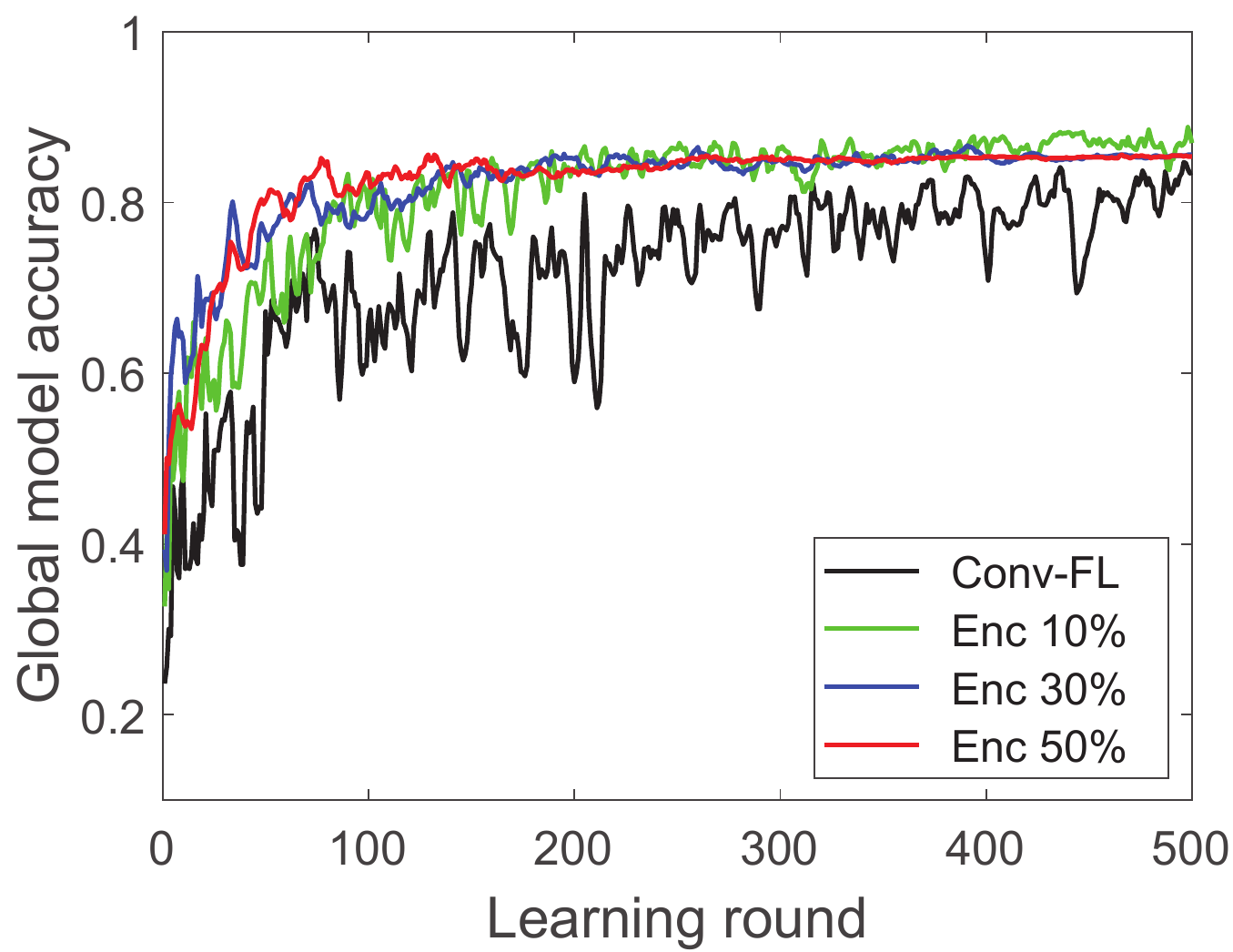} &
			\hspace*{-.3cm}
			\epsfxsize=2 in \epsffile{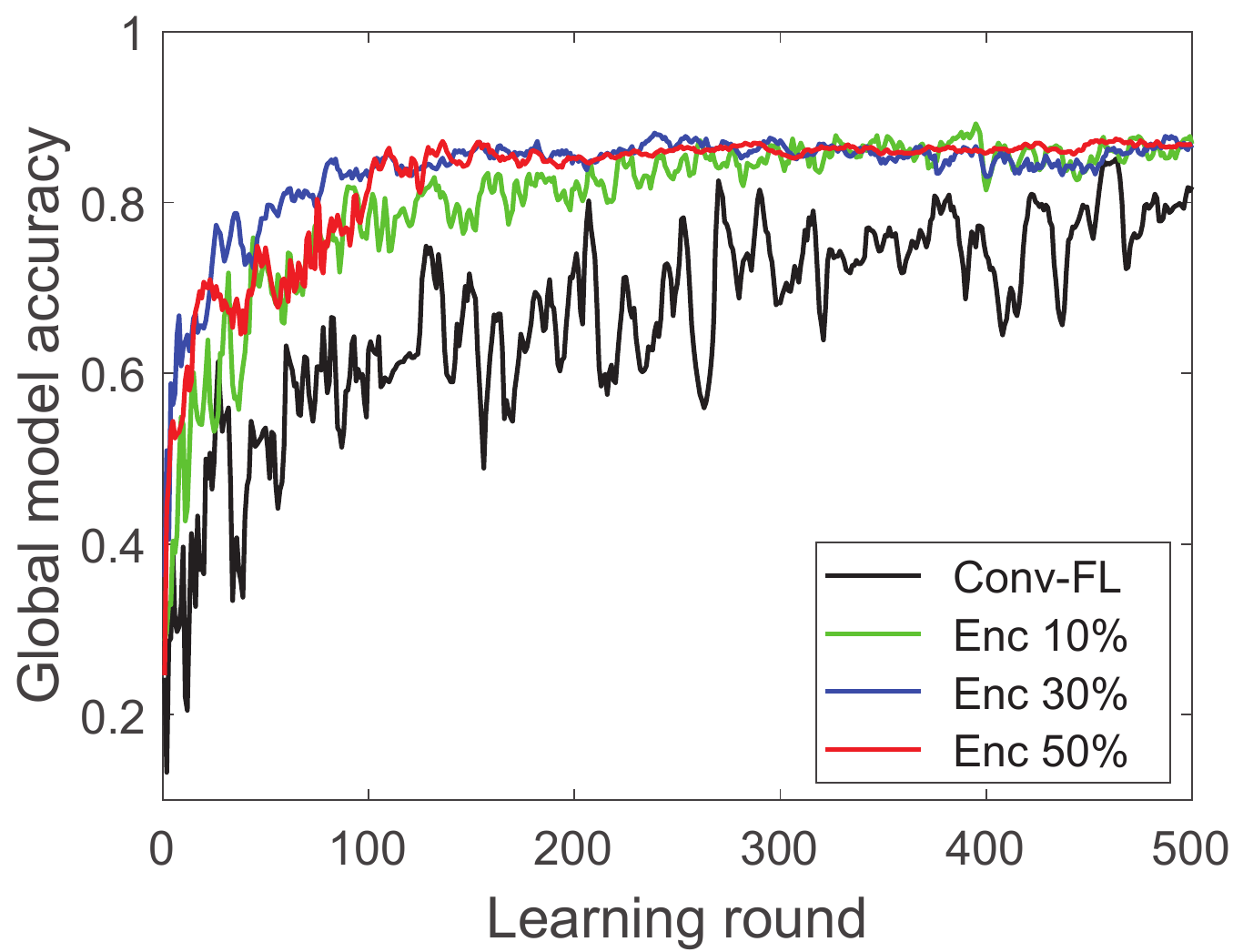} &
			\hspace*{-.3cm}
			\epsfxsize=2 in \epsffile{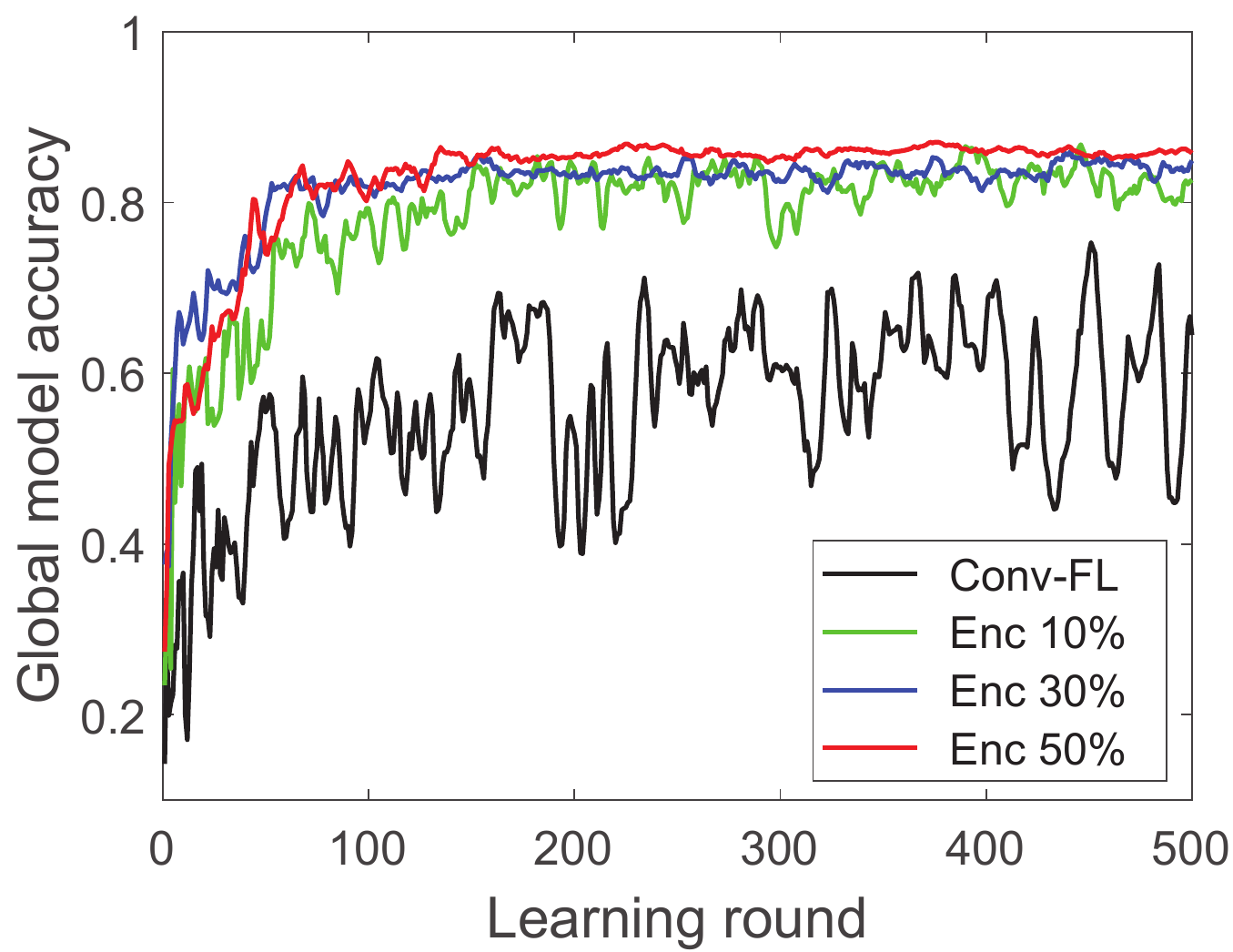} \\ [-0.1cm]
			\text{\footnotesize (a) 20\% straggling probability} & \text{\footnotesize (b) 50\% straggling probability} & \text{\footnotesize (c) 80\% straggling probability}  \\ [0.2cm]
			\end{array}$
			\vspace{-0.5\baselineskip}
			\caption{The performance of various FL scenarios using different straggling probabilities.}
			\vspace{-0.5em}
			\label{fig:accuracy_niid_straggling}
		\end{center}
	\end{figure*}
	
	%\begin{figure}[!t]
	%	\centering
	%	\includegraphics[scale=0.3]{Figs/conf_matrix2c}
	%	\caption{Confusion matrix of the proposed coded FL when 10 participating MUs are utilized.}
	%%	\vspace{-0.5em}
	%	\label{fig:conf_matrix}
	%\end{figure}
	
	In this section, we first analyze performance of the proposed FL with straggling mitigation and privacy-awareness when the MAP updates the encrypted global model upon successfully receiving the encrypted local models from a fixed number of participating MUs at each learning round. As shown in Fig.~\ref{fig:accuracy_niid}, the proposed FL generally can maintain its accuracy performance for all participating MU scenarios regardless the encrypted dataset proportion of the participating MUs. Specifically, when all local models from 10 participating MUs are applied to update the encrypted global model at each round in Fig.~\ref{fig:accuracy_niid}(a), the accuracy performance for the conv-FL and proposed FL approaches can achieve the best results due to the equal size of training samples at each participating MU (including at the MAP for the proposed FL with 10\% encrypted dataset). Although the conv-FL can eventually obtain the same final accuracy level at 88\% as the proposed FL with 10\% encrypted dataset proportion, the proposed FL provides more stable accuracy performance at each learning round and achieves the convergence rate 28.5\% faster than that of the conv-FL method. Note that the proposed FL with 30\% and 50\% encrypted datasets converge to 2\% lower accuracy level (i.e., at level 86\%) compared with that of the conv-FL due to high imbalanced size of non-i.i.d training samples at the participating MUs and the MAP, i.e., the number of training samples at the MAP is much larger than that at each MU. As such, the aggregation of encrypted local models from participating MUs and the MAP will slightly degrade final accuracy of the converged global model~\cite{Lim:2020}. 
	
	When less encrypted local models, i.e., 7 participating MUs to 1 participating MU, are used to update the encrypted global model (representing worse straggling issues), the proposed FL for all encrypted dataset proportions can maintain their accuracy performance at the same final accuracy level, thanks to the additional training at the MAP. Meanwhile, the performance of conv-FL suffers from the accuracy degradation due to insufficient number of non-i.i.d samples to boost the accuracy level. To this end, the proposed FL for all encrypted dataset proportions can obtain the accuracy gap with the conv-FL up to 1.72 times, 1.72 times, 2.9 times, and 4.6 times using 7 MUs, 5 MUs, 3 MUs, and 1 MU scenarios, respectively. Additionally, in Fig.~\ref{fig:accuracy_niid}(e), we can observe that the use of higher encrypted dataset proportion sent to the MAP, i.e., 50\% encrypted dataset, can bring more stable accuracy performance at each learning round and reach the convergence 70\% (70 rounds) and 127\% (140 rounds) faster than those of 10\% and 30\% encrypted dataset scenarios, respectively. This means that when more devices have straggling  problems (i.e., only very few participating MUs are used as the FL learners) at each learning round, the use of larger encrypted training at the MAP can greatly compensate the insufficient training samples at the MUs. In this case, within 500 learning rounds, the conv-FL even cannot reach the convergence due to high fluctuation of the accuracy performance. Then, the confusion matrix to show how good the prediction model towards each activity label for the proposed FL with 10\% encrypted dataset and 10 participating MUs can be seen in Fig.~\ref{fig:accuracy_niid}(f). Based on the above performances, it can be implied that a higher type of the MAP (which corresponds to a larger encrypted dataset proportion at the MAP) generally can speed up the convergence rate, improve global model accuracy, and provide more stable performance especially when a small fixed number of MUs participate in the FL process.
	
	\subsubsection{Accuracy with various straggling probabilities}

	To further show the superiority of proposed FL, we consider various straggling probabilities, i.e, the probabilities that participating MUs experience the straggling problems, e.g., due to low computing resources and/or bad communication resources, such that they cannot send the encrypted local models to the MAP at a certain learning round. Using fixed 5 participating MUs at each learning round, we consider 20\%, 50\%, and 80\% straggling probabilities. In this case, the higher the straggling probability is, the lower the number of participating MUs that can send the encrypted local models to the MAP at each round. As shown in Fig.~\ref{fig:accuracy_niid_straggling}, the proposed FL can preserve its accuracy approximately at level 86\% for all straggling probability scenarios. In contrast, the conv-FL cannot maintain its accuracy when the straggling probability gets worse. As such, in terms of the converged accuracy performance, the proposed FL can outperform the conv-FL up to 1.84 times, 3.35 times, and 3.44 times for 20\%, 50\%, and 80\% straggling probability scenarios, respectively. The interesting point is that the proposed FL with 50\% encrypted dataset can fully keep its accuracy at the same level with stable values for most of learning rounds, while the ones with 10\% and 30\% encrypted datasets degrade their accuracy performances when the straggling probability gets higher. The reason is that the higher encrypted training at the MAP can compensate the straggling problems at the participating MUs more due to its learning process improvement with respect to the larger size of encrypted dataset. 
	
	\begin{figure}[t]
		\begin{center}
			$\begin{array}{cc} 
			\epsfxsize=1.65 in \epsffile{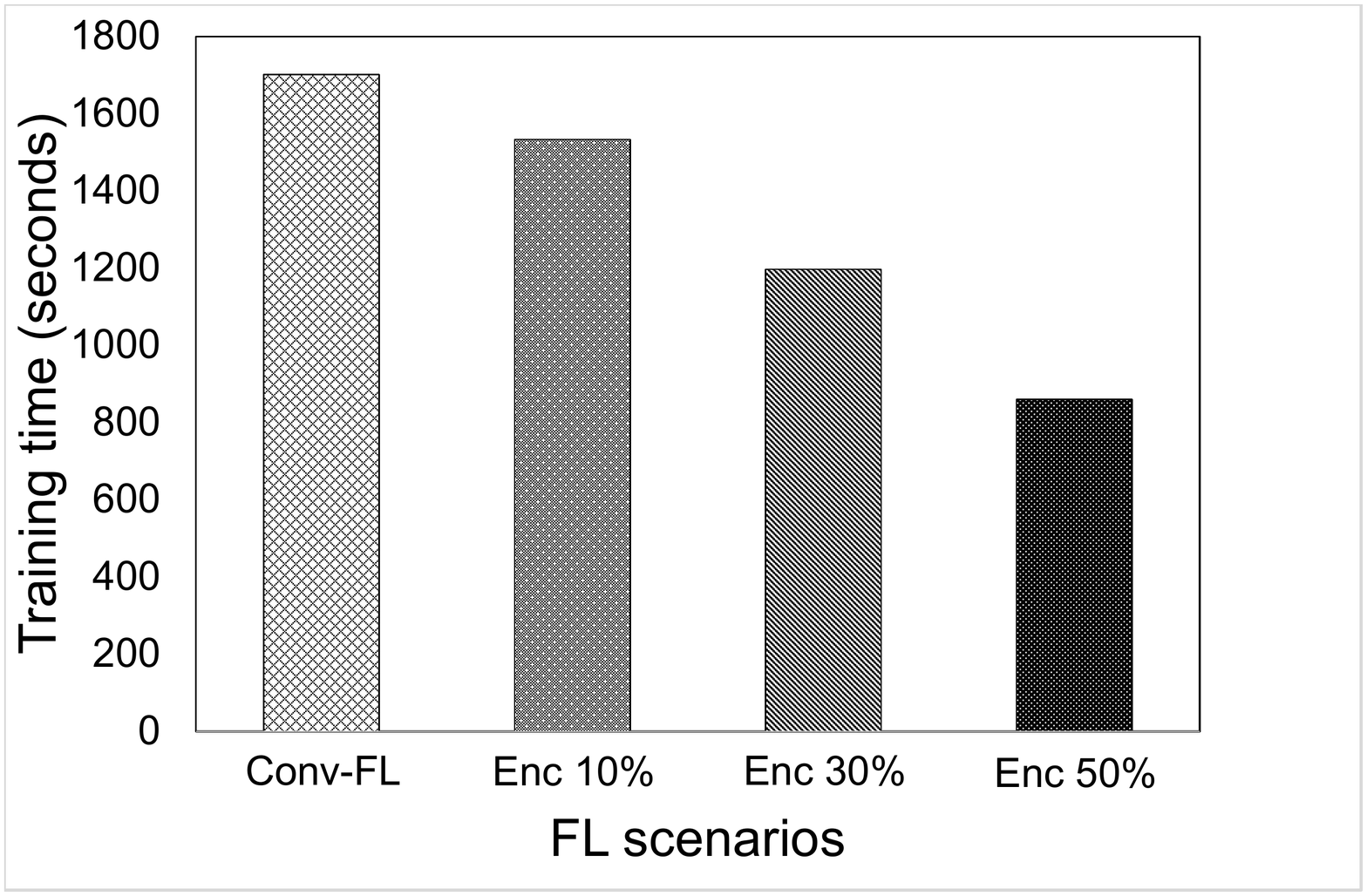} & 
			\hspace*{-.0cm}
			\epsfxsize=1.6 in \epsffile{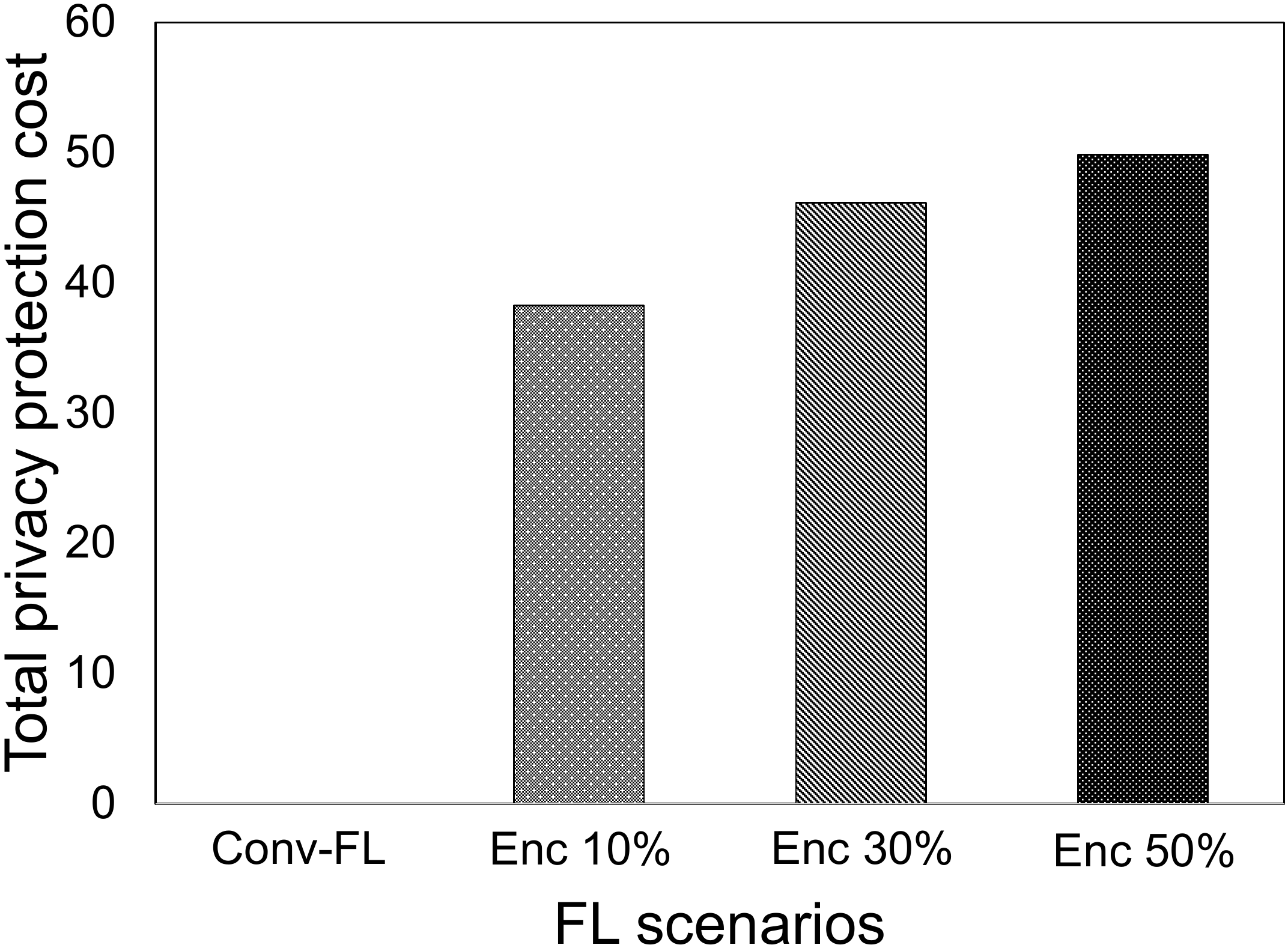} \\ [0.1cm]
			\text{\footnotesize (a) Training time} & \text{\footnotesize (b) Total privacy protection cost} \\ [0.2cm]
			\vspace*{-0cm}
			\end{array}$
			\vspace{-0.5\baselineskip}
			\caption{The training time to achieve 500 learning rounds and the total privacy protection cost of encrypted dataset sharing for various FL scenarios.}
			%\vspace{-1.5em}
			\label{fig:train_priv}
		\end{center}
	\end{figure} 
	
	%\begin{figure}[!t]
	%	\centering
	%	\includegraphics[scale=0.3]{Figs/training_time_new}
	%	\caption{The training time to achieve 500 learning rounds for various FL scenarios.}
	%	%\vspace{-0.5em}
	%	\label{fig:train_time}
	%\end{figure}
	%
	%\begin{figure}[!t]
	%	\centering
	%	\includegraphics[scale=0.3]{Figs/privacy_cost_new}
	%	\caption{The total privacy protection cost of coded dataset sharing for various FL scenarios.}
	%	%\vspace{-0.5em}
	%	\label{fig:privacy_cost}
	%\end{figure}
	
	\subsection{Trade-off Between Training Time and Privacy Cost}
	\label{subsec:tra-off}
	
	In addition to the accuracy and convergence rate performance, the proposed FL outperforms the conv-FL in terms of the training time to complete 500 learning rounds. Specifically, as observed in Fig.~\ref{fig:train_priv}(a), the conv-FL requires the longest training time, i.e., 1702 seconds, to reach 500 learning rounds. Meanwhile, for the proposed FL, the larger the encrypted dataset proportion is trained at the MAP (i.e., the higher the type of the MAP is), the faster the training time to complete 500 learning rounds. This is due to the smaller number of local datasets to be trained at the participating MUs such that the encrypted local models can be aggregated faster at the MAP to update the encrypted global model. In this case, the proposed FL with 10\%, 30\%, 50\% encrypted datasets can reduce the training time up to 10\%, 30\%, and 49\%, respectively, compared with that of the conv-FL. The training time result contradicts with the total privacy protection cost of all participating MUs for sharing encrypted datasets to the MAP. As observed in Fig.~\ref{fig:train_priv}(b), the performance of total privacy cost follows a logarithmic function (as defined in (\ref{eqn:for2c})) with respect to the proportion of total uploaded encrypted dataset. As such, although more encrypted dataset offloading comes with higher privacy protection cost, the privacy of offloaded encrypted datasets can be still protected efficiently. This interesting trade-off provides an insightful information to selectively determine the best FL scenarios which can improve the learning quality, i.e., straggling problem reduction with higher global model accuracy and faster training time/convergence rate, while minimizing the privacy protection cost of encrypted dataset sharing in the FL process.
	
	\section{Conclusion}
	\label{sec:Conc}
	
	In this paper, we have proposed the novel FL-based framework with straggling mitigation and privacy-awareness for the mobile application service to maximize the utilities for the MAP and participating MUs while improving the global model accuracy with additional encrypted training. Particularly, we have designed the MU selection scheme to determine the best MUs based on the abstract of dataset and device information. Using the selected MUs, we have developed the MPOA-based FL contract problem under the MAP's common constraints (i.e., the IR and IC constraints), the competition among participating MUs, and unknown computing resources from the MAP. To obtain the optimal contracts containing the MUs' optimal sizes of encrypted and local datasets for the FL process, we have transformed the problem and implemented the iterative FL contract algorithm that can achieve the equilibrium solution for all the participating MUs. The experiments have demonstrated that our proposed framework can significantly improve the training time, prediction accuracy, utilities, and social welfare of the MAP and participating MUs (while considering the privacy protection cost of encrypted datasets) compared with those of other baseline methods.
	
	% Can use something like this to put references on a page
	% by themselves when using endfloat and the captionsoff option.
	\ifCLASSOPTIONcaptionsoff
	\newpage
	\fi

	\vspace{0cm}

	\clearpage
	\appendices
	
	\section{Proof of Proposition 1}
	\label{appx:pro1}
	
	From constraints (\ref{eqn:for9a})-(\ref{eqn:for9e}), there exist $I$ trainable encrypted dataset constraints, $N$ trainable local dataset constraints, $I \times N$ total dataset constraints, $I$ IR constraints, and $I(I-1)$ IC constraints, where $N$ is a constant. %As such, the total number of constraints in problem $(\mathbf{P}_2)$ is $\big(J + J + J(J-1)\big) = J + J^2$. 
	Hence, the problem $(\mathbf{P}_3)$ can be solved with the computational complexity $O(I^2)$.
	
	\section{Proof of Lemma 1}
	\label{appx:lemma1}
	
	We first show that if $\pi_i \geq \pi_{i^*}$, then $\boldsymbol{D}_i^o \geq \boldsymbol{D}_{i^*}^o$ using the IC constraint in~(\ref{eqn:for7}), i.e.,
	\begin{align}
	&\pi_iG_o(\boldsymbol{\hat \eta}_i,\boldsymbol{D}^o_i) - C_o(\boldsymbol{\hat \eta}_i,\boldsymbol{\rho}^o_i,\boldsymbol{D}^o_i) \geq \nonumber\\
	&\pi_{i}G_o(\boldsymbol{\hat \eta}_{i^*},\boldsymbol{D}^o_{i^*}) - C_o(\boldsymbol{\hat \eta}_{i^*},\boldsymbol{\rho}^o_{i^*},\boldsymbol{D}^o_{i^*}). \label{eqn:for7c}
	\end{align}
	Meanwhile, when the MAP has type $i^*$, we have
	\begin{align}
	&\pi_{i^*}G_o(\boldsymbol{\hat \eta}_{i^*},\boldsymbol{D}^o_{i^*}) - C_o(\boldsymbol{\hat \eta}_{i^*},\boldsymbol{\rho}^o_{i^*},\boldsymbol{D}^o_{i^*}) \geq \nonumber\\
	&\pi_{i^*}G_o(\boldsymbol{\hat \eta}_{i},\boldsymbol{D}^o_{i}) - C_o(\boldsymbol{\hat \eta}_{i},\boldsymbol{\rho}^o_{i},\boldsymbol{D}^o_{i}), \label{eqn:for7d}
	\end{align}
	where $i \neq {i^*}$ and $i, {i^*} \in \mathcal{I}$. From~(\ref{eqn:for7c}) and~(\ref{eqn:for7d}), we can derive
	\begin{equation}
	\label{eqn:for7e}
	\begin{aligned}
	&\pi_iG_o(\boldsymbol{\hat \eta}_i,\boldsymbol{D}^o_i) + \pi_{i^*}G_o(\boldsymbol{\hat \eta}_{i^*},\boldsymbol{D}^o_{i^*}) \geq \\ &\pi_{i}G_o(\boldsymbol{\hat \eta}_{i^*},\boldsymbol{D}^o_{i^*}) + \pi_{i^*}G_o(\boldsymbol{\hat \eta}_{i},\boldsymbol{D}^o_{i}),\\
	&\pi_iG_o(\boldsymbol{\hat \eta}_i,\boldsymbol{D}^o_i) - \pi_{i^*}G_o(\boldsymbol{\hat \eta}_{i},\boldsymbol{D}^o_{i}) \geq \\ &\pi_{i}G_o(\boldsymbol{\hat \eta}_{i^*},\boldsymbol{D}^o_{i^*}) - \pi_{i^*}G_o(\boldsymbol{\hat \eta}_{i^*},\boldsymbol{D}^o_{i^*}),\\
	&(\pi_i - \pi_{i^*})G_o(\boldsymbol{\hat \eta}_i,\boldsymbol{D}^o_i) \geq (\pi_i - \pi_{i^*})G_o(\boldsymbol{\hat \eta}_{i^*},\boldsymbol{D}^o_{i^*}).
	\end{aligned}
	\end{equation}
	By omitting $(\pi_i - \pi_{i^*})$, we obtain $G_o(\boldsymbol{\hat \eta}_i,\boldsymbol{D}^o_i) \geq G_o(\boldsymbol{\hat \eta}_{i^*},\boldsymbol{D}^o_{i^*})$. As the gain function in~(\ref{eqn:for30a}) is monotonically increasing in $\boldsymbol{D}^o_i$, we obtain that if $G_o(\boldsymbol{\hat \eta}_i,\boldsymbol{D}^o_i) \geq G_o(\boldsymbol{\hat \eta}_{i^*},\boldsymbol{D}^o_{i^*})$, then $\boldsymbol{D}_i^o \geq \boldsymbol{D}_{i^*}^o$.
	
	Additionally, we show that if $\boldsymbol{D}_i^o \geq \boldsymbol{D}_{i^*}^o$ then $\pi_i \geq \pi_{i^*}$. From~(\ref{eqn:for7c})-(\ref{eqn:for7e}), we derive that $\pi_i(G_o(\boldsymbol{\hat \eta}_i,\boldsymbol{D}^o_i) - G_o(\boldsymbol{\hat \eta}_{i^*},\boldsymbol{D}^o_{i^*})) \geq \pi_{i^*}(G_o(\boldsymbol{\hat \eta}_i,\boldsymbol{D}^o_i) - G_o(\boldsymbol{\hat \eta}_{i^*},\boldsymbol{D}^o_{i^*}))$.
%	\begin{equation}
%	\label{eqn:for7f}
%	\begin{aligned}
%	&\pi_i(G_o(\boldsymbol{\hat \eta}_i,\boldsymbol{D}^o_i) - G_o(\boldsymbol{\hat \eta}_{i^*},\boldsymbol{D}^o_{i^*})) \geq \\ &\pi_{i^*}(G_o(\boldsymbol{\hat \eta}_i,\boldsymbol{D}^o_i) - G_o(\boldsymbol{\hat \eta}_{i^*},\boldsymbol{D}^o_{i^*})).
%	\end{aligned}
%	\end{equation}
	By removing $G_o(\boldsymbol{\hat \eta}_i,\boldsymbol{D}^o_i) - G_o(\boldsymbol{\hat \eta}_{i^*},\boldsymbol{D}^o_{i^*})$ from both sides, we have $\pi_i \geq \pi_{i^*}$ accounting for $G_o(\boldsymbol{\hat \eta}_i,\boldsymbol{D}^o_i) - G_o(\boldsymbol{\hat \eta}_{i^*},\boldsymbol{D}^o_{i^*}) \geq 0$ since $G_o(\boldsymbol{\hat \eta}_i,\boldsymbol{D}^o_i) \geq G_o(\boldsymbol{\hat \eta}_{i^*},\boldsymbol{D}^o_{i^*})$, and thus $\boldsymbol{D}_i^o \geq \boldsymbol{D}_{i^*}^o$. As a result, if $\boldsymbol{D}_i^o \geq \boldsymbol{D}_{i^*}^o$, then $\pi_i \geq \pi_{i^*}$. This concludes the proof.

	\section{Proof of Proposition 2}
	\label{appx:pro2}
	
	From~(\ref{eqn:for7c}) and (\ref{eqn:for7d}), we obtain that
	\begin{equation}
	\label{eqn:for1401a}
	\begin{aligned}
	&C_o(\boldsymbol{\hat \eta}_i,\boldsymbol{\rho}^o_i,\boldsymbol{D}^o_i) - C_o(\boldsymbol{\hat \eta}_{i^*},\boldsymbol{\rho}^o_{i^*},\boldsymbol{D}^o_{i^*}) \leq \\
	&\pi_i \Big(G_o(\boldsymbol{\hat\eta}_i,\boldsymbol{D}^o_i) - G_o(\boldsymbol{\hat\eta}_{i^*},\boldsymbol{D}^o_{i^*})\Big),
	\end{aligned}
	\end{equation}
	\begin{equation}
	\label{eqn:for1401b}
	\begin{aligned}
	&C_o(\boldsymbol{\hat \eta}_i,\boldsymbol{\rho}^o_i,\boldsymbol{D}^o_i) - C_o(\boldsymbol{\hat\eta}_{i^*},\boldsymbol{\rho}^o_{i^*},\boldsymbol{D}^o_{i^*}) \geq \\  
	&\pi_{i^*} \Big(G_o(\boldsymbol{\hat \eta}_i,\boldsymbol{D}^o_i) - G_o(\boldsymbol{\hat \eta}_{i^*},\boldsymbol{D}^o_{i^*})\Big).
	\end{aligned}
	\end{equation}
	From~(\ref{eqn:for1401b}), we have $C_o(\boldsymbol{\hat \eta}_i,\boldsymbol{\rho}^o_i,\boldsymbol{D}^o_i) \geq C_o(\boldsymbol{\hat \eta}_{i^*},\boldsymbol{\rho}^o_{i^*},\boldsymbol{D}^o_{i^*})$ if $G_o(\boldsymbol{\hat \eta}_i,\boldsymbol{D}^o_i) \geq G_o(\boldsymbol{\hat \eta}_{i^*},\boldsymbol{D}^o_{i^*})$, and thus $\boldsymbol{\rho}_i^o \geq \boldsymbol{\rho}_{i^*}^o$ as the cost function~(\ref{eqn:for30a2}) follows a monotonic increasing function in $\boldsymbol{\rho}_i^o$. We can also observe from~(\ref{eqn:for1401a}) that $G_o(\boldsymbol{\hat \eta}_i,\boldsymbol{D}^o_i) \geq G_o(\boldsymbol{\hat \eta}_{i^*},\boldsymbol{D}^o_{i^*})$ if $C_o(\boldsymbol{\hat \eta}_i,\boldsymbol{\rho}^o_i,\boldsymbol{D}^o_i) \geq C_o(\boldsymbol{\hat \eta}_{i^*},\boldsymbol{\rho}^o_{i^*},\boldsymbol{D}^o_{i^*})$, and thus $\boldsymbol{D}_i^o \geq \boldsymbol{D}_{i^*}^o$.
	
	\section{Proof of Proposition 3}
	\label{appx:pro3}
	
	From Lemma~\ref{lemma1} and Proposition~\ref{prop1}, we obtain $\boldsymbol{D}_i^o \geq \boldsymbol{D}_{i^*}^o$ and $\boldsymbol{\rho}_i^o \geq \boldsymbol{\rho}_{i^*}^o$. If $\pi_i \geq \pi_{i^*}$, then
	\begin{align}
	&\pi_iG_o(\boldsymbol{\hat \eta}_i,\boldsymbol{D}^o_i) - C_o(\boldsymbol{\hat \eta}_i,\boldsymbol{\rho}^o_i,\boldsymbol{D}^o_i) \geq \nonumber\\
	&\pi_{i}G_o(\boldsymbol{\hat \eta}_{i^*},\boldsymbol{D}^o_{i^*}) - C_o(\boldsymbol{\hat \eta}_{i^*},\boldsymbol{\rho}^o_{i^*},\boldsymbol{D}^o_{i^*}) \geq \nonumber \\
	&\pi_{i^*}G_o(\boldsymbol{\hat \eta}_{i^*},\boldsymbol{D}^o_{i^*}) - C_o(\boldsymbol{\hat \eta}_{i^*},\boldsymbol{\rho}^o_{i^*},\boldsymbol{D}^o_{i^*}) \geq 0. \label{eqn:for1402}
	\end{align}
	Thus, $\pi_iG_o(\boldsymbol{\hat \eta}_i,\boldsymbol{D}^o_i) - C_o(\boldsymbol{\hat \eta}_i,\boldsymbol{\rho}^o_i,\boldsymbol{D}^o_i) \geq \pi_{i^*}G_o(\boldsymbol{\hat \eta}_{i^*},\boldsymbol{D}^o_{i^*}) - C_o(\boldsymbol{\hat \eta}_{i^*},\boldsymbol{\rho}^o_{i^*},\boldsymbol{D}^o_{i^*})$ when $\pi_i \geq \pi_{i^*}$.
	
	\section{Proof of Lemma 2}
	\label{appx:lemma2}
	
	We first specify downward ICs (DICs) to be the IC contraints between the MAP with type indices $i$ and $i^*$, $\forall i^* \in \{1, \ldots, i-1\}$, and upward ICs (UICs) to be the IC contraints between the MAP with type indices $i$ and $i^*$, $\forall i^* \in \{i+1, \ldots, I\}$, where $\pi_iG_o(\boldsymbol{\hat \eta}_i,\boldsymbol{D}^o_i) - C_o(\boldsymbol{\hat \eta}_i,\boldsymbol{\rho}^o_i,\boldsymbol{D}^o_i) \geq
	\pi_{i}G_o(\boldsymbol{\hat \eta}_{i^*},\boldsymbol{D}^o_{i^*}) - C_o(\boldsymbol{\hat \eta}_{i^*},\boldsymbol{\rho}^o_{i^*},\boldsymbol{D}^o_{i^*})$.
%	\begin{align}
%	&\pi_iG_o(\boldsymbol{\hat \eta}_i,\boldsymbol{D}^o_i) - C_o(\boldsymbol{\hat \eta}_i,\boldsymbol{\rho}^o_i,\boldsymbol{D}^o_i) \geq \nonumber\\
%	&\pi_{i}G_o(\boldsymbol{\hat \eta}_{i^*},\boldsymbol{D}^o_{i^*}) - C_o(\boldsymbol{\hat \eta}_{i^*},\boldsymbol{\rho}^o_{i^*},\boldsymbol{D}^o_{i^*}). \label{eqn:for7c1}
%	\end{align}
	We then show that DICs can be transformed into two consecutive types, i.e., local DICs (LDICs), which are the IC constraints between the MAP with type indices $i$ and $i-1$. Given that $\pi_{i-1} < \pi_i < \pi_{i+1}, i \in \{2, \ldots, I-1\}$, we obtain
	\begin{align}
	&\pi_{i+1}G_o(\boldsymbol{\hat \eta}_{i+1},\boldsymbol{D}^o_{i+1}) - C_o(\boldsymbol{\hat \eta}_{i+1},\boldsymbol{\rho}^o_{i+1},\boldsymbol{D}^o_{i+1}) \geq \nonumber\\
	&\pi_{i+1}G_o(\boldsymbol{\hat \eta}_{i},\boldsymbol{D}^o_{i}) - C_o(\boldsymbol{\hat \eta}_{i},\boldsymbol{\rho}^o_{i},\boldsymbol{D}^o_{i}), \label{eqn:for7c2}
	\end{align}
	\begin{align}
	&\pi_{i}G_o(\boldsymbol{\hat \eta}_{i},\boldsymbol{D}^o_{i}) - C_o(\boldsymbol{\hat \eta}_{i},\boldsymbol{\rho}^o_{i},\boldsymbol{D}^o_{i}) \geq \nonumber\\
	&\pi_{i}G_o(\boldsymbol{\hat \eta}_{i-1},\boldsymbol{D}^o_{i-1}) - C_o(\boldsymbol{\hat \eta}_{i-1},\boldsymbol{\rho}^o_{i-1},\boldsymbol{D}^o_{i-1}). \label{eqn:for7c3}
	\end{align}
	From Lemma~\ref{lemma1}, we show that if $\pi_i \geq \pi_{i^*}$ then $G_o(\boldsymbol{\hat \eta}_i,\boldsymbol{D}^o_i) \geq G_o(\boldsymbol{\hat \eta}_{i^*},\boldsymbol{D}^o_{i^*})$. Based on this condition, we have
	\begin{align}
	&\pi_{i+1}(G_o(\boldsymbol{\hat \eta}_{i},\boldsymbol{D}^o_{i})-G_o(\boldsymbol{\hat \eta}_{i-1},\boldsymbol{D}^o_{i-1})) \geq \nonumber\\ 
	&\pi_{i}(G_o(\boldsymbol{\hat \eta}_{i},\boldsymbol{D}^o_{i})-G_o(\boldsymbol{\hat \eta}_{i-1},\boldsymbol{D}^o_{i-1}))  \geq  \label{eqn:for7c4} \\ 
	&C_o(\boldsymbol{\hat \eta}_{i},\boldsymbol{\rho}^o_{i},\boldsymbol{D}^o_{i}) - C_o(\boldsymbol{\hat \eta}_{i-1},\boldsymbol{\rho}^o_{i-1},\boldsymbol{D}^o_{i-1}), \text{ and} \nonumber
	\end{align}
	\begin{align}
	&\pi_{i+1}G_o(\boldsymbol{\hat \eta}_{i+1},\boldsymbol{D}^o_{i+1}) - C_o(\boldsymbol{\hat \eta}_{i+1},\boldsymbol{\rho}^o_{i+1},\boldsymbol{D}^o_{i+1}) \geq \nonumber\\
	&\pi_{i+1}G_o(\boldsymbol{\hat \eta}_{i},\boldsymbol{D}^o_{i}) - C_o(\boldsymbol{\hat \eta}_{i},\boldsymbol{\rho}^o_{i},\boldsymbol{D}^o_{i}) \geq  \label{eqn:for7c5}\\
	&\pi_{i+1}G_o(\boldsymbol{\hat \eta}_{i-1},\boldsymbol{D}^o_{i-1}) - C_o(\boldsymbol{\hat \eta}_{i-1},\boldsymbol{\rho}^o_{i-1},\boldsymbol{D}^o_{i-1}). \nonumber
	\end{align}
	Since $\pi_{i+1}G_o(\boldsymbol{\hat \eta}_{i+1},\boldsymbol{D}^o_{i+1}) - C_o(\boldsymbol{\hat \eta}_{i+1},\boldsymbol{\rho}^o_{i+1},\boldsymbol{D}^o_{i+1}) \geq \pi_{i+1}G_o(\boldsymbol{\hat \eta}_{i-1},\boldsymbol{D}^o_{i-1}) - C_o(\boldsymbol{\hat \eta}_{i-1},\boldsymbol{\rho}^o_{i-1},\boldsymbol{D}^o_{i-1})$,
	we can extend it from the MAP with type index $i-1$ to $1$ to prove that all the DICs hold, i.e.,
	\begin{align}
	&\pi_{i+1}G_o(\boldsymbol{\hat \eta}_{i+1},\boldsymbol{D}^o_{i+1}) - C_o(\boldsymbol{\hat \eta}_{i+1},\boldsymbol{\rho}^o_{i+1},\boldsymbol{D}^o_{i+1}) \geq \nonumber\\
	&\pi_{i+1}G_o(\boldsymbol{\hat \eta}_{i-1},\boldsymbol{D}^o_{i-1}) - C_o(\boldsymbol{\hat \eta}_{i-1},\boldsymbol{\rho}^o_{i-1},\boldsymbol{D}^o_{i-1}) \geq \ldots \geq \label{eqn:for7c7}\\
	&\pi_{i+1}G_o(\boldsymbol{\hat \eta}_{1},\boldsymbol{D}^o_{1}) - C_o(\boldsymbol{\hat \eta}_{1},\boldsymbol{\rho}^o_{1},\boldsymbol{D}^o_{1}), \forall i \in \{1,\ldots,I-1\} \nonumber.
	\end{align}
	Therefore, we can summarize that all the DICs satisfy by utilizing the LDICs, and thus the DICs can be reduced. Similarly, we can show that all the UICs hold considering the local UICs (LUICs) and Lemma~\ref{lemma1}. This concludes the proof.
	
	\section{Proof of Proposition 4}
	\label{appx:pro4}
	
	From the constraints (\ref{eqn:for9a})-(\ref{eqn:for9c}) and (\ref{eqn:for9arev})-(\ref{eqn:for9crev}), there are  $I$ trainable encrypted dataset constraints, $N$ trainable local dataset constraints, $I \times N$ total dataset constraints, one IR constraint, and $2(I-1)$ IC constraints, where $N$ is a constant. Hence, we can solve $(\mathbf{P}_4)$ with computational complexity $O(I)$.
	
	\section{Proof of Theorem 1}
	\label{appx:theorem1}
	
	For simplicity, we define $\boldsymbol{D}^{(\tau)}_n = \Big(\boldsymbol{D}^{o,(\tau)}_n, \boldsymbol{D}^{l,(\tau)}_n\Big)$ and $\boldsymbol{\rho}^{(\tau)}_n = \Big(\boldsymbol{\rho}^{o,(\tau)}_n, \boldsymbol{\rho}^{l,(\tau)}_n\Big)$. Similarly, we denote  $\boldsymbol{D}^{(\tau)}_{-n} = \Big(\boldsymbol{D}^{o,(\tau)}_{-n}, \boldsymbol{D}^{l,(\tau)}_{-n}\Big)$ and $\boldsymbol{\rho}^{(\tau)}_{-n} = \Big(\boldsymbol{\rho}^{o,(\tau)}_{-n}, \boldsymbol{\rho}^{l,(\tau)}_{-n}\Big)$. Then, suppose that $\mathbb{Q}$ has $\Big(\boldsymbol{D}^{(\tau+1)}_{-n},\boldsymbol{\rho}^{(\tau+1)}_{-n}\Big)$ and $\Big(\boldsymbol{D}^{(\tau)}_{-n},\boldsymbol{\rho}^{(\tau)}_{-n}\Big)$ such that $\Big(\boldsymbol{D}^{(\tau+1)}_{-n},\boldsymbol{\rho}^{(\tau+1)}_{-n}\Big) > \Big(\boldsymbol{D}^{(\tau)}_{-n},\boldsymbol{\rho}^{(\tau)}_{-n}\Big)$.
	%\begin{equation}
	%\label{eqn:for16a}
	%\begin{aligned}
	%\Big(\boldsymbol{D}^{(\tau+1)}_{-n},\boldsymbol{\rho}^{(\tau+1)}_{-n}\Big) > \Big(\boldsymbol{D}^{(\tau)}_{-n},\boldsymbol{\rho}^{(\tau)}_{-n}\Big).
	%\end{aligned}
	%\end{equation}
	Using ${\hat \Theta}^{(\tau+1)}_n\Big(\boldsymbol{D}^{(\tau)}_{-n}, \boldsymbol{\rho}^{(\tau)}_{-n}\Big) = \max \phantom{1} { \Theta}^{(\tau+1)}_n\Big(\boldsymbol{D}^{(\tau)}_{-n}, \boldsymbol{\rho}^{(\tau)}_{-n}\Big)$, we can first prove that the selection ${\hat \Theta}_n$ in ${\Theta}_n$ is increasing, i.e., ${\hat \Theta}^{(\tau+2)}_n\Big(\boldsymbol{D}^{(\tau+1)}_{-n}, \boldsymbol{\rho}^{(\tau+1)}_{-n}\Big) \geq {\hat \Theta}^{(\tau+1)}_n\Big(\boldsymbol{D}^{(\tau)}_{-n}, \boldsymbol{\rho}^{(\tau)}_{-n}\Big)$.
	%\begin{equation}
	%\label{eqn:for16b}
	%\begin{aligned}
	%{\hat \Theta}^{(\tau+2)}_n\Big(\boldsymbol{D}^{(\tau+1)}_{-n}, \boldsymbol{\rho}^{(\tau+1)}_{-n}\Big) \geq {\hat \Theta}^{(\tau+1)}_n\Big(\boldsymbol{D}^{(\tau)}_{-n}, \boldsymbol{\rho}^{(\tau)}_{-n}\Big).
	%\end{aligned}
	%\end{equation}
	From $\Big(\boldsymbol{D}^{(\tau+1)}_{-n}, \boldsymbol{\rho}^{(\tau+1)}_{-n}\Big) > \Big(\boldsymbol{D}^{(\tau)}_{-n}, \boldsymbol{\rho}^{(\tau)}_{-n}\Big)$ and ${\hat \Theta}^{(\tau+2)}_n\Big(\boldsymbol{D}^{(\tau+1)}_{-n}, \boldsymbol{\rho}^{(\tau+1)}_{-n}\Big) < {\hat \Theta}^{(\tau+1)}_n\Big(\boldsymbol{D}^{(\tau)}_{-n}, \boldsymbol{\rho}^{(\tau)}_{-n}\Big)$ for a contradiction, we have 
	\begin{equation}
	{\fontsize{11.3}{11}\selectfont
		\label{eqn:for16c}
		\begin{aligned}
		&U_n\bigg({\hat \Theta}^{(\tau+1)}_n\Big(\boldsymbol{D}^{(\tau)}_{-n}, \boldsymbol{\rho}^{(\tau)}_{-n}\Big),\boldsymbol{D}^{(\tau+2)}_{-n}, \boldsymbol{\rho}^{(\tau+2)}_{-n}\bigg) +\\ &U_n\bigg({\hat \Theta}^{(\tau+2)}_n\Big(\boldsymbol{D}^{(\tau+1)}_{-n}, \boldsymbol{\rho}^{(\tau+1)}_{-n}\Big),\boldsymbol{D}^{(\tau+1)}_{-n}, \boldsymbol{\rho}^{(\tau+1)}_{-n}\bigg) \geq \\
		&U_n\bigg({\hat \Theta}^{(\tau+2)}_n\Big(\boldsymbol{D}^{(\tau+1)}_{-n}, \boldsymbol{\rho}^{(\tau+1)}_{-n}\Big),\boldsymbol{D}^{(\tau+2)}_{-n}, \boldsymbol{\rho}^{(\tau+2)}_{-n}\bigg) + \\ &U_n\bigg({\hat \Theta}^{(\tau+1)}_n\Big(\boldsymbol{D}^{(\tau)}_{-n}, \boldsymbol{\rho}^{(\tau)}_{-n}\Big),\boldsymbol{D}^{(\tau+1)}_{-n}, \boldsymbol{\rho}^{(\tau+1)}_{-n}\bigg).
		\end{aligned}}
	\end{equation}
	From the definition of $\Theta^{(\tau+1)}_n$, ${\hat \Theta}^{(\tau+1)}_n\Big(\boldsymbol{D}^{(\tau)}_{-n}, \boldsymbol{\rho}^{(\tau)}_{-n}\Big) \in { \Theta}^{(\tau+1)}_n\Big(\boldsymbol{D}^{(\tau)}_{-n}, \boldsymbol{\rho}^{(\tau)}_{-n}\Big)$ reveals that
	\begin{equation}
	\label{eqn:for16d}
	\begin{aligned}
	&U_n\bigg({\hat \Theta}^{(\tau+1)}_n\Big(\boldsymbol{D}^{(\tau)}_{-n}, \boldsymbol{\rho}^{(\tau)}_{-n}\Big),\boldsymbol{D}^{(\tau+1)}_{-n}, \boldsymbol{\rho}^{(\tau+1)}_{-n}\bigg) \geq \\ &U_n\bigg({\hat \Theta}^{(\tau+2)}_n\Big(\boldsymbol{D}^{(\tau+1)}_{-n}, \boldsymbol{\rho}^{(\tau+1)}_{-n}\Big),\boldsymbol{D}^{(\tau+1)}_{-n}, \boldsymbol{\rho}^{(\tau+1)}_{-n}\bigg).
	\end{aligned}
	\end{equation}
	Correspondingly, from~(\ref{eqn:for16c}) and~(\ref{eqn:for16d}), we obtain that	
	\begin{equation}
	\label{eqn:for16e}
	\begin{aligned}
	&U_n\bigg({\hat \Theta}^{(\tau+1)}_n\Big(\boldsymbol{D}^{(\tau)}_{-n}, \boldsymbol{\rho}^{(\tau)}_{-n}\Big),\boldsymbol{D}^{(\tau+2)}_{-n}, \boldsymbol{\rho}^{(\tau+2)}_{-n}\bigg) \geq \\ &U_n\bigg({\hat \Theta}^{(\tau+2)}_n\Big(\boldsymbol{D}^{(\tau+1)}_{-n}, \boldsymbol{\rho}^{(\tau+1)}_{-n}\Big),\boldsymbol{D}^{(\tau+2)}_{-n}, \boldsymbol{\rho}^{(\tau+2)}_{-n}\bigg),
	\end{aligned}
	\end{equation}
	and thus we have ${\hat \Theta}^{(\tau+1)}_n\Big(\boldsymbol{D}^{(\tau)}_{-n}, \boldsymbol{\rho}^{(\tau)}_{-n}\Big) \in { \Theta}^{(\tau+2)}_n\Big(\boldsymbol{D}^{(\tau+1)}_{-n}, \boldsymbol{\rho}^{(\tau+1)}_{-n}\Big)$. Alternatively, it can be observed that ${\hat \Theta}^{(\tau+1)}_n\Big(\boldsymbol{D}^{(\tau)}_{-n}, \boldsymbol{\rho}^{(\tau)}_{-n}\Big) \geq {\hat \Theta}^{(\tau+2)}_n\Big(\boldsymbol{D}^{(\tau+1)}_{-n}, \boldsymbol{\rho}^{(\tau+1)}_{-n}\Big)$. From the original definition of ${\hat \Theta}_n^{(\tau+1)}$, we can notice that ${\hat \Theta}^{(\tau+2)}_n\Big(\boldsymbol{D}^{(\tau+1)}_{-n}, \boldsymbol{\rho}^{(\tau+1)}_{-n}\Big) \geq {\hat \Theta}^{(\tau+1)}_n\Big(\boldsymbol{D}^{(\tau)}_{-n}, \boldsymbol{\rho}^{(\tau)}_{-n}\Big)$. As such, the latter condition contradicts with the former condition, and thus ${\hat \Theta}_n$ is an increasing function. 
	
	Utilizing the increasing function of ${\hat \Theta}_n$, the MAP can renew the contract of MU-$n$ at iteration $\tau + 1$, i.e., $\Big(\boldsymbol{D}^{(\tau+1)}_{n}, \boldsymbol{\rho}^{(\tau+1)}_{n}\Big) \in \Theta^{(\tau+1)}_n$, if the condition in~(\ref{eqn:for11c}) holds. For that, the use of $\sigma$ can be defined as the optimality tolerance to terminate the iterative algorithm. Specifically, when the condition in~(\ref{eqn:for11c}) is not satisfied at iteration $\tau = \tau^\dagger$ for all MU-$n$, $\forall n \in \mathcal{N}$, we obtain that $\Big(\boldsymbol{D}^{(\tau^\dagger+1)}_{n}, \boldsymbol{\rho}^{(\tau^\dagger+1)}_{n}\Big) = \Big(\boldsymbol{D}^{(\tau^\dagger)}_{n}, \boldsymbol{\rho}^{(\tau^\dagger)}_{n}\Big)$, $\forall n \in \mathcal{N}$. In other words, the iterative algorithm produces the same $U_n\Big(\boldsymbol{D}^{(\tau+1)}_{n}, \boldsymbol{\rho}^{(\tau+1)}_{n},\boldsymbol{D}^{(\tau)}_{-n}, \boldsymbol{\rho}^{(\tau)}_{-n}\Big), \forall n \in \mathcal{N}$, for the rest of $\tau > \tau^\dagger$. As a result, the iterative algorithm converges under $\sigma$.

	\section{Proof of Theorem 2}
	\label{appx:theorem2}
	
	We first demonstrate that an equilibrium exists by finding a fixed point in $\Theta$. Consider that $\mathbb{Q}^*$ contains $\Big(\boldsymbol{D}^o, \boldsymbol{\rho}^o,\boldsymbol{D}^l, \boldsymbol{\rho}^l\Big) \in \mathbb{Q}$.
	This $\mathbb{Q}^*$ is a non-empty contract space since ${\hat \Theta}\big(\min \mathbb{Q}\big) \geq \min \mathbb{Q}$, ${\hat \Theta} \in \Theta$, thereby ${\hat \Theta}\big(\max \mathbb{Q}^*\big) \geq \max \mathbb{Q}^*$.
	As ${\hat \Theta}$ is an increasing function, we have ${\hat \Theta}\Big({\hat \Theta}\big(\max \mathbb{Q}^*\big)\Big) \geq {\hat \Theta}\big(\max \mathbb{Q}^*\big)$, and thus ${\hat \Theta}\big(\max \mathbb{Q}^*\big) \in \mathbb{Q}^*$. Then, we also have ${\hat \Theta}\big(\max \mathbb{Q}^*\big) \leq \max \mathbb{Q}^*$ and produce ${\hat \Theta}\big(\max \mathbb{Q}^*\big) = \max \mathbb{Q}^*$. Therefore, $\max \mathbb{Q}^*$ is a fixed point of $\Theta$ which has $\Big(\boldsymbol{D}^o, \boldsymbol{\rho}^o,\boldsymbol{D}^l, \boldsymbol{\rho}^l\Big)$ and implies that the equilibrium exists.
	As the utilities of all MUs in $\mathcal{N}$ follow the increasing function until reaching the convergence, the MUs cannot further improve the utilities when their optimal contracts are obtained. This means that the iterative algorithm must converge to the equilibrium contract $\Big(\boldsymbol{D}^o, \boldsymbol{\rho}^o,\boldsymbol{D}^l, \boldsymbol{\rho}^l\Big)$, to ensure that there is no MU-$n$ that can unilaterally improve its utility, i.e., $U_n\Big(\boldsymbol{\hat D}_n^o, \boldsymbol{\hat \rho}_n^o,\boldsymbol{\hat D}_n^l, \boldsymbol{\hat \rho}_n^l,\boldsymbol{\hat D}_{-n}^o, \boldsymbol{\hat \rho}_{-n}^o,\boldsymbol{\hat D}_{-n}^l, \boldsymbol{\hat \rho}_{-n}^l\Big)$~\cite{Bernheim:1986,Fraysse:1993}.
	
	\section{Proof of Theorem 3}
	\label{appx:theorem3}
	
	%We adopt this proof from~\cite{Durand:2016}. 
	Suppose that there are $n^\star$ MUs in $\mathcal{N}$ that have finished the best response process at iteration $\tau+1$ without changing their current contracts. Let $u \in [0,1]$ denote the current normalized candidate utility. Considering $P$ to be the total number of contract strategies, the $n^\star$-th MU has $P-1$ new contract strategies in $\mathbb{Q}_{n^\star}$ when the MU calculates its ${\hat \Theta}^{(\tau+1)}_{n^\star}\Big(\boldsymbol{D}^{(\tau)}_{-n^\star}, \boldsymbol{\rho}^{(\tau)}_{-n^\star}\Big)$. As such, the probability that all the $P-1$ contract strategies cannot obtain higher normalized candidate utility than $u$ is $u^{P-1}$. As a result, the $(n^\star+1)$-th MU can then try to obtain its ${\hat \Theta}^{(\tau+1)}_{(n^\star+1)}\Big(\boldsymbol{D}^{(\tau)}_{-(n^\star+1)}, \boldsymbol{\rho}^{(\tau)}_{-(n^\star+1)}\Big)$. Meanwhile, the probability that one of the $P-1$ contract strategies is the best response can be expressed by $1-u^{P-1}$. In other words, the normalized candidate utility $u$ can be updated to a value larger than ${\hat U}_{n^\star}\Big(\boldsymbol{D}^{(\tau)}_{n^\star}, \boldsymbol{\rho}^{(\tau)}_{n^\star},\boldsymbol{D}^{(\tau)}_{-n^\star}, \boldsymbol{\rho}^{(\tau)}_{-n^\star}\Big) + \sigma$ (or ${\hat U}_{n^\star} + \sigma$ for the simplified form) with a probability $1 - ({\hat U}_{n^\star}+\sigma)^{P-1}$, where ${\hat U}_{n^\star} \in [0,1]$ is the normalized current utility of MU-$n^\star$. In this condition, the number of MUs in $\mathcal{N}$ which completes the best response can be set back to $1$.
	
	Let $V(u,n^\star)$ denote the number of steps in Algorithm~\ref{ICEA} when $n^\star$ MUs have completed the best response without modifying the contracts with $u$ prior to the convergence. Using the Markov chain, we have~\cite{Durand:2016}
	\begin{equation}
	\label{eqn:for17a}
	\begin{aligned}
	V(u,n^\star) &= u^{P-1}V(u,n^\star+1) \\&+\int_{u}^{1}(P-1){\tilde U}_{n^\star}^{P-2}(V({\tilde U}_{n^\star},1)+1)d{\tilde U}_{n^\star},
	\end{aligned}
	\end{equation}
	where ${\tilde U}_{n^\star} = {\hat U}_{n^\star}+\sigma$. Specifically, the first component of (\ref{eqn:for17a}) indicates that the $(n^\star+1)$-th MU does not change its contract, i.e., $V(u,n^\star) = V(u,n^\star+1)$, with probability $u^{P-1}$. Moreover, the second component of (\ref{eqn:for17a}) specifies that the $(n^\star+1)$-th MU obtains a new best response with normalized candidate utility ${\tilde U}_{n^\star}$ with the probability density $(P-1){\tilde U}_{n^\star}^{P-2}$. To this end, the Algorithm~\ref{ICEA} produces one additional step $V(u,n^\star) = V({\tilde U}_{n^\star},1)+1$.
	
	Using the Markov chain's boundary conditions $V(u,N) = 0$ and $V(1,n^\star) = 0$ for all $u$ and $n^\star$, respectively, we can obtain $V(u,1) = u^{P-1}V(u,2) + \Phi(u),\ldots,V(u,N-2) = u^{P-1}V(u,N-1) + \Phi(u),V(u,N-1) = \Phi(u)$.
	%\begin{equation}
	%\label{eqn:for17b}
	%\begin{aligned}
	%\left\{	\begin{array}{ll}
	%V(u,1) &= u^{P-1}V(u,2) + \Phi(u),\\
	%%V(u,2) &= u^{P-1}V(u,3) + \Phi(u),\\
	%\vdots &= \vdots \\
	%V(u,N-2) &= u^{P-1}V(u,N-1) + \Phi(u),\\
	%V(u,N-1) &= \Phi(u).
	%\end{array}	\right.\\
	%\end{aligned}
	%\end{equation}
	In this case, $\Phi(u) = \overset{1}{\underset{u}{\int}}(P-1){\tilde U}_{n^\star}^{P-2}(V({\tilde U}_{n^\star},1)+1)d{\tilde U}_{n^\star}$. Then, we can derive that $V(u,1) = \Phi(i)\Xi(u)$, where $\Xi(u) =  1 + u^{P-1} + \ldots + u^{(P-1)(N-2)}$,
	%\begin{equation}
	%\label{eqn:for17b2}
	%\begin{aligned}
	%\Xi(u) =  1 + u^{P-1} + \ldots + u^{(P-1)(N-2)},
	%\end{aligned}
	%\end{equation}
	and thus the differential equation of $V(u,1)$ with respect to $u$ can be written as
	\begin{equation}
	\label{eqn:for17c}
	\begin{aligned}
	&\frac{dV(u,1)}{du} + \Big((P-1)u^{P-2}\Xi(u) - \frac{1}{\Xi(u)}\frac{d\Xi(u)}{du}\Big)V(u,1) \\&= -(P-1)u^{P-2}\Xi(u).
	\end{aligned}
	\end{equation}
	Given $V(1,1) = 0$ and H(u) = $\int_{0}^{u}\Big((P-1){\tilde U}_{1}^{P-2}\Xi({\tilde U}_{1}) - \frac{1}{\Xi({\tilde U}_{1})}\frac{d\Xi({\tilde U}_{1})}{d{\tilde U}_{1}}\Big)d{\tilde U}_{1} 
	= \int_{0}^{u}(P-1){\tilde U}_{1}^{P-2}\Xi({\tilde U}_{1})d{\tilde U}_{1}  - \log(\Xi(u))$, 
	%\begin{equation}
	%\label{eqn:for17e}
	%\begin{aligned}
	%H(u) = \int_{0}^{u}\Big((P-1){\tilde U}_{1}^{P-2}\Xi({\tilde U}_{1}) - \frac{1}{\Xi({\tilde U}_{1})}\frac{d\Xi({\tilde U}_{1})}{d{\tilde U}_{1}}\Big)d{\tilde U}_{1} 
	%= \int_{0}^{u}(P-1){\tilde U}_{1}^{P-2}\Xi({\tilde U}_{1})d{\tilde U}_{1}  - \log(\Xi(u)),
	%\end{aligned}
	%\end{equation} 
	the $V(u,1)$ becomes $V(u,1) = e^{-H(u)}\int_{u}^{1}(P-1){\tilde U}_{1}^{P-2}\Xi({\tilde U}_{1})e^{H({\tilde U}_{1})}d{\tilde U}_{1}$.
	%	\begin{equation}
	%	\label{eqn:for17d}
	%	\begin{aligned}
	%	\vartheta(U,1) = e^{-\Omega(U)}\int_{U}^{1}(R-1){\tilde \mu}_{1}^{R-2}\Pi({\tilde \mu}_{1})e^{\Omega({\tilde \mu}_{1})}d{\tilde \mu}_{1}.
	%	\end{aligned}
	%	\end{equation}
	Since $V(u,1)$ is a decreasing function in $U$, the number of steps is upper-bounded by $V(0,1)$. Particularly, given that $\Xi(0) = 1$ and $H(0) = 0$, we have
	\begin{align}
	&V(0,1) = \int_{0}^{1}(P-1){\tilde U}_{1}^{P-2}e^{\sum_{n=0}^{N-2}\frac{{\tilde U}_{1}^{(P-1)(n+1)}}{n+1}}d{\tilde U}_{1} \nonumber \\
	%&= \int_{0}^{1}e^{\sum_{n=0}^{N-2}\frac{{\tilde U}_{1}^{(P-1)(n+1)}}{n+1}}d{\tilde U}_{1}^{(P-1)} \nonumber \\
	&= \int_{0}^{1 - \frac{1}{N}}e^{\sum_{n=1}^{N-1}\frac{{\tilde U}_{1}^{(P-1)n}}{n}}d{\tilde U}_{1}^{(P-1)} +\nonumber\\& \int_{1 - \frac{1}{N}}^{1}e^{\sum_{n=1}^{N-1}\frac{{\tilde U}_{1}^{(P-1)n}}{n}}d{\tilde U}_{1}^{(P-1)} \nonumber \\
	&\leq \int_{0}^{1 - \frac{1}{N}}e^{\sum_{n=1}^{\infty}\frac{{\tilde U}_{1}^{(P-1)n}}{n}}d{\tilde U}_{1}^{(P-1)} + \frac{1}{N}e^{\sum_{n=1}^{N-1}\frac{1}{n}} \nonumber \\ %= \int_{0}^{1 - \frac{1}{N}}\frac{d{\tilde U}_{1}^{(P-1)}}{1 - {\tilde U}_{1}^{(P-1)}} + e^\epsilon + O(1/N)
	&= \log(N) + e^\epsilon +O(1/N)), \label{eqn:for17f}
	\end{align}
	where $\epsilon \approx 0.5772$ is the Euler constant. %~\cite{Hazewinkel:2001}.
	From (\ref{eqn:for17f}), it is shown that the number of steps of the Algorithm~\ref{ICEA} has the polynomial complexity with upper bound $\log(N) + e^\epsilon +O(1/N)$~\cite{Saputra:2020,Durand:2016}.
	
	\section{Convergence Analysis of Algorithm~2}
	\label{appx:theorem4}
	
	We analyze the convergence of the proposed privacy-aware FL algorithm in the following. Specifically, we use the gap between the expected global loss after $\theta^\diamond$ rounds, i.e., $\mathbb{E}[\varphi(\mathbf{\tilde\Upsilon}^{(\theta^\diamond)})]$ and the final global loss, i.e., $\varphi^*(\mathbf{\tilde\Upsilon}^*)$. Given that the probability of an MU is selected by the MAP for the FL process at any round along with the MAP is $\frac{N}{M+1}$ (where $M+1$ is the total MUs along with the MAP), our FL algorithm is still unbiased towards full participation. Alternatively, the high-quality MU selection is random across the active MUs at each learning round~\cite{Ren:2021}. To this end, we can show that the global loss gap is upper bounded by the expected squared L2-norm encrypted global model gap, i.e., $\mathbb{E}[\|\mathbf{\tilde\Upsilon}^{(\theta^\diamond)} - \mathbf{\tilde\Upsilon}^*\|_2^2]$, and it eventually reaches zero as stated in Theorem~\ref{theorem1}.
	
	\begin{theorem}
		\label{theorem1}
		The proposed FL with privacy-awareness in Algorithm~\ref{DDL-PC} will converge to the minimum global loss $\varphi^*(\mathbf{\tilde\Upsilon}^*)$ according to the global loss gap condition $\big[\mathbb{E}[\varphi(\mathbf{\tilde\Upsilon}^{(\theta^\diamond)})] - \varphi^*(\mathbf{\tilde\Upsilon}^*)\big] \leq \frac{\delta}{2}\mathbb{E}[\|\mathbf{\tilde\Upsilon}^{(\theta^\diamond)} - \mathbf{\tilde\Upsilon}^*\|_2^2]$,
		%		\begin{align}
		%		&0 \leq \big[\mathbb{E}[\Psi(\mathbf{X}(T^*))] - \Psi^*(\mathbf{X}^*)\big] \leq \frac{\omega}{2}\mathbb{E}[\|\mathbf{X}(T^*) - \mathbf{X}^*\|_2^2], \label{eqn3k}
		%		\end{align}
		where $\delta$ is a positive constant. 
	\end{theorem}

	\begin{proof}
		
		We first denote the data distribution bias across MUs and the MAP as $\Lambda \triangleq \varphi^*(\mathbf{\tilde\Upsilon}^*) - \frac{1}{M+1}\sum_{m=1}^{M+1}\varphi_m(\mathbf{\tilde\Upsilon}^*) > 0$ for non-i.i.d data distribution and the largest encrypted global model norm for all learning rounds as $\Gamma = \max \|\mathbf{\tilde\Upsilon}^{(\theta)}\|$. Using the assumptions $1 - 3$ in~\cite{Amiri:2021}, $0 < \lambda^{(\theta)} \leq \min(1,\frac{1}{\delta_1\vartheta_{\emph{\mbox{th}}}}), \forall \theta$, where $\delta_1$ is a positive constant and $\vartheta_{\emph{\mbox{th}}}$ is pre-defined local iteration number, and $\lambda^{(\theta)} = \lambda_1^{(\theta)} = \ldots = \lambda_{N+1}^{(\theta)}$~\cite{Amiri:2021}, we can derive the expected squared L2-norm encrypted global model gap with an additional positive constant $\delta_2$ as follows:
		\begin{equation}
		\label{eqn18c}
		\begin{aligned}
		\mathbb{E}[\|\mathbf{\tilde\Upsilon}^{(\theta)} - \mathbf{\tilde\Upsilon}^*\|_2^2] &\leq \bigg(\prod_{\theta^*=0}^{\theta-1}f_1(\theta^*)\bigg)\|\mathbf{\tilde\Upsilon}^{(0)} - \mathbf{\tilde\Upsilon}^*\|_2^2 \\&+ \sum_{\theta^\dagger=0}^{\theta-1}f_2(\theta^\dagger)\prod_{\theta^*=\theta^\dagger+1}^{\theta-1}f_1(\theta^*),
		\end{aligned}
		\end{equation}
		where
		\begin{equation}
		\label{eqn18d}
		\begin{aligned}
		f_1(\theta^*) \triangleq 1 - \delta_2\lambda^{(\theta^*)}(\vartheta_{\emph{\mbox{th}}} - \lambda^{(\theta^*)}(\vartheta_{\emph{\mbox{th}}}-1)),
		\end{aligned}
		\end{equation}
		\begin{equation}
		\label{eqn18e}
		\begin{aligned}
		&f_2(\theta^*) \triangleq \frac{(M+1-N)\Big(\lambda^{(\theta^*)}\Big)^2\vartheta^2_{\emph{\mbox{th}}}\Gamma^2}{NM} + 2\lambda^{(\theta^*)}(\vartheta_{\emph{\mbox{th}}} - 1)\Lambda + \\
		&\Big(1+\delta_2\Big(1-\lambda^{(\theta^*)}\Big)\Big)\Big(\lambda^{(\theta^*)}\Big)^2\Gamma^2\frac{\vartheta_{\emph{\mbox{th}}}(\vartheta_{\emph{\mbox{th}}}-1)(2\vartheta_{\emph{\mbox{th}}}-1)}{6} + \\
		&\Big(\lambda^{(\theta^*)}\Big)^2(\vartheta^2_{\emph{\mbox{th}}} + \vartheta_{\emph{\mbox{th}}} - 1)\Gamma^2.
		\end{aligned}
		\end{equation}

	Using the $\delta$-smoothness in Assumption~1 of~\cite{Amiri:2021} for the global loss $\varphi(.)$, after $\theta^\diamond$ rounds, we have 
		\begin{align}
		&\mathbb{E}[\varphi(\mathbf{\tilde\Upsilon}^{(\theta^\diamond)})] - \varphi^*(\mathbf{\tilde\Upsilon}^*) \leq \frac{\delta}{2}\mathbb{E}[\|\mathbf{\tilde\Upsilon}^{(\theta^\diamond)} - \mathbf{\tilde\Upsilon}^*\|_2^2] \nonumber\\
		&\leq \frac{\delta}{2}\bigg(\prod_{\theta^*=0}^{\theta^\diamond-1}f_1(\theta^*)\bigg)\|\mathbf{\tilde\Upsilon}^{(0)} - \mathbf{\tilde\Upsilon}^*\|_2^2 \label{eqn3l}\\&+ \frac{\delta}{2}\sum_{\theta^\dagger=0}^{\theta^\diamond-1}f_2(\theta^\dagger)\prod_{\theta^*=\theta^\dagger+1}^{\theta^\diamond-1}f_1(\theta^*). \nonumber
		\end{align}
		Equation (\ref{eqn3l}) reveals that the global loss gap is upper bounded by $\frac{\delta}{2}\mathbb{E}[\|\mathbf{\tilde\Upsilon}^{(\theta^\diamond)} - \mathbf{\tilde\Upsilon}^*\|_2^2]$. Since the learning rate is decreasing, i.e., $\lim_{\theta\rightarrow\infty}\lambda^{(\theta)} = 0$, we show that $\lim_{\theta\diamond\rightarrow\infty}\big[\mathbb{E}[\varphi(\mathbf{\tilde\Upsilon}^{(\theta^\diamond)})] - \varphi^*(\mathbf{\tilde\Upsilon}^*)\big] = 0$, which indicates that the global loss gap eventually will reach zero.
		
	\begin{remark}
		The convex loss function assumption in~\cite{Amiri:2021} is used to show the theoretical convergence bound of the global loss gap condition of the proposed FL. The convergence analysis can be further extended for the non-convex loss function, e.g., CNN and BLSTM, by omitting the convex assumption, identifying the convergence rate using the average expected squared L2-norm of the encrypted global gradient, i.e., $\frac{1}{\theta^\diamond}\mathbb{E}\sum_{\theta=0}^{\theta^\diamond}\|\nabla\mathbf{\tilde\Upsilon}^{(\theta)}\|_2^2$, and bounding $\lim_{\theta^\diamond\rightarrow\infty}\mathbb{E}[\|\nabla\mathbf{\tilde\Upsilon}^{(\theta^\diamond)}\|_2^2]$. This will be discussed in the future work. %However, to the best of our knowledge, there currently exists no theoretical convergence analysis using a non-convex loss function for the proposed FL with limited high-quality MU selection (due to its intractibility). This will be investigated in the future work.
	\end{remark}
		
	\end{proof}

	%\begin{IEEEbiography}{Yuris Mulya Saputra}
	%Biography text here.
	%\end{IEEEbiography}
	%
	%\begin{IEEEbiography}{Dinh Thai Hoang}
	%Biography text here.
	%\end{IEEEbiography}
	%
	%\begin{IEEEbiography}{Diep N. Nguyen}
	%Biography text here.
	%\end{IEEEbiography}
%
%\begin{IEEEbiography}{Eryk Dutkiewicz}
%Biography text here.
%\end{IEEEbiography}

% that's all folks
\end{document}